\documentclass[a4paper,11pt]{article}

\usepackage{amsmath}
\usepackage{amsfonts}
\usepackage{amssymb}
\usepackage{amsthm}
\usepackage{caption}
\usepackage{subcaption}
\usepackage[english]{babel}
\usepackage{fontenc}
\usepackage{graphicx}
\usepackage{mathrsfs}
\usepackage{bbm}
\usepackage{array}
\usepackage{float}
\usepackage{enumerate}
\usepackage{xcolor}
\usepackage{hyperref}
\usepackage[labelformat = empty,position=top]{subcaption}
\usepackage{adjustbox}
\usepackage{float}
\usepackage{tikz}
\usepackage{tkz-graph}
\usetikzlibrary{shapes.geometric}
\usetikzlibrary{backgrounds}
\usepackage{paralist}
\usepackage[framemethod=TikZ]{mdframed}
\usepackage{authblk}
\usepackage[normalem]{ulem}

\usepackage{fullpage}
\usepackage{centernot}
\usepackage{natbib}
\usepackage{chapterbib}

\theoremstyle{plain}
\newtheorem{thm}{Theorem}
\newtheorem{lem}[thm]{Lemma}
\newtheorem{prop}[thm]{Proposition}
\newtheorem{cor}[thm]{Corollary}
\newtheorem{cond}{Condition}

\theoremstyle{remark}

\newcommand{\abs}[1]{\left| #1\right|}
\newcommand{\norm}[1]{\left\lVert #1 \right\rVert}
\newcommand{\ind}{\mathbbm{1}}
\newcommand{\pr}{\mathbb{P}}
\newcommand{\R}{\mathbb{R}}
\newcommand{\E}{\mathbb{E}}

\newcommand{\Cov}{\mathrm{Cov}}

\newcommand{\Var}{\mathrm{Var}}

\newcommand{\argmin}[1]{\underset{#1}{\operatorname{argmin}}\;}
\newcommand{\sgn}{\mathrm{sgn}}

\newcommand{\eqdist}{\stackrel{d}{=}}

\newcommand{\tr}{\mathrm{tr}}
\newcommand{\independent}{\mbox{${}\perp\mkern-11mu\perp{}$}}
\newcommand{\notindependent}{\mbox{${}\not\!\perp\mkern-11mu\perp{}$}}

\newcommand{\mb}{\mathbf}

\newcommand{\V}{V}

\newcommand{\pca}{\textit{pca}}
\newcommand{\rsvp}{\textit{rsvp}}
\newcommand{\srsvp}{\textit{rsvp-split}}
\newcommand{\subrsvp}{\textit{rsvp-sub}}

\title{  RSVP-graphs: Fast High-dimensional Covariance Matrix Estimation Under Latent Confounding   }

\author[1]{Rajen D. Shah\thanks{Supported by an EPSRC First Grant and the Alan Turing Institute under the EPSRC grant EP/N510129/1.}}
\author[2]{Benjamin Frot}
\author[2]{Gian-Andrea Thanei}
\author[2]{Nicolai Meinshausen}
\affil[1]{University of Cambridge}
\affil[2]{ETH Z\"urich}

\begin{document}
\maketitle
\begin{cbunit}
\begin{abstract}
We consider the problem of estimating a high-dimensional $p \times p$ covariance matrix $\Sigma$, given $n$ observations of confounded data with covariance $\Sigma + \Gamma \Gamma^T$, where $\Gamma$ is an unknown $p \times q$ matrix of latent factor loadings. We propose a simple and scalable estimator based on the projection on to the right singular vectors of the observed data matrix, which we call RSVP. Our theoretical analysis of this method reveals that in contrast to approaches based on removal of principal components, RSVP is able to cope well with settings where the smallest eigenvalue of $\Gamma^T \Gamma$ is relatively close to the largest eigenvalue of $\Sigma$, as well as when the eigenvalues of $\Gamma^T \Gamma$ are diverging fast.
RSVP does not require knowledge or estimation of the number of latent factors $q$, but only recovers $\Sigma$ up to an unknown positive scale factor. We argue this suffices in many applications, for example if an estimate of the correlation matrix is desired. We also show that by using subsampling, we can further improve the performance of the method.
We demonstrate the favourable performance of RSVP through simulation experiments and an analysis of gene expression datasets collated by the GTEX consortium.
\end{abstract}

\section{Introduction} \label{sec:intro}
Suppose a random vector $w \in \R^p$ follows a multivariate normal distribution with covariance matrix $\Sigma$,
\[
w \sim \mathcal{N}_p(\mu_w,\Sigma).
\]
Given $n$ i.i.d.\ copies of $w$ whose rows form a data matrix $W \in
\R^{n \times p}$ it is often of interest to estimate either $\Sigma$,
or certain quantities derived from this such as the precision matrix $\Omega := \Sigma^{-1}$ or collections of conditional independencies that may then be used to infer causal structure \citep{spirtes2000causation}.

Suppose now that we cannot observe $W$ directly, but we instead observe $n$ i.i.d.\ copies of a random vector $x$ which form the rows of $X \in \R^{n\times p}$; $x$ is related to $w$ through
\begin{equation} \label{eq:SEM1}
x = w  + \Gamma h.
\end{equation}
Here $h \in \R^q$ is a vector of unobserved latent random variables, and $\Gamma \in \R^{p \times q}$ a fixed matrix of loadings.
If we assume that $h$ is normally distributed, without loss of generality we may take $h \sim \mathcal{N}_q(\mu_h, I)$. We then have that the covariance $\Theta$ of the observed $x$ contains a
 contribution $\Gamma \Gamma^T$ from latent confounding and a
contribution $\Sigma$ from idiosyncratic noise:
\[
\Theta = \Cov(x) =  \Gamma \Gamma^T + \Sigma.
\]
If we simply ignore the  confounding, we will have the
covariance $\Theta$ as the target of inference instead of $\Sigma$, and the
two can be very different.

Applications where such confounding is
important in practice include the following.
\begin{enumerate}[(a)]
\item  Cell biology. The activities of proteins and mRna, for example,
  can be confounded by environmental factors. Two highly correlated
  protein activities are thus not necessarily close in a causal
  network \citep{Leek2007, Stegle2012}.
\item Financial assets. The returns of various stocks  will be
  confounded by some latent factors (such as general market movement
  or sector influences) without the covariance necessarily revealing
  anything about causal connections between companies \citep{menchero2010global}.
\item Confounding in biology and genetics can also occur due to
  technical malfunction and laboratory effects  \citep{Speed2013}.
\end{enumerate}
Thus in a number of settings, in order to infer meaningful connections between variables we would like to remove the effect
of confounding from the empirical covariance $\hat{\Theta}$ of $X$ and estimate $\Sigma$.

As well as the intrinsic ill-posedness of the problem of separating $\Sigma$ from a noisy observation of $\Sigma + \Gamma\Gamma^T$ with $\Gamma$ unknown, a further challenge in the applications above and many others is that the dimension $p$ may be very large indeed, on the order of thousands or more. This high-dimensionality brings computational difficulties that must be addressed by any practical procedure.


In order for $\Sigma$ to be identifiable, appropriate assumptions on both $\Sigma$ and $\Gamma$ must be made. One natural assumption is that the minimum eigenvalue $\gamma_l$ of $\Gamma^T \Gamma$ is larger than the largest eigenvalue $\sigma_u$ of $\Sigma$.
In this setting, a popular strategy to deal with unwanted confounding is removal of top principal components from $\hat{\Theta}$. This has been proposed in \citet{Speed2013, fan2013large}.
The latter work, a JRSSB discussion paper, shows that when $\sigma_u$ is bounded and $\gamma_l = O(p)$, so the gap between the quantities is large, $\Sigma$ may be recovered consistently.
In this case the top $q$ eigenvalues of $\hat{\Theta}$ will be well separated from the rest, and so exactly $q$ principal components can be removed from $\hat{\Theta}$: this is important as removing too many or too few principal components can result in a poor estimate.

However, as several discussants of \citet{fan2013large} pointed out, in many settings empirical covariances do not display well-separated eigenvalues even when latent factors are known to be present.
When the gap between $\sigma_u$ and $\gamma_l$ is not large enough, the top $q$ eigenvalues can be close to the bulk, making estimation of $q$ challenging and potentially impossible \citep{barigozzi2018consistent}. Furthermore the top principal components (PCs) of the empirical covariance can be far away from those of $\Theta$ \citep{donoho2018}, so even if $q$ were known, the PC-removal approach would not work well.

In this paper, we propose a simple approach to estimating $\Sigma$ that is able to cope with settings where the gap between $\gamma_l$ and $\sigma_u$ may range from large and $O(p)$ to potentially small. In order to achieve this ambitious objective, the method sacrifices estimation of the scale of $\Sigma$: we only recover $\Sigma$ up to an unknown positive scalar factor.
The loss of scale however is inconsequential when the ultimate goal is rather to estimate the correlation matrix $\tilde{\Sigma}$, or to locate the top $s$ largest entries in $\Sigma$ for a pre-specified $s$, in order to build a network. In fact, we show that the scale-free nature of our estimator gives it an in-built robustness in that if the rows of $X$ have elliptical distributions, its distribution is precisely the same as if the data were Gaussian (see Proposition~\ref{prop:RSVP_dist}).

Let $V \in \R^{p \times (n-1)}$ be the matrix of right singular vectors with nonzero singular values of a column centred version of $X$. Our estimator is based on $\hat{\Sigma}_{\rsvp} := VV^T$; we call this right singular vector projection (RSVP).
The PC-removal estimate is proportional to $VH^2 V^T$ where $H$ is a diagonal matrix of singular values of the centred $X$ with the first $q$ entries set to $0$ (when $q$ is known). Thus RSVP may be seen as a highly regularised version of PC-removal, where the random $H$ is set to the identity matrix to reduce its variance. In fact, we show that each entry of $\hat{\Sigma}_{\rsvp}$ concentrates around its expectation at the same rate as the empirical covariance matrix after rescaling, even in settings where $q$ is allowed to grow at almost the same rate as $n$ (see Theorem~\ref{thm:conc_Pi}).

Despite the aggressive regularisation, it turns out the bias is dominated by the variance provided that $p \gg n$ so $n/p$ is small. As a consequence, we can show that with high probability,
\[
\inf_{\kappa > 0} \max_{j,k} | \Sigma_{jk} - \kappa \hat{\Sigma}_{\rsvp,jk}| \leq c \sqrt{\frac{\log(p)}{n}}
\]
for some constant $c > 0$,
even in certain settings when $\gamma_l$ is only larger than
$\sigma_u$ by a constant factor, and the latter is bounded.
In fact, we show that the statistical properties of $\hat{\Sigma}_{\rsvp}$ are such that when used as input to several standard procedures for conditional independence graph estimation or causal discovery procedures, the performances of the resulting estimates are, in many settings, identical to those attained when working with the unconfounded data, up to constant factors.

One requirement for $\hat{\Sigma}_{\rsvp}$ to work well is that $p \gg n$. For settings where $n$ is large, we can circumvent this condition using a subsampling strategy. 
We 
show that, surprisingly, by computing our estimator on subsamples of
the data and averaging \citep{breiman1996bagging}, the bias may be
reduced, and the variance only inflated by a factor of $\sqrt{\log(p)}$. Subsampling
with a very small number of samples in each subsample  is both statistically and computationally attractive, and is the approach we would recommend in settings where we do not have $p \gg n$.


\subsection{Related work}

There is large body of work on high-dimensional covariance and precision matrix estimation: see for example the recent review paper \citep{cai2016} and references therein. Much of the work on the specific setting with latent confounding has focussed on estimation of the precision matrix $\Omega$, which is assumed to be sparse. The presence of the latent confounding causes the overall precision matrix of $x$ to be a sum of $\Omega$ and a low rank component.
One approach to sparse precision matrix estimation in the absence of confounding is the graphical Lasso \citep{yuan05model,Yuan2010,Friedman2008}.
Building on this and work on sparse--dense matrix decompositions in the noiseless setting \citep{CandesEtAl11, ChandrasekaranEtAl11}, the work of \citet{chandrasekaran2012latent} formulates  a convex objective involving nuclear norm penalisation for Gaussian graphical model estimation with latent confounders. The work of \citet{FrotEtAl19} uses this as a stepping-stone for causal structure learning and causal effect estimation in low-dimensional settings. A challenge for nuclear norm penalisation and related approaches is that although the objective is convex, optimising it is nevertheless a computationally intensive task that does not scale to very large dimensions.

A second approach to precision matrix estimation exploits the fact that coefficients from regressions of each variable on all others, known as nodewise regressions, match the entries of the precision matrix up to scale factors \citep{lauritzen96graphical, meinshausen04consistent}.
  Adjusting
  for confounding can be built into a nodewise regression procedure, for example by using the Lava method of \citet{chernozhukov2017lava}  which employs a sparse--dense
  decomposition of the regression coefficients; the sparse part of the coefficients can then be retained as the dense part is generally due to
  confounding. This regression may be formulated as a Lasso regression
  with a transformed response and a particular preconditioned design
  matrix, see also \citet{rohe2014note} for an earlier equivalent proposal.
  \citet{Cevid18spectral} studies the theoretical properties of the Lava approach as well as more general forms of preconditioning including the Puffer transform proposed in \cite{jia2015preconditioning} and further investigated in \citet{Wang15}. This, in analogy with RSVP, modifies the design matrix by replacing non-zero singular values with a constant.
We also note that the Asymptotic Normal Thresholding procedure of \citet{ren2015asymptotic}, which employs nodewise regressions in a different fashion, is robust to weak confounding.

There has been comparatively less work on covariance matrix estimation in the presence of confounding, though, as we discuss and make use of in this work, an estimated covariance can be used as a starting point for conditional independence graph estimation or causal discovery. In addition to the work of \citet{fan2013large} and \citet{Speed2013} mentioned earlier, \citet{Fan2018-jg} proposes a PC-removal approach that can be applied to heavy-tailed data that follows an elliptical distribution.

\subsection{Organisation of the paper}
The rest of the paper is organised as follows. In
Section~\ref{sec:RSVP} we first discuss asymptotic identifiability of
$\Sigma$ and then introduce our RSVP estimator $\hat{\Sigma}_{\rsvp}$ and
versions involving subsampling. We present
theoretical properties of $\hat{\Sigma}_{\rsvp}$ and RSVP with sample-splitting in
Section~\ref{sec:theory}. In Section~\ref{sec:precision} we present results on the use of RSVP estimators as input to methods for conditional independence graph estimation and for causal discovery via the PC algorithm.
Numerical experiments are contained in Section~\ref{sec:experiments} and we conclude with a discussion in Section~\ref{sec:discuss}. The supplementary material for this paper contains all proofs and further results concerning the GTEX data analyses presented in Section~\ref{sec:GTEX}.

\subsection{Notation} \label{sec:notation}
We write $a \lesssim b$ as shorthand for `there exists constant $c > 0$ such that $a \leq cb$'. This constant may be a universal constant, or a function of quantities that have been designated as constants in our assumptions. If $a \lesssim b$ and $b \lesssim a$, we may write $a \asymp b$. For a matrix $A \in \R^{d \times m}$, $\|A\|$ will denote the operator norm, and $\|A\|_{\infty} = \max_{i=1,\ldots,d,j=1,\ldots,m} |A_{ij}|$.

When $d=m$ so $A$ is square, we will write $\lambda_{\max}(A)$ and $\lambda_{\min}(A)$ for the maximum and minimum eigenvalues of $A$ respectively. Further, given sets $I, J \subseteq D := \{1,\ldots,d\}$, we will denote by $A_{I,J}$ the $|I| \times |J|$ submatrix of $A$ formed from those rows and columns of $A$ indexed by $I$ and $J$ respectively. Such matrix subsetting operations will always be considered to have been performed first so that for example when $A_{I,I}$ is invertible, $A_{I,I}^{-1} \equiv (A_{I, I})^{-1}$. For $j \in D$, $j$ or $-j$ used in place of the subscripts $I$ or $J$ above will represent $\{j\}$ and $D \setminus \{j\}$ respectively, so $A_{j,-j}$ for  example is the $1 \times (d-1)$ matrix formed from the $j$th row of $A$ with its $j$th entry removed.

In analogy with the matrix subsetting notation set out above, we will write for a vector $v \in \R^d$, $v_I$ for the subvector formed from the components of $v$ indexed by $I$. Also for $j,k \in D$, $v_{-j}$ and $v_{-jk}$ will be subvectors of $v$ with $j$th and both the $j$th and $k$th components removed, respectively.
 We denote by $e_j$ the $j$th standard basis vector; the dimension of this will be clear from the context.

\section{RSVP: right singular vector projection} \label{sec:RSVP}
Let us assume that the observed data matrix $X \in \R^{(n+1) \times p}$ has rows given by $n+1$ independent realisations of a $\mathcal{N}_p(\mu, \Theta)$ random vector (we will later relax the Gaussian assumption, see Proposition~\ref{prop:RSVP_dist}). The $n+1$ rather than $n$ is for mathematical convenience: the column centred version $\tilde{X} := \Pi X$ of $X$ effectively contains $n$ observations. Here $\Pi = I - (n+1)^{-1} \mb 1 \mb 1^T$ where $\mb 1$ is an $(n+1)$-vector of 1's.
Our goal is to construct an estimate of $\Sigma$ based on this data where $\Sigma + \Gamma \Gamma^T = \Theta$ and both $\Gamma \in \R^{p \times q}$ and $q$ are unknown. We are interested in the case $p \gg n$ and will assume $p > cn$ for some $c > 1$, unless specified otherwise.

In what follows we first study the identifiability of $\Sigma$ in the model above. In Section~\ref{sec:spectrum} we discuss a general approach for estimating $\Sigma$ based on transforming the spectrum of the covariance matrix, which includes PC-removal and our RSVP method presented in Section~\ref{sec:RSVP2} as special cases. Finally we introduce a sample-splitting version of RSVP in Section~\ref{sec:subagged_RSVP}.

\subsection{Asymptotic identifiability} \label{sec:identify}
Let us first consider an artificial setting where $\Theta$ itself is directly observed. Even in this noiseless setting, certain conditions must be placed on $\Gamma$ and $\Sigma$ in order for $\Sigma$ to be recoverable given $\Theta$. Define
\begin{gather*}
\lambda_{\min}(\Gamma^T \Gamma) := \gamma_l, \qquad \lambda_{\max}(\Gamma^T \Gamma) := \gamma_u,\\
\lambda_{\min}(\Sigma) := \sigma_l, \qquad \lambda_{\max}(\Sigma) := \sigma_u.
\end{gather*}
If $\gamma_l$ is large compared to $\sigma_u$, we might hope that the top $q$ eigenvectors of $\Theta$ will span most of the column space of $\Gamma$. Therefore removing these from $\Theta$ should yield a matrix that is close to $\Sigma$. Proposition~\ref{prop:pop} below, based in part on an application of the Davis--Kahan $\sin(\theta)$ theorem~\citep{davis1970rotation}, formalises this intuition.

Let $\Theta$ have eigendecomposition $PD^2 P^T$ where the diagonal matrix $D$ has $D_{11} \geq D_{22} \geq \cdots \geq D_{pp}$.
Also define for $\ell \in \{1,\ldots,p\}$, function $H_{\ell}$ taking as argument a square matrix, and outputting a matrix of the same dimension, by
\begin{equation*}
(H_{\ell}(E))_{jk} = \begin{cases}
0 \qquad &\text{if } j,k \leq \ell \\
E_{jk} \qquad &\text{otherwise.}
\end{cases}
\end{equation*}
Thus the top left $\ell \times \ell$ submatrix of $H_{\ell}(E)$ is a matrix of $0$'s.
Define $\Pi_{\Gamma} := \Gamma(\Gamma^T\Gamma)^{-1} \Gamma^T$,
\begin{align*}
\rho_1 := \|\Pi_{\Gamma} \Sigma\| \qquad \text{and} \qquad \rho_2 := \max_j \|\Pi_{\Gamma} e_j\|_2.
\end{align*}
\begin{prop} \label{prop:pop}
Suppose $\sigma_l$ is bounded away from $0$ and $\gamma_l > c \sigma_u$ for a constant $c > 1$. Then
\begin{equation} \label{eq:pop}
\|P H_q(D^2) P^T- \Sigma\|_\infty \lesssim \rho_1 \rho_2 + \gamma_u \rho_1^2 / \gamma_l^2.
\end{equation}
\end{prop}
In order that removal of $q$ principal components yields a matrix close to $\Sigma$ at the population level, we require $\rho_2$ to be small; this essentially requires that the column space of $\Gamma$ is not too closely aligned with any of the standard basis vectors.
We always have the bound $\rho_1 \leq \sigma_u$. However, in the setting where $\Gamma$ is entirely uninformative about $\Sigma$, one might expect that $\rho_1$ may be smaller. Specifically, if we imagine that nature has chosen the column space of $\Gamma$ uniformly at random,
 we will have
with high probability that
\begin{align} \label{eq:rho_bd}
\rho_1^2 \lesssim \frac{1}{p}\{\tr(\Sigma^2) + \sqrt{q \tr(\Sigma^4)}\} \qquad \text{and}\qquad \rho_2^2 \lesssim \frac{q}{p} \left\{ 1 + \max\left(\frac{\log(p)}{q},\, \sqrt{\frac{\log(p)}{q}}  \right) \right\}.
\end{align}
See Section~\ref{sec:rho_bd} in the supplementary material for a derivation. Asymptotic identifiability results related to Proposition~\ref{prop:pop} are given in \citet{fan2013large,Fan2018-jg} when $\sigma_u$ and $q$ are bounded, and both $\gamma_l$ and $\gamma_u$ are $O(p)$. In these settings it is straightforward to show that $\rho_2 \lesssim p^{-1/2}$, in which case the right-hand side of \eqref{eq:pop} may be replaced by $p^{-1/2}$.

\subsection{Spectral transformations} \label{sec:spectrum}
We now return to the original noisy version of the problem.
The empirical covariance matrix $\hat{\Theta} = \tilde{X}^T\tilde{X} / n$ has expectation $\Theta = PD^2 P^T$, so we would ideally like to modify $\hat{\Theta}$ such that the eigenstructure of its expectation more closely resembles $P H_q(D^2) P$. Let us therefore consider the following family of estimators that involve transforming the spectrum of $\hat{\Theta}$.

Note that as $\tilde{X}$ has been centred and $p > n$, the rank of $\tilde{X}$ is $n$. Let the SVD of $\tilde{X}$ be given by $\tilde{X} = U \Lambda V^T$ where $\Lambda \in \R^{n \times n}$ is diagonal, and $U \in \R^{(n+1) \times n}$ and $V \in \R^{p \times n}$ each have orthonormal columns. Define
\begin{equation} \label{eq:Sigma_H}
\hat{\Sigma}_H = \frac{1}{n}V H(\Lambda^2) V^T
\end{equation}
where function $H$ here outputs $n \times n$ diagonal matrices. For such estimators, we have the following property.

\begin{prop} \label{prop:spec_trans}
We have that $\E \hat{\Sigma}_H = P C_H^2 P^T$ where $C_H$ is diagonal.
\end{prop}
The fact that the eigenvectors of $\E \hat{\Sigma}_H$ coincide with those of $\Theta$ suggests we should pick function $H$ such that $C^2_H$ is close to $H_q(D^2)$.
A natural choice is a simple PCA-based adjustment \citep{fan2013large, Speed2013, Fan2018-jg} of the form
\[
\hat{\Sigma}_{\pca} (\ell)   := \hat{\Sigma}_{H_{\ell}} =  n^{-1} VH_{\ell}(\Lambda^2)V^T.
\]
The resulting PC-removal estimator can be further thresholded as in
\citet{bickel2008covariance,fan2013large}, though if our aim is to recover
the locations of the largest entries of the covariance, this additional thresholding step is
without consequence.
The choice of the number $\ell$ of principal components to remove is rather critical to the method, but can be
challenging. Even if we had knowledge about the dimensionality $q$ of the
latent confounders, the optimal choice would depend on the relative
magnitude of the eigenvalues of $\Gamma^T \Gamma$ in relation to the
eigenvalues of $\Sigma$. In the absence of this knowledge, one might
resort to cross-validation schemes. Since the target of inference is
the unobserved idiosyncratic part $\Sigma$ of the covariance, it is
not obvious how such a cross-validation can be set up in a meaningful
way. Information criteria may be used as in \citet{fan2013large}, but these rely on $\gamma_l / \sigma_u \gtrsim p$.

\subsection{RSVP} \label{sec:RSVP2}
One reason that the PC-removal approach can struggle in settings where the separation between $\gamma_l$ and $\gamma_u$ is relatively small is that the top $q$ eigenvectors of $\hat{\Theta}$ need not span the column space of $\Gamma$ well, and in general will have high variability. Thus whilst $\hat{\Theta} = n^{-1} V\Lambda^2 V^T$ concentrates well around its expectation $\Theta$ in $\ell_\infty$-norm, an approach that involves manipulating the contributions of individual singular vectors in $V$ to the overall estimator, is likely to have high variance.
This suggests some form of regularisation may be helpful.

Taking the function $H$ as one which always returns $n$ times the identity matrix results in the simple estimator
\begin{align*}
\hat{\Sigma}_{\rsvp}  := \V  \V^T.
\end{align*}
Note this is invariant to permutations of the columns of $V$, and so is less dependent on properties of individual eigenvectors. As a consequence of the regularisation, we have lost the scaling of the original covariance: the estimator is invariant to multiplying $X$ from the left by any invertible $n \times n$ matrix. Thus we can only hope to recover $\Sigma$ up to a constant scale factor. This suffices for our purposes, and we argue this gives the estimator a certain robustness in that it is insensitive to particular pre-transformations of the data such as scaling of the rows of $X$. In fact $\hat{\Sigma}_{\rsvp}$ is more generally robust, see Proposition~\ref{prop:RSVP_dist} below. The computation time is dominated by the matrix multiplication of $V$ and $V^T$ which is $O(np^2)$; thus the computational complexity is the same as that for computing the empirical covariance.

In a regression context, an analogous approach for
preconditioning the design matrix has been explored in
\citet{jia2015preconditioning, Wang15}. The
Lava estimator \citep{chernozhukov2017lava} employs a similar preconditioning strategy
but, instead of setting all non-zero singular values of the design matrix to 1, the
singular values $d_i$ are transformed implicitly as
$\{d_i^2/(1+c d_i^2 )\}^{1/2}$, where the constant $c$ depends on a tuning parameter
and the sample size.

It may seem as if all information regarding the eigenvalues of $\Sigma$ has been lost in the regularisation as $\Lambda$ does not play a role in the estimator. However, we show in Section~\ref{sec:theory} that in certain high-dimensional settings, we can even estimate $\Sigma$ in $\ell_\infty$-norm at the same rate as the empirical covariance matrix in the absence of confounding, though only up to an unknown scale factor. Intuitively, the reason is that when $p \gg n$, with the exception of certain large eigenvalues in $\Lambda$ due to large eigenvalues in $\Gamma^T\Gamma$, the rest of the eigenvalues are essentially noise and bear no resemblance to the eigenvalues of $\Sigma$. This peculiar blessing of high-dimensionality is a phenomenon that fails when $p$ is of the same order as $n$, for example. It is however possible to subsample the data, and average over estimates computed on the samples, in order to mimic the high-dimensional setting. We discuss this below.

\subsection{Subsampling RSVP} \label{sec:subagged_RSVP}
Given $m \in \{1,\ldots,n\}$, let $\V^{(b)}$ be the matrix of right singular vectors of a random sample of $m$ rows of $\tilde{X}$. We define the subsampling RSVP estimator as
\begin{align*}
 \hat{\Sigma}_{\subrsvp} := \frac 1 B \sum_{b=1}^B
    \V^{(b)} (\V^{(b)})^T.
\end{align*}
The sample-splitting RSVP estimator $\hat{\Sigma}_{\srsvp}$ is defined similarly but where the sets of indices of the sampled rows are disjoint, and so  $B=\lceil (n+1)/m \rceil$. In practice, the subsampling estimator is preferable as the additional sampling can help to reduce the variance of the estimator. Our main reason for introducing the sample splitting version is that it is simpler to understand its theoretical properties (see Theorem~\ref{thm:Sigma_rec_bag}); however sample splitting still performs well empirically as we demonstrate in Section~\ref{sec:experiments}.

Both estimators are trivially parallelisable: the SVD computations for each subsample can be performed simultaneously, and then added at the end. If $B$ machines were available for the computations, the overall parallel computation time would be $O(mp^2)$ provided $\log(B) \lesssim m$.

\subsection{Example}

\begin{figure}
\begin{center}
\includegraphics[width=0.98\textwidth]{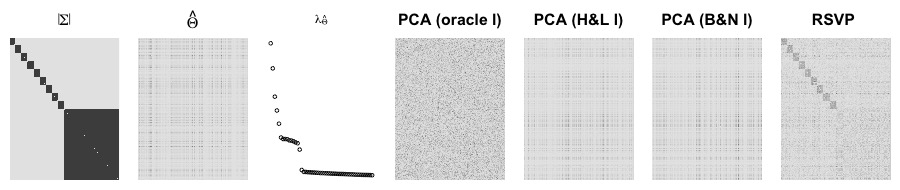}
\caption{An example with $p=1000$ variables, sample size $n=500$ and
  $q=20$ latent confounders with strength $\nu=0.5$, as described in detail in Section~\ref{sec:experiments}.  From left to right: (i) the absolute values of the
  idiosyncratic covariance matrix $\Sigma$, showing block structure
  with different block sizes; (ii) the empirical covariance matrix $\hat{\Theta}$;
  (iii) the eigenvalues of the $\hat{\theta}$ on a log-scale;
  (iv)-(vi) the absolute values of PC-removal estimator
  $\hat{\Sigma}(\ell)$, where $\ell$ is chosen first as the oracle
  value $\ell=q$ and next as the two empirical estimators of $q$ suggested in
  the POET \citep{fan2013large}; (vii) the proposed RSVP estimator with a subsample
  size of $m=20$. RSVP manages to recover the smaller blocks, while
  the PC-removal methods seemingly fail to recover any structure in the
  covariance matrix. \label{fig:example1}}
\end{center}
\end{figure}

\begin{figure}
\begin{center}
\includegraphics[width=0.95\textwidth]{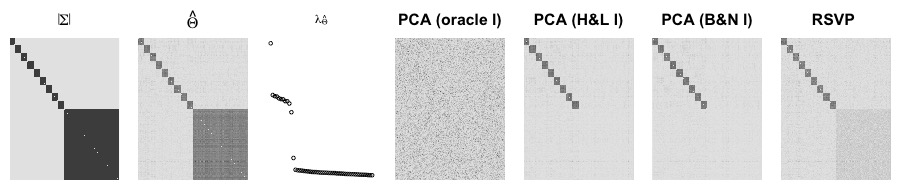}
\caption{The same setting as in Fig \ref{fig:example1} but for sample size
  $n=1000$ and very weak latent confounding ($\nu=0.01$). The large
  block is now even visible in the empirical covariance matrix. The
  PC-removal based methods fail to recover the structure of the large block as they
  all remove at least one principal component.  \label{fig:example2} }
\end{center}
\end{figure}

Figure~\ref{fig:example1} shows an example of the proposed sample-splitting
RSVP estimator, compared to the ground truth and
PC-removal. The latent confounding is so strong that the
empirical covariance shows very little visual indication of the block
structure of the idiosyncratic covariance. Likewise, PC-removal
fails to recover the structure, whether we use an oracle for
determining the number of factors to remove or  estimate the optimal
number of factors. RSVP in contrast recovers the smaller blocks. It is
shown here for  $m=20$ number of samples in each subsample (default) but results do not
change appreciably  when choosing a different subsample size. When
reducing the strength of the latent confounders, the empirical
covariance shows the correct underlying structure visually but all
PC-removal methods fail to recover the largest block of variables as
even just removing the first principal component removes the large block.

\section{Theoretical properties} \label{sec:theory}
In this section, we present some theoretical properties of $\hat{\Sigma}_{\rsvp}$ and $\hat{\Sigma}_{\srsvp}$. We first explain how $\hat{\Sigma}_{\rsvp}$ has low variance, and then argue that its bias is also well-controlled in the high-dimensional setting. We then discuss the consequences for $\hat{\Sigma}_{\srsvp}$. In the following we will consider an asymptotic regime We will assume Condition~\ref{cond} below in several of the results to follow.
\begin{cond}
\label{cond}
There exists constant $0 < c_1 < c_2$ such that $c_1 < \sigma_l \leq \max_j \Var(x_j) < c_2$. There exists constant $c_3 > 1$ such that $ \gamma_l > c_3 \sigma_u$ and $p > c_3 n$. Furthermore $\log(p) = o(n)$.
\end{cond}
\begin{thm} \label{thm:conc_Pi}
Assume Condition~\ref{cond} and that $\sigma_u \lesssim p / (n \log p)$ and $q \lesssim n / \log(p)$. Then there exists constant $c > 0$ such that with probability at least $1-c/p$,
\[
\frac{p}{n} \|\hat{\Sigma}_{\rsvp} - \E \hat{\Sigma}_{\rsvp} \|_{\infty} \lesssim \sqrt{\frac{\log(p)}{n}}.
\]
\end{thm}
We show in Theorem~\ref{thm:exp_Pi} that the entries in $\E \hat{\Sigma}_{\rsvp}$ are of the order $n/p$, so the result shows that the rate at which $\hat{\Sigma}_{\rsvp}$ concentrates is equivalent to that enjoyed by the empirical covariance matrix in the absence of confounding. The proof, given in Section~\ref{sec:Gauss_conc} in the supplementary material, is based on a variant of the classical concentration inequality for a  Lipschitz function $f:\R^d \to \R$ of i.i.d.\ Gaussian random variables $\zeta \sim \mathcal{N}_d(0, I)$, which may be of independent interest. Whereas the original result guarantees fast concentration when $\sup_{ v \in \R^d} \| \nabla f(v)\|_2$ is small, our new result (Theorem~\ref{thm:grad_conc}) only requires a high probability bound on $\| \nabla f(\zeta)\|_2$, and a potentially loose bound on $\E \| \nabla f(\zeta)\|_2^2$.
See also Lemma~1.3 of \citet{2018arXiv181203548K} for a related result.

Although our proof technique for concentration of $\hat{\Sigma}_{\rsvp}$ makes use of particular properties of Gaussian distributions, one attractive feature of the estimator is that it enjoys a certain in-built robustness to deviations from Gaussianity in the distribution of $X$. Indeed, consider now the weaker requirement that
\begin{equation} \label{eq:RSVP_dist}
X = M Z \Theta^{1/2} + \mb 1 \mu^T
\end{equation}
where $M \in \R^{(n+1) \times (n+1)}$ is invertible and the rows of $Z \in \R^{(n+1) \times p}$ are independent following (potentially different) spherically symmetric distributions, so $Z Q \eqdist Z$ for any orthogonal matrix $Q \in \R^{(n+1) \times (n+1)}$. A sufficient condition for this to occur is that the rows of $X$ are i.i.d. and have a density with elliptical contours.
In this more general setting we have the following result.
\begin{prop} \label{prop:RSVP_dist}
The law of $\hat{\Sigma}_{\rsvp}$ under \eqref{eq:RSVP_dist} above is the same as that when $X$ has independent rows distributed as $\mathcal{N}_p(\mu, \Theta)$.
\end{prop}
For example, the entries in $Z$ can have arbitrarily heavy tails; provided the spherical symmetry is satisfied, all results in this section hold under this setting and more generally under \eqref{eq:RSVP_dist}.
This may seems surprising at first sight, but is analogous to how if $\zeta$ has a spherically symmetric distribution, then the distribution of $\zeta / \|\zeta\|_2$ is simply the uniform distribution on the $d$-dimensional spherical shell, and in particular identical to the distribution obtained when $\zeta \sim \mathcal{N}_d(0, I)$.

We now turn to the expectation of $\hat{\Sigma}_{\rsvp}$.
Theorem~\ref{thm:exp_Pi} below shows that $\E \hat{\Sigma}_{\rsvp}$ is approximately a scaled version of $\Sigma$.
\begin{thm} \label{thm:exp_Pi}
Assume Condition~\ref{cond}. We have that
$\E \hat{\Sigma}_{\rsvp} = PC^2 P^T$ where $C$ is a diagonal matrix with $C$ satisfying
\begin{equation} \label{eq:ratio}
\max_{j,k \in \{q+1,\ldots,p\}}\abs{\frac{C_{jj}^2}{D_{jj}^2} - \frac{C_{kk}^2}{D_{kk}^2}} \lesssim \sigma_u \frac{n^2}{p^2}.
\end{equation}
\end{thm}
The result above shows that the ratio of $C_{jj}^2$ to $D_{jj}^2$ does not vary much across $j \in \{q+1,\ldots,p\}$ provided $p \gg n$.
In fact we also have
\begin{equation}
\max_{j \in \{q+1,\ldots,p\}} \abs{C_{jj}^2 - \frac{(n-q)D_{jj}^2}{\sum_{k=q+1}^p D_{kk}^2}} \lesssim \sqrt{\frac{n}{p}} + \frac{p}{\gamma_l n}, \label{eq:exact1}
\end{equation}
in the case where $\sigma_u$ is bounded, which reveals the form of the scale factor, and in particular its dependence on the unknown $q$. A derivation is given in Section~\ref{sec:exp_Pi} of the supplementary material. We do not make direct use of this in the proof of Theorem~\ref{thm:Sigma_rec} below however as it is only useful when $\gamma_l$ is large; in contrast, \eqref{eq:ratio} is valid for any value of $\gamma_l$.

Combining the results of Proposition~\ref{prop:pop} and Theorems~\ref{thm:conc_Pi} and \ref{thm:exp_Pi} gives the following high probability bound on the $\ell_\infty$-norm error of estimating $\Sigma$, up to an unknown scale factor.
\begin{thm} \label{thm:Sigma_rec}
Assume Condition~\ref{cond} and that $\sigma_u \lesssim p / (n \log p)$, $q \lesssim n / \log(p)$.
With probability at least $1-c/p$ for some constant $c>0$, we have that there exists $\kappa >0 $ such that
\begin{equation} \label{eq:fin_res1}
\|\Sigma - \kappa \hat{\Sigma}_{\rsvp}\|_\infty \lesssim \frac{\gamma_u \rho_1^2}{\gamma_l^2} + \rho_1 \rho_2  + \min\left(\frac{p}{n}, \gamma_u\right) \rho_2^2 + \sigma_u \frac{n}{p} + \sqrt{\frac{\log(p)}{n}}.
\end{equation}
If we additionally assume that $\rho_2^2 \lesssim q/p$ and $\rho_1$ is bounded, we have there exists $\kappa >0$ such that
\begin{equation} \label{eq:fin_res2}
\|\Sigma - \kappa \hat{\Sigma}_{\rsvp}\|_\infty \lesssim \frac{\gamma_u}{\gamma_l^2} + \sqrt{\frac{q}{p}} + \frac{q}{n} +  \sigma_u \frac{n}{p} + \sqrt{\frac{\log(p)}{n}}.
\end{equation}
\end{thm}
The first two terms in the bounds \eqref{eq:fin_res1} and \eqref{eq:fin_res2} come directly from the population-level result Proposition~\ref{prop:pop}. The remaining terms do not depend on $\gamma_l$, demonstrating how RSVP, in contrast to the PC-removal approach, does not rely on a large eigengap between $\Gamma^T\Gamma$ and $\Sigma$. The final $\sqrt{\log(p)/n}$ term is due to the variance (see Theorem~\ref{thm:conc_Pi}). Considering \eqref{eq:fin_res2}, in the case where $\sigma_u \lesssim p\sqrt{\log (p)} / n^{3/2}$, $q \lesssim \sqrt{n \log(p)}$ and $p \sqrt{\log(p)} \geq n^{3/2}$, we have that with high probability
\begin{equation*}
\inf_{\kappa> 0} \|\Sigma - \kappa \hat{\Sigma}_{\rsvp}\|_\infty \lesssim \frac{\gamma_u}{\gamma_l^2} + \sqrt{\frac{\log(p)}{n}}.
\end{equation*}
If the condition number of $\Gamma^T\Gamma$ were bounded, we only need $\gamma_l \gtrsim \sqrt{n/ \log(p)}$  for the $\ell_\infty$-norm error above to be of the same order as that achieved by the empirical covariance matrix of the (unobserved) unconfounded data $W$.

Whilst RSVP does not require strong eigengap conditions, we do need $p
\gg n$ so that the term involving $\sigma_u n/p$ due to the bias of
the estimator, is small. By sample-slitting and averaging in constructing
$\hat{\Sigma}_{\srsvp}$, we effectively reduce $n$, but only introduce an extra $\sqrt{\log(p)}$ factor in the variance term, as the following result shows.
\begin{thm} \label{thm:Sigma_rec_bag}
Let $\hat{\Sigma}_{\srsvp}$ be the sample-splitting RSVP estimator with $B$ subsamples of size $m$, so $n+1=mB$. We consider for simplicity the case where the data is column-centred in each subsample. Assume Condition~\ref{cond}, but without the requirement that $c_3 n < p$; instead suppose $c_1 m < p$ for $c_1 > 1$, and $B < p^{c_2}$ for some $c_2 > 0$.
Assume that $1 \lesssim \sigma_u \lesssim p / (m \log p)$, $q \lesssim m / \log(p)$. With probability at least $1-c/ p$ for some constant $c>0$, we have that there exists $\kappa >0 $ such that
\begin{equation*}
\|\Sigma - \kappa \hat{\Sigma}_{\srsvp}\|_\infty \lesssim \frac{\gamma_u \rho_1^2}{\gamma_l^2} + \rho_1 \rho_2  + \min\left(\frac{p}{m}, \gamma_u\right) \rho_2^2 + \sigma_u \frac{m}{p} + \frac{\log(p)}{\sqrt{n}}.
\end{equation*}
If we additionally assume that $\rho_2^2 \lesssim q/p$ and $\rho_1$ is bounded, we have that there exists $\kappa >0$ such that
\begin{equation} \label{eq:fin_res3}
\|\Sigma - \kappa \hat{\Sigma}_{\srsvp}\|_\infty \lesssim \frac{\gamma_u}{\gamma_l^2} + \sqrt{\frac{q}{p}} + \frac{q}{m} +  \sigma_u \frac{m}{p} + \frac{\log(p)}{\sqrt{n}}.
\end{equation}
\end{thm}
Considering \eqref{eq:fin_res3}, we see that for an optimal $m \asymp \sqrt{pq/\sigma_u}$ we have with high probability that
\begin{equation} \label{eq:optim_m}
\inf_{\kappa > 0} \|\Sigma - \kappa \hat{\Sigma}_{\srsvp}\|_\infty \lesssim \frac{\gamma_u}{\gamma_l^2}  + \sqrt{\frac{q \sigma_u}{p}} + \frac{\log(p)}{\sqrt{n}}.
\end{equation}
While the simple RSVP estimator
is most useful in the high-dimensional case $p\gg n$, this
  result shows that 
  sample-splitting gives good performance in moderate to
  low-dimensional settings, which will be confirmed empirically in Section~\ref{sec:experiments}.
  
\subsection{Weak confounding}
The results and discussion thus far have considered the case where $\gamma_l > \sigma_u$. In cases where the confounding is sufficiently weak such that
\[
\|\Theta - \Sigma\|_\infty = \|\Gamma \Gamma^T\|_\infty = \max_j\|\Gamma^T\Pi_{\Gamma}e_j\|_2^2 \leq \gamma_u \rho_2^2
\]
is small, and so the empirical covariance $\hat{\Theta}$ is itself a good estimator of $\Sigma$, a straightforward consequence of our previous results and their proofs is that RSVP will behave similarly to the empirical covariance.
\begin{cor}
Consider the setup of Theorem~\ref{thm:Sigma_rec} but now without any restriction on $q$ and $\gamma_l$ (so in particular $\gamma_l < \sigma_u$ is permitted). With probability at least $1-c/p$ for some constant $c>0$, there exists $\kappa > 0$ such that
\[
\|\Sigma - \kappa \hat{\Sigma}_{\rsvp}\|_\infty \lesssim \gamma_u \rho_2^2 + (\gamma_u + \sigma_u) \frac{n}{p} + \sqrt{\frac{\log(p)}{n}}.
\]
Suppose additionally that $\gamma_u/\gamma_l$ is bounded, $\rho_2^2 \lesssim q/p$, $\sigma_u \lesssim p \sqrt{\log(p)}/n^{3/2}$, $q \lesssim n \max\left(\sqrt{\log(p)/n}, 1/\log(p)\right)$ and $p \log(p) \geq n^2$. Then with probability at least $1-c/p$,
\[
\|\Sigma - \kappa \hat{\Sigma}_{\rsvp}\|_\infty \lesssim \sqrt{\frac{\log(p)}{n}}.
\]
\end{cor}
Note that the final result holds regardless of the strength of confounding, which can be arbitrarily weak or strong, though it relies on the condition number of $\Gamma^T\Gamma$ being bounded.
\section{Conditional Independence Graph Estimation and Causal Structure Learning} \label{sec:precision}
In this section we consider using RSVP in conjunction with existing methods for conditional independence graph  estimation and causal structure learning. 
We first turn to the problem of estimating the conditional independence graph corresponding to $\Sigma$: this is the undirected graph on $p$ nodes with an edge between nodes $j$ and $k$ with $j \neq k$ if and only if $w_j \notindependent w_k | w_{-jk}$, where recall $w \sim \mathcal{N}_p(\mu_w, \Sigma)$
Equivalently, we have an edge between $j$ and $k$ if and only if the precision matrix $\Omega = \Sigma^{-1}$ has $\Omega_{jk} \neq 0$.

\subsection{Conditional Independence Graph Estimation}
Methods for conditional independence graph (CIG) estimation when $p \gg n$ typically rely on $\Omega$ being sparse. Applying them directly to the observed data $X$ will in general not work well, firstly as the inverse covariance $\Theta^{-1}$ of the observed data may be far from $\Omega$, and secondly because $\Theta^{-1}$ will not be sparse but rather a sum of the sparse $\Omega$ and a low-rank component due to the presence of latent confounding. However, many of the methods for sparse precision matrix estimation require only an estimated covariance as input, and so can be readily applied to any estimate of $\Sigma$. Examples include neighbourhood selection \citep{meinshausen04consistent}, the graphical Lasso \citep{yuan05model,Yuan2010,Friedman2008} and CLIME \citep{Cai2011}. Note that as RSVP only estimates $\Sigma$ up to an unknown scale factor, we can similarly only hope to recover the precision matrix up to an unknown scale factor; this however suffices for estimating the CIG. Theoretical results for CLIME and the graphical Lasso only require an initial estimate of $\Sigma$ that is close in $\ell_\infty$-norm, so our estimation error bounds for $\Sigma$ translate directly into estimation error bounds on $\Sigma^{-1}$. We now present the corresponding result for neighbourhood selection, which is more delicate.

The procedure of neighbourhood selection involves running $p$ Lasso regressions of each variable against all others. The resulting coefficient estimates may then be used to derive an estimate of the CIG. Phrased in terms of an estimate $\hat{\Sigma}$ of the covariance, the so-called nodewise regressions take the following form:
\begin{equation} \label{eq:NS}
\hat{\beta}^{(j)} := \argmin{b \in \R^{p} : b_j=0} \left\{\frac{1}{2} b^T \hat{\Sigma} b - b^T \hat{\Sigma}_j + \lambda_j \|b\|_1 \right\}.
\end{equation}
The population level minimiser $\beta^{(j)} \in \R^p$ (i.e.\ with $\hat{\Sigma}$ replaced by $\Sigma$) of the above when $\lambda_j=0$ satisfies
\[
 {\beta}^{(j)}_l = \begin{cases}
 (\Sigma_{-j,-j}^{-1}\Sigma_{-j,j})_l \qquad & l < j,  \\
 0 \qquad & l=j, \\
 (\Sigma_{-j,-j}^{-1}\Sigma_{-j,j})_{l-1} \qquad & l > j.
 \end{cases}.
\]
The $\{\beta^{(j)}\}_{j=1}^p$ encode the CIG, indeed $w_j \independent w_k | w_{-jk}$ if and only if $\beta^{(j)}_k = \beta^{(k)}_j = 0$.  Here we will take $\hat{\Sigma} = \hat{\Sigma}_{\rsvp}$ in \eqref{eq:NS}; we thus expect that a scaled version of $\hat{\beta}^{(j)}$ gets close to $\beta^{(j)}$. 

In order to present our result on the statistical properties of $\hat{\beta}^{(j)}$, we introduce the following quantities. Let $S_j = \{l : \beta^{(j)}_l \neq 0\}$ and let $s_j := |S_j|$ and $s = \max_j s_j$; thus $s_j$  and $s$ are the degree of the $j$th node and the maximal degree in the CIG respectively.
Also define
\[
\eta_j = (\beta^{(j)})^T \Gamma \Gamma^T \beta^{(j)}.
\]
Our theory will require the $\eta_j$ to be small. We always have $\eta_j \lesssim s_j$ for all $j$. Indeed
\[
(\beta^{(j)})^T \Gamma \Gamma^T \beta^{(j)} \leq (\beta^{(j)})^T \Theta \beta^{(j)} =  \beta^{(j)}_{S_j}\Theta_{S_j,S_j} \beta^{(j)}_{S_j}.
\]
As $\Theta$ is positive semi-definite, we have $|\Theta_{lk}| \leq \max(\Theta_{ll}, \Theta_{kk}) \lesssim 1$ for all $l,k$. Thus by the Gershgorin circle theorem, $\lambda_{\max}(\Theta_{S_j,S_j}) \lesssim s_j$. Also, as
\[
1 \gtrsim \Theta_{jj} \geq \Var(w_j) \geq \Var(w_j|w_{-j}) = \|\Sigma_{-j,-j}^{-1/2}\Sigma_{-j,j} \|_2^2 \geq \sigma_l \|\Sigma_{-j,-j}^{-1}\Sigma_{-j,j}\|_2^2,
\]
we have $\|\beta^{(j)}\|_2 \lesssim 1$, whence $\eta_j \lesssim s_j$.

However in many settings we can expect the $\eta_j$ to be smaller: if we consider the column space of $\Gamma$ to have been chosen (by nature) uniformly at random conditional on $\Sigma$, then we have 
\begin{equation} \label{eq:eta_bd}
\eta_j \lesssim 1 \text{ for all } j.
\end{equation}
A derivation of this is given in Section~\ref{sec:eta_bd} of the supplementary material.

\begin{thm} \label{thm:NS}
Assume Condition~\ref{cond}.
Let
\[
\Delta := \sqrt{\frac{sn}{\log(p)}}\left\{ \frac{\gamma_u \rho_1^2}{\gamma_l^2} + \rho_1 \rho_2  + \min\left(\frac{p}{n}, \gamma_u\right)\rho_2^2 + \sigma_u \frac{n}{p} \right\}.
\]
Let $\hat{\beta}^{(j)}$ be the nodewise regression coefficient when $\hat{\Sigma} = \hat{\Sigma}_{\rsvp}$ and $\lambda_j = A\sqrt{\max(\eta_j, \Delta,1) n \log(p)}/p$ for constant $A>0$. Suppose $s=o(\sqrt{n/\log(p)})$. We have that for $A$, $n$ and $p$ sufficiently large, with probability at least $1-c/p$ for some constant $c>0$,
\begin{align*}
\|\hat{\beta}^{(j)} - \beta^{(j)}\|_2 &\lesssim \sqrt{s_j \log(p)\max(\eta_j, \Delta,1) / n} \\
\|\hat{\beta}^{(j)} - \beta^{(j)}\|_1 &\lesssim s_j\sqrt{\log(p)\max(\eta_j, \Delta,1) /n}
\end{align*}
for all $j=1,\ldots,p$.
\end{thm}
Suppose
$\rho_2^2 \lesssim q/p$, $\rho_1 \lesssim 1$, $\gamma_l^2 / \gamma_u \gtrsim \sqrt{sn / \log(p)}$, $q \lesssim \sqrt{n\log(p)/s}$ and $\sigma_u \lesssim p\sqrt{\log(p)/(sn^3)}$;
then $\Delta \lesssim 1$. If in addition $\eta_j \lesssim 1$ for all $j$, we recover the usual estimation error rates for the Lasso:
\[
\|\hat{\beta}^{(j)} - \beta^{(j)}\|_2 \lesssim \sqrt{s_j \log(p)/n}, \qquad \|\hat{\beta}^{(j)} - \beta^{(j)}\|_1 \lesssim s_j \sqrt{\log(p)/n}.
\]
The following simple corollary shows that under a minimum signal strength condition, appropriately thresholding the estimates $\hat{\beta}^{(j)}$ recovers the true CIG.
\begin{cor} \label{cor:CIG_recover}
Consider the setup of Theorem~\ref{thm:NS} and suppose $\max(\Delta, \eta_j)$. Suppose that
\[
\min_{k \in S_j} |\beta^{(j)}_k| \geq C \sqrt{s_j \log(p) / n}
\]
for all $j$ and some $C>0$. For $C$ sufficiently large, with probability at least $1-c/p$ for some $c>0$, there exists $\tau > 0$ such that defining
\[
\hat{S}_j = \{k : |\hat{\beta}^{(j)}_k| \geq \tau \sqrt{s_j \log(p)/n} \},
\]
we have $\hat{S}_j=S_j$ for all $j$.
\end{cor}

While edges in a CIG are typically given a causal interpretation, 
structural equation models \citep{pearl2009causality} and graphical
modelling with directed acyclic graphs \citet{lauritzen96graphical}
offer a more principled approach for causal inference.
Below we explain how the popular PC algorithm
\citep{spirtes2000causation}  may be run with our RSVP estimate as its input to allow for causal structure learning in the presence of hidden confounding.

\subsection{Causal Structure Learning}
In this section we describe how our RSVP estimator may be used for
causal structure learning concerning the unconfounded $w \sim \mathcal{N}_p(
\mu_w, \Sigma)$. If we assume a structural causal model for $w$ with an
underlying \emph{directed acyclic graph} (DAG) encoding parent--children
relationships \citep{pearl2009causality}, then the observational
distribution factorises according to this directed acyclic graph.
The interventional distributions under do-interventions can then be obtained by truncated
factorisations \citep{robins1986new,pearl2009causality} under an assumption known as autonomy \citep{haavelmo1944statistical}.

If the underlying DAG $G$ is unknown, it needs to be estimated from
data; for a general overview of causal structure learning see for
example 
\citet{heinze2018causal}. Under a faithfulness assumption \citep{Meek1995}, the set of conditional independencies in the observational distribution will be exactly those that may be inferred via $d$-separation from $G$. In general, there will be many DAGs compatible with the observational distribution in this way and these form an equivalence class which may be conveniently represented through  a \emph{completed partially directed acyclic
  graph} (CPDAG).  A CPDAG
contains both directed and undirected edges,
and essentially contains all the information relating to causal
structure that may be inferred from a given observational distribution
under the assumption of faithfulness.

Our goal here is to infer the CPDAG corresponding the distribution of
the unconfounded data. To do this we employ the PC algorithm
\citep{spirtes2000causation,kalisch2007estimating}. The population version of the PC algorithm is a procedure for determining the CPDAG $C(G)$ corresponding to a distribution $P$ faithful to a DAG $G$ given a list of conditional independencies satisfied by $P$. In our context where $P = \mathcal{N}_p( \mu_w, \Sigma)$ with $\Sigma$ positive definite, these conditional independencies may be equivalently represented by partial correlations: we have for $w \sim \mathcal{N}_p(\mu_w, \Sigma)$ that
\begin{equation} \label{eq:parcor}
w_j \independent w_k | w_S \;\; \Leftrightarrow \;\; \rho_{jk|S} = 0,
\end{equation}
where the partial correlation $\rho_{jk|S}$ satisfies
\begin{align} \label{eq:parcor_def}
\rho_{jk|S} = -\frac{\Psi_{jk}}{\sqrt{\Psi_{uu} \Psi_{vv}}},
\end{align}
and $\Psi^{-1} = \Sigma_{A, A}$ with $A = \{j\} \cup \{k\} \cup S$ \citep{harris2013pc}. Note that here we have indexed the rows and columns of $\Psi$ according to the elements of $A$.

The sample version of the PC algorithm replaces queries of conditional independence with conditional independence tests. In our case, in analogy with \eqref{eq:parcor} and \eqref{eq:parcor_def} we will consider tests that declare the conditional dependence $w_j \notindependent w_k | w_S$ if and only if
\begin{equation} \label{eq:CI_test}
\frac{|\hat{\Psi}_{jk}|}{\sqrt{\hat{\Psi}_{uu} \hat{\Psi}_{vv}}} \geq \tau,
\end{equation}
where $\hat{\Psi}^{-1} = \hat{\Sigma}_{A, A}$ where $\hat{\Sigma}$ is either $\hat{\Sigma}_{\rsvp}$ or $\hat{\Sigma}_{\srsvp}$; here $A$ is defined as above and threshold $\tau$ is a tuning parameter. If $\hat{\Sigma}_{A, A}$ is not invertible, we will simply accept the null of conditional independence.

In the case where confounding is not present, the PC algorithm requires faithfulness and a certain minimum signal strength condition for partial correlations. We will therefore assume that $\mathcal{N}_p(\mu_w, \Sigma)$ is faithful to a DAG $G$ and our target of inference will be the corresponding CPDAG $C(G) =: C$. We denote the maximum degree of $G$ by $d$. Define also the following parameter controlling minimum signal strength:
\[
\omega := \min\{|\rho_{jk|S}| :j,k \in V, \; S \subseteq V, \; |S| \leq d, \; \rho_{jk|S} \neq 0\}.
\]
It will also be convenient to introduce a particular minimum
restricted eigenvalue $\sigma_r$ of $\Sigma$ defined through $\sigma_r
:= \min_{I:|I| \leq d+2} \lambda_{\min}(\Sigma_{I,I})$. Note that we
always have $\sigma_r \geq \sigma_l$. 

The result below follows directly from the proof of Theorem~8 in \citet{harris2013pc}.
\begin{lem} \label{lem:PC}
Let $\hat{C}_\tau$ be the output of the PC algorithm using conditional independence tests given by \eqref{eq:CI_test} with threshold $\tau$. For any $A \geq 1$ we have
\begin{align*}
\pr\left(\hat{C}_\tau = C \text{ for all } \tau \in  \left[\frac{\omega}{2A}, 1-\frac{\omega}{2A}\right]\right) \geq  \pr\left(\inf_{\kappa>0}\|\kappa \hat{\Sigma} - \Sigma\|_\infty \leq \frac{\omega\sigma_r^2}{(4A + \omega + \sigma_r\omega)(d+2)} \right).
\end{align*}
\end{lem}
Taking $\hat\Sigma$ as either $\hat{\Sigma}_{\rsvp}$ or its
sample-splitting variant $\hat{\Sigma}_{\srsvp}$, by combining
Lemma~\ref{lem:PC} with one of Theorem~\ref{thm:Sigma_rec} or
\ref{thm:Sigma_rec_bag}, we can obtain high probability guarantees on
recovering the CPDAG corresponding to the unconfounded data. As an  example, we consider the setting where the assumptions of Theorem~\ref{thm:Sigma_rec_bag} and those leading to \eqref{eq:optim_m} hold. Additionally, consider an asymptotic regime where
\begin{equation} \label{eq:PC}
\frac{\sigma_u}{\sigma_l^2} + \sqrt{\frac{q \sigma_u}{p}} + \frac{\log(p)}{\sqrt{n}} = o(\omega/d).
\end{equation}
Then using $\hat{\Sigma}_{\srsvp}$ with an optimal subsample size $m \asymp \sqrt{pq/\sigma_u}$ we have the following conclusion: there exists a sequence $a_n \to 0$ and constant $c>0$ such that $\hat{C}_\tau = C$ for all $\tau \in [a_n, 1-a_n]$ with probability at least $1 - c/p$.
We may compare this conclusion to the results obtained in \citet{kalisch2007estimating} that provide similar guarantees for the PC algorithm when confounding is not present. If we assume the final term of $\log(p)/\sqrt{n}$ on the left-hand side of \eqref{eq:PC} is the dominant one, our requirement is $\log(p)/\sqrt{n} = o(\omega / d)$ whereas the equivalent result in \citet{kalisch2007estimating} only requires $\sqrt{\log(p)/n} = o(\omega / \sqrt{d})$. In particular, we see that in our setting, the maximal degree $d$ cannot grow as quickly. This restriction is also present in the analogous result of \citet{harris2013pc} who consider applying the PC algorithm (in the absence of hidden confounding) using conditional independence tests based on partial correlations derived from rank correlations.

\section{Numerical results} \label{sec:experiments}

\subsection{Simulation experiments}
In this section we provide some numerical results for various scenarios and compare
the proposed estimator with the PC-removal estimators, as
employed in POET \citep{fan2013large}. Results for shrinkage
estimators of Ledoit--Wolf type \citep{ledoit2004well} are also be included in our comparison.

\subsubsection{Experimental setups}
We consider five different scenarios described below. For each of these, we generate $n\in \{100,200,500,1000,2000\}$
independent samples from $\mathcal{N}_p(0,\Theta)$ for a covariance
matrix $\Theta \in \mathbb{R}^{p\times p}$ that has an idiosyncratic
component and a component due to confounding $\Theta= \Sigma + \Gamma^T\Gamma$.
The number of variables is varied in $p\in\{100,200,500,1000,2000\}$.
For $q$ latent confounders, the entries of the  matrix
$\Gamma\in \mathbb{R}^{p\times q}$ are sampled independently from a
  standard normal distribution, and  column $k\in \{1,\ldots,q\}$ of
  $\Gamma$ is scaled by a factor $\nu \exp(-k)$ to have a decaying
  spectrum among the latent confounders. The strength $\nu\in \{0.01,0.1,0.5,1,5,20\}$ allows for a  variation of the overall strength of the latent confounding.
The five scenarios considered distinguish themselves by a different
structure of the idiosyncratic covariance matrix $\Sigma$ and the
number of latent confounders $q$. All diagonal entries of $\Sigma$ are
set to 1.

\begin{figure}
\begin{center}
\includegraphics[width=0.95\textwidth]{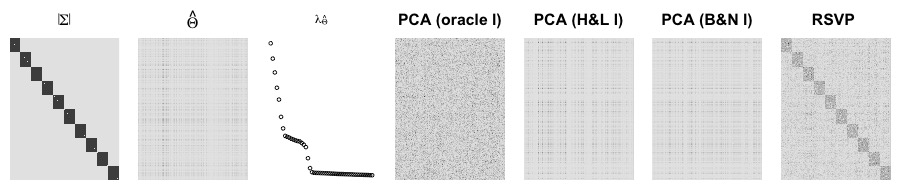}
\caption{An example of  block structure with $p=1000$ and  $n=100$ and
  strong latent confounders ($\nu=20$). The results are presented as in Figure~\ref{fig:example1}. The empirical covariance matrix $\hat{\Theta}$
 and the PC-removal estimates fail to recover the block structure.  \label{fig:example3} }
\end{center}
\end{figure}

\begin{figure}
\begin{center}
\includegraphics[width=0.95\textwidth]{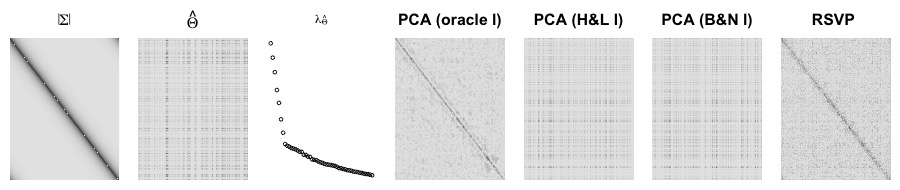}
\caption{The same setup as in Figure~\ref{fig:example3} for a Toeplitz
  stucture of the idiosyncratic covariance and $q=20$ latent confounders.  \label{fig:example4} }
\end{center}
\end{figure}

\begin{enumerate}[]
\item[] {\bf Block structure.} The $p$ variables are divided into ten
  blocks of equal size. The correlation within each block is set
  uniformly  to
  $0.95$ and $0$ outside of blocks, with unit variance for all variables. There are $q=20$ latent
  variables in this scenario.
\item[] {\bf Block structure II.} Half of the variables are divided into ten
  blocks of equal size, similarly to the previous scenario. The remaining
  variables form one large block. The within-block correlation is $0.5$
  and between-block correlation is again $0$, The correlation within each block is set to
  $0.95$ and unit variance for all variables. There are $q=20$ latent
  variables in this scenario.
\item[] {\bf Toeplitz structure.} The inverse idiosyncratic covariance marix is set
  to a unit diagonal and first off-diagonal entries equal to $-0.4999$
  (with circular extension). Variables are then scaled to have unit
  variance.  There are again $q=20$ latent  variables in this
  scenario.
\item[] {\bf Toeplitz structure II.} Identical to the previous
  Toeplitz design, except that the number of latent confounders is reduced to
  $q=3$.
\item[] {\bf Erd\H{o}s--R\'enyi.} The nonzero entries of the inverse
  idiosyncratic covariance are chosen randomly, each edge being
  selected with probability $10/p$. The diagonal of the inverse is set
  to unit values initially, and all off-diagonal entries are set to
  constant such that the sum of all non-diagonal entries in each row
  is bounded by 0.99 and the inverse matrix hence diagonal dominant
  and invertible. The variables are in a second step again scaled to
  have unit diagonal entries in the idiosyncratic covariance $\Sigma$.
\end{enumerate}

Varying the structure, number of samples $n$, dimension $p$, and
strength $\nu$ of the latent confounders, we run 200 simulations
of each unique parameter configuration and compute the following:
\begin{enumerate}[(i)]
\item The estimated covariance matrix
$\hat{\Sigma}_\pca(\ell)$, where the number $\ell$ is chosen first as
$\ell=0$, leading to the empirical covariance matrix. This first
estimator is also the basis for comparisons with Ledoit--Wolf type
shrinkage \citep{ledoit2004well}\footnote{The results for a Ledoit--Wolf covariance estimator with the identity matrix as the shrinkage target are identical to those for PC-removal with $\ell=0$ (i.e.\ the empirical covariance matrix) as the objective we measure will be unchanged by the shrinkage.}. Next we use the oracle value $\ell=q$ (which is of course unavailable in practice)
and then, as suggested in \citet{fan2013large}, the values of the
two estimators of $q$ that are based on the respective first information criteria in~\citet{bai2002determining}
and~\citet{hallin2007determining}. We henceforth refer to these as B\&N and H\&L respectively.
\item The sample-splitting RSVP estimator $\hat{\Sigma}_\srsvp$ for subsample size $m\in
  \{20,50,70\}$.
\end{enumerate}
Other possible approaches such as the sparse--dense decomposition approach of
\citet{chandrasekaran2012latent} are unfortunately computationally
infeasible for these settings.

\begin{figure}
\begin{center}
\includegraphics[width=0.95\textwidth]{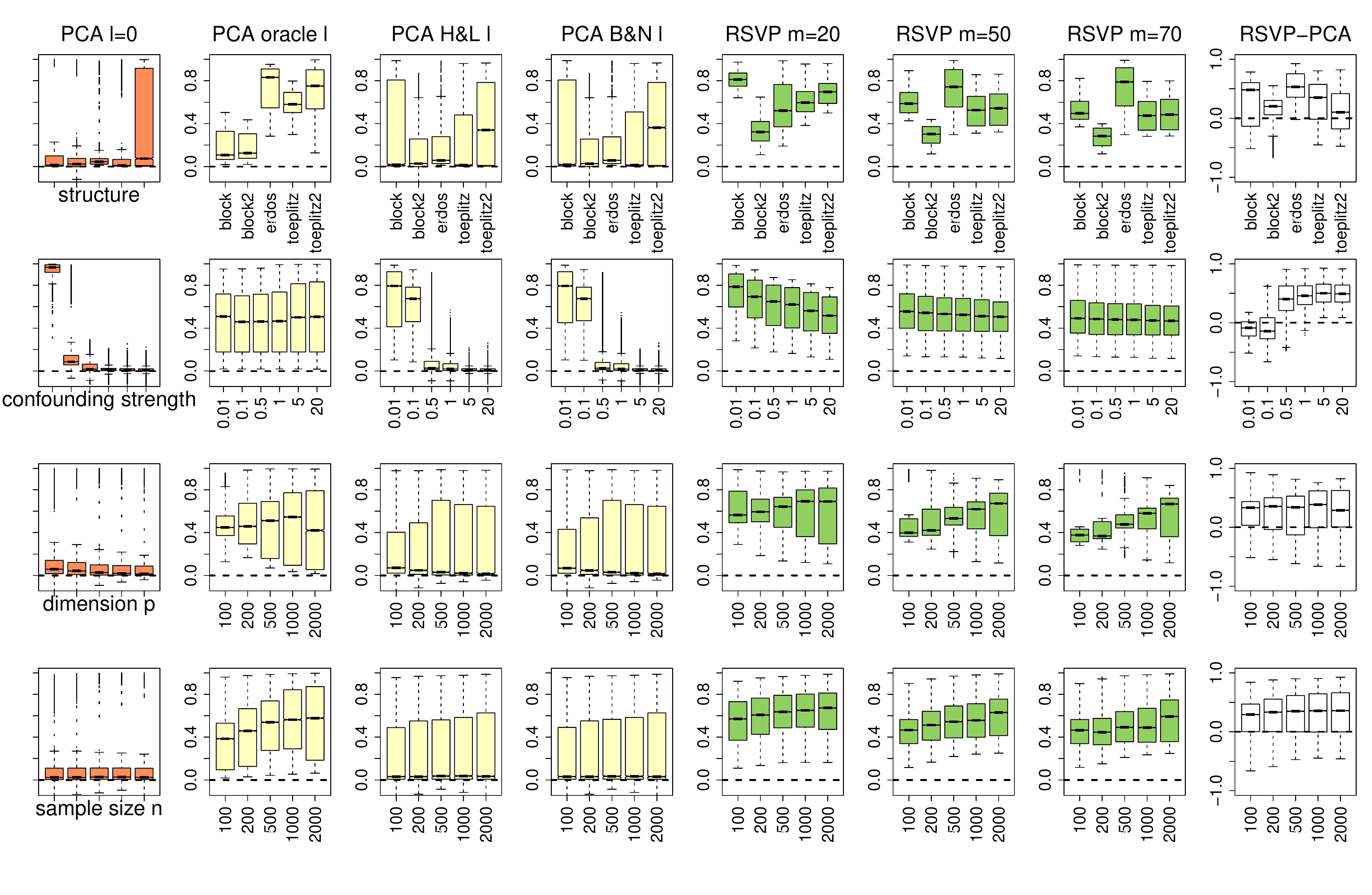}
\caption{  Boxplots of the correlation $\rho_{\Sigma,\hat{\Sigma}}$
  for various methods (columns) and stratified according to (from top to
  bottom row): design matrix structure, strength $\nu$ of the latent
  variables, dimension $p$ and sample size $n$. The methods are as
  in Figure~\ref{fig:example1} but here also include larger
  number $m$ of samples in each subsample for the RSVP estimator. The last column is a
  paired comparison: the difference between the RSVP estimator with
  $m=70$ and the PC-removal estimator with a H\&L choice of the number
  $\ell$ of components to remove. The relative advantage of RSVP
  grows with stronger latent confounding and larger sample size.\label{fig:all} }
\end{center}
\end{figure}

\begin{figure}
\begin{center}
\includegraphics[width=0.95\textwidth]{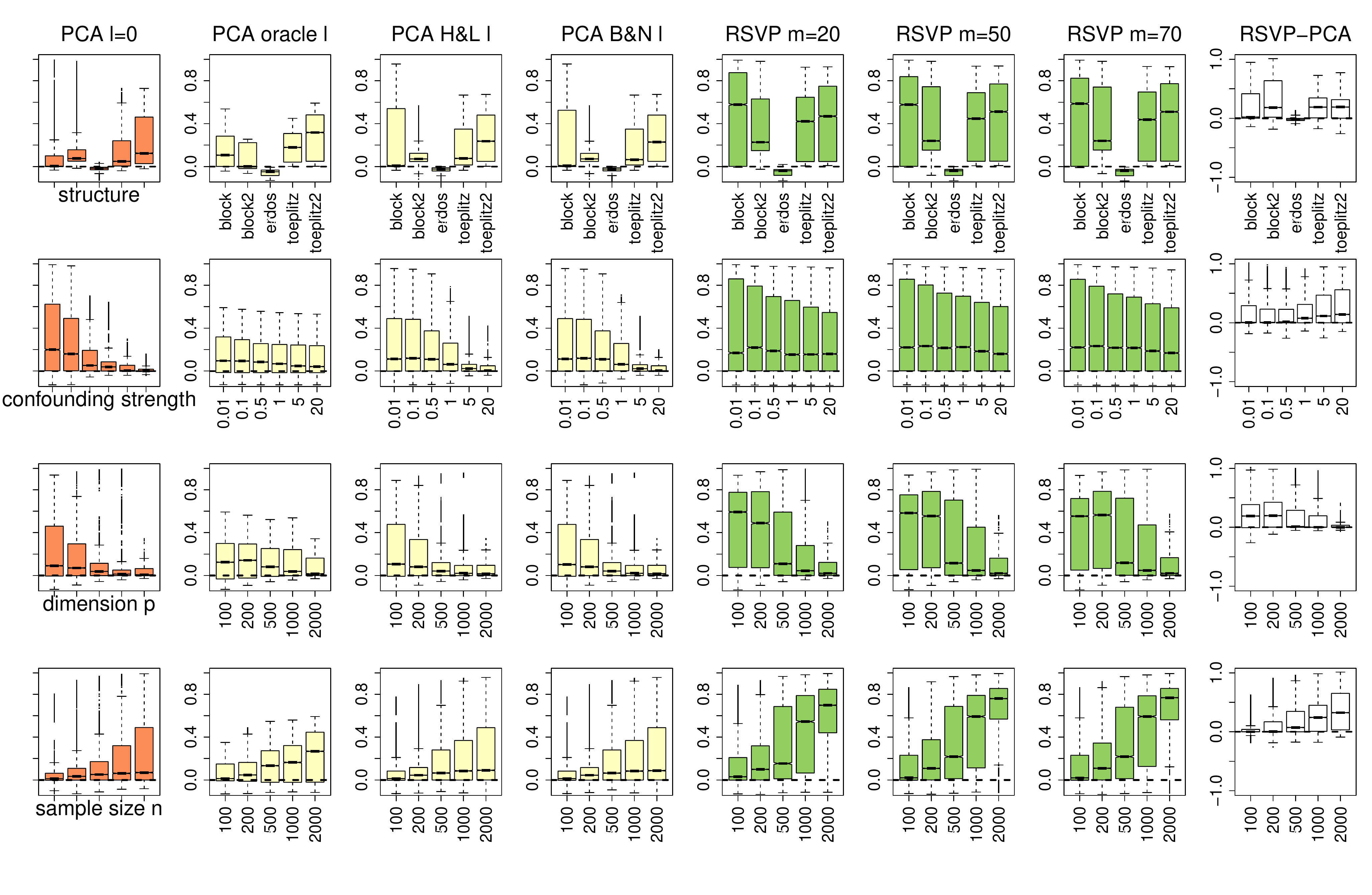}
\caption{  Analogous results to those in Figure~\ref{fig:all} for
  inverse covariance matrix estimation. As before, the relative advantage of RSVP
  grows with stronger latent confounding and larger sample size  \label{fig:allinv} }
\end{center}
\end{figure}

We would like to compare for each estimate its accuracy with respect
to the true idiosyncratic covariance in a suitable
norm, which we chose here for simplicity as the Frobenius norm. To be
invariant with respect to scaling, we may consider
 \[ \inf_{\kappa > 0}  \| \Sigma -\kappa \hat{\Sigma}\|_F ,\]
which is monotonically decreasing with the empirical correlation
$\rho_{\Sigma,\hat{\Sigma}}$
between the vectorized matrices $\Sigma$ and $\hat{\Sigma}$; we
will use $\rho_{\Sigma,\hat{\Sigma}}$ as a criterion for simplicity, and also omit the diagonals from $\Sigma$ and $\hat{\Sigma}$ in the computation.
For \emph{inverse} covariance matrix estimation, we invert the
estimators above using the approach of \citet{meinshausen04consistent} as implemented in the R package \texttt{glasso} \citep{Friedman2018}. The
penalty parameter is set to a very small uniform value of
$\lambda=10^{-6}$ for computational speed and easier comparison
between methods.
Cross-validation of the penalty is also not straightforward to implement here as we do not have access to clean data that would be
free of the influence of the latent confounders.

\subsubsection{Results}
A summary of results from each of the $750=5 \times 5 \times 6 \times 5$ unique
parameter setting is shown in Figure~\ref{fig:all}. The  RSVP estimator with low
number $m=20$ of samples in each subsample in general dominates the other estimators (in terms
of having higher mean correlation and higher quartiles), no matter
whether we stratify according to design matrix structure, strength of
latent confounders, sample size or dimension of the graph. The only exception
seems to be the case of $\nu=0.01$, where the latent confounders are
effectively absent. Here the empirical covariance improves the RSVP
estimator, as expected.

Comparing the various PC-removal approaches, it is noteworthy that
for an increasing strength of the latent confounding, the oracle (true)
value of $q$ performs much better than using any of the suggested
empirical estimates of $q$. In contrast, for weak confounding, removing
all $q$ latent confounders performs worse in general due to the decaying
spectrum of the latent confounding: too much of the
idiosyncratic covariance is removed by the oracle estimate in these
cases. RSVP tends to perform at least as good as the optimal approach
among the three PC-removal approaches across all strengths of the
latent confounding, even though in practice the oracle choice of $q$ for
PC-removal is
clearly not even available.

\begin{figure}
\begin{center}
\includegraphics[width=0.95\textwidth]{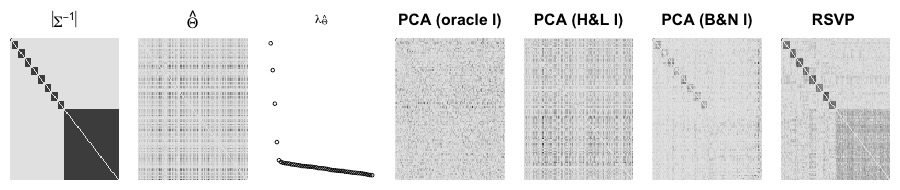}
\caption{An example of  block structure with $p=100$ and  $n=500$ and
  medium strong latent confounders ($\nu=20$). The results presented are analogous to those in
  Figure~\ref{fig:example1}, but here we are interested in inverse
  covariance estimation (via nodewise regression on the estimated
  idiosyncratic covariances). The leftmost panel shows the absolute
  values of  $\Sigma^{-1}$, while the four panels on the right show the
  absolute values of the inverted estimated $\hat{\Sigma}$, using the
  same methods as in Figure~\ref{fig:example1}.  \label{fig:example5} }
\end{center}
\end{figure}

Analogous results for inverse covariance matrix estimation are
shown in Figure~\ref{fig:allinv}, with a single example outcome in Figure~\ref{fig:example5}. The differences between the
RSVP with different number of samples in each subsample are smaller, arguably because the
error introduced by matrix inversion dominates the relatively small
differences. While estimating the covariance of a random
Erd\H{o}s--R\'enyi graph seems easy for the covariance, it becomes relatively
hard for the inverse covariance matrix. Finally, while a dimension of $p=2000$  still
yields very good results in Frobenius norm for covariance estimation,
it seems to become very challenging for inverse covariance
estimation.

\begin{figure}
\begin{center}
\includegraphics[width=0.98\textwidth]{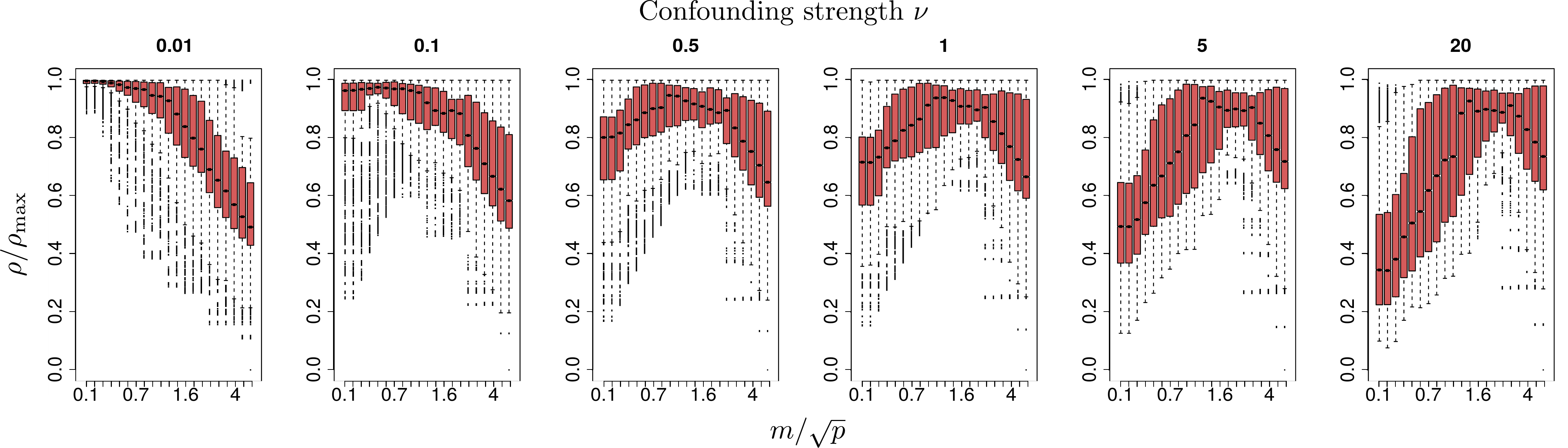}
\caption{
 Performance 
 as a function of sample size $m$ used for the sample-splitting
 version of 
  RSVP. For each scenario, we divide $\rho_{\Sigma,\hat{\Sigma}}$ by
  the maximal value across all subsample sizes $m$. The figure shows the
  boxplots for different strengths of latent confounding as a
  function of $m/\sqrt{p}$. For weak latent confounding (leftmost plot)
  taking $m=1$ is optimal as then RSVP is equivalent to the empirical
  covariance matrix on row-normalised data. For stronger latent confounding, a value $m=c
  \sqrt{p}$ with  $c \approx 2$ performs well across a wide range   of scenarios. \label{fig:optm} }
\end{center}
\end{figure}

The relative performance of the  sample-splitting version of RSVP as a
function of number of samples $m$ in each subsample is shown in Figure~\ref{fig:optm}.
For very weak latent confounding, taking very small values of $m$
performs optimally as the sampling-splitting RSVP estimator then converges to the
empirical covariance matrix.  While the scaling of the optimal $m$ as
proportional to $\sqrt{pq}/\sigma_u$ emerges from the theory,
In our examples the choice $m=2\sqrt{p}$ seems to be a good
rule-of-thumb choice for the  size of the subsamples.

\subsubsection{Model violations}

\begin{figure}
\begin{center}
\includegraphics[width=0.95\textwidth]{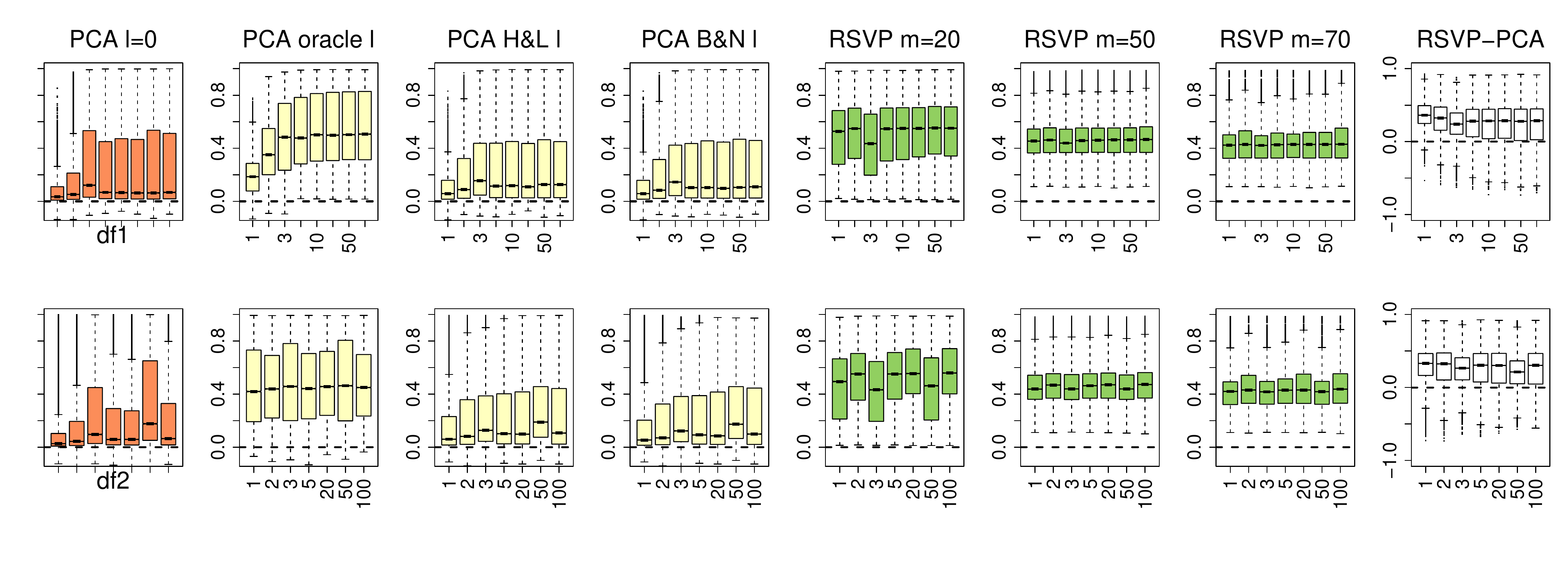}
\caption{ The performance of PC-removal methods and RSVP when
  replacing multivariate normal distributions for $w$ and $h$, and
  generation mechanism for the entries of $\Gamma$ by multivariate
  $t$-distributions with $\mathrm{df}_1$ and $\mathrm{df}_2$ degrees
  of freedom respectively. While the performance of PC-removal approaches deteriorates considerably for small degrees of freedom in the idiosyncratic noise distributions, the performance of RSVP is largely unaffected by these heavy-tailed distributions. \label{fig:t-distr} }
\end{center}
\end{figure}

To investigate robustness against model violations for covariance
estimation, we replace the normal distribution for the idiosyncratic
noise for $X$ and $H$ by multivariate $t$-distributions with
$\mathrm{df}_1$ degrees of freedom, where
$\mathrm{df}_1\in\{1,2,3,5,10,20,50,100\}$. We also generate the
loading matrix $\Gamma$ using a multivariate $t$-distributions with
$\mathrm{df}_2$ degrees of freedom and vary this parameter among the
same set of values as those used for $\mathrm{df}_1$. Analogously to
Figure~\ref{fig:all}, Figure~\ref{fig:t-distr} shows the performance
for covariance estimation marginally as a function of both
$\mathrm{df}_1$ and $\mathrm{df}_2$, where the remaining parameters
(graph structure, dimension, sample size  and strength of confounding)
are averaged out. The cases $\mathrm{df}_1=1$ and
  $\mathrm{df}_2=1$ correspond to Cauchy distributions
  respectively. We comment here that $\Sigma$ does not correspond to a
  covariance if $\mathrm{df}_1 \leq 2$. Nevertheless, $\Sigma$ can
  still be identifiable from the distribution of $w$.

\begin{figure}
\begin{center}
\includegraphics[width=0.95\textwidth]{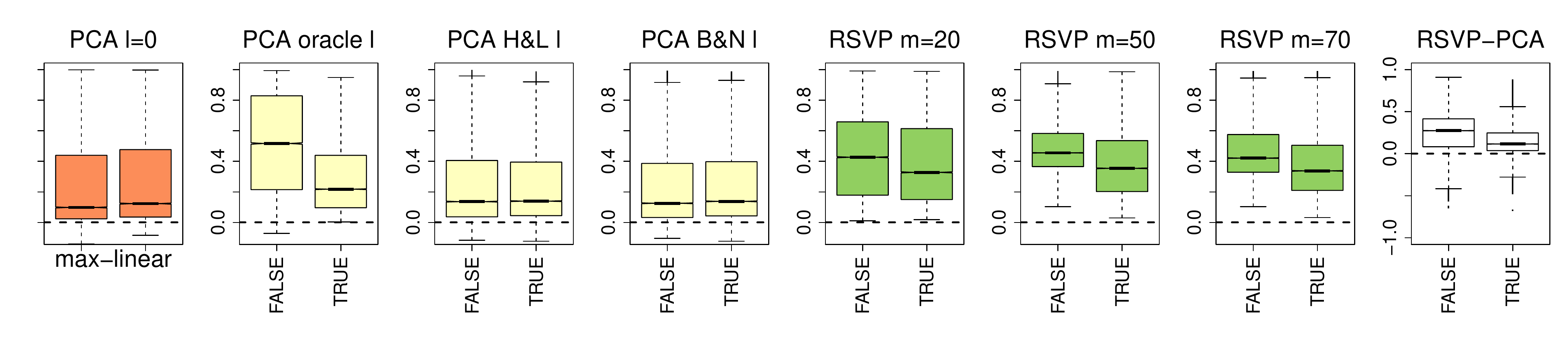}
\caption{ The performance of PC-removal methods and RSVP when the linear structural equation for $x$ is replaced with a max-linear model \eqref{eq:SEM2}.  The average performance for the linear model is shown as a boxplot for max-linear equal to `false', while the max-linear case corresponds to the boxplot with max-linear equal to `true'. We see the advantage of RSVP over PC-removal methods deteriorates under the max-linear model.\label{fig:max-linear} }
\end{center}
\end{figure}
As an additional test of robustness, we consider, in a second set of
experiments, replacing the linear structural equation $x = w + \Gamma
h$ \eqref{eq:SEM1} with a max-linear model \citep{gissibl2018}
\begin{equation} \label{eq:SEM2}
x _j = \max(w_j, (\Gamma h)_j);
\end{equation}
our goal is as before to recover $\Sigma = \Cov(w)$. We present in Figure~\ref{fig:max-linear} the results averaged over all other parameters of our simulation setup (graph structure, dimension, sample size, strength of confounders, $\mathrm{df}_1$, and  $\mathrm{df}_2$). The performances of both oracle PC-removal and RSVP suffer in the max-linear case and drop to similar levels to the data-driven PC-removal methods. However, even in this case, RSVP outperforms data-driven PC-removal approaches; in the case of the H\&L choice of the number of components, RSVP gives better results in more than three quarters of all simulation settings, as can be seen in the rightmost panel of Figure~\ref{fig:max-linear}.

\subsection{GTEX data analysis} \label{sec:GTEX}

In this section we illustrate the key properties of RSVP on a collection
of gene expression datasets made publicly available by the GTEX consortium \citep{Aguet2017}. Such datasets are particularly prone to the type of confounding studied in this paper \citep{Leek2007, Stegle2012, Speed2013}.
Our aim is to determine which genes are biologically related in that they regulate each other.
To validate our results, we use the gene ontology database \citep{Ashburner2000}.

The GTEX consortium conducted a large-scale RNA-seq experiment which resulted in the
the collection
of gene expression data from hundreds of donors in more than 50 human tissues.
In order to carry out their analyses, they estimated
confounders by leveraging external information such as gender and genetic relatedness between donors, and by inferring some
confounders from the data itself using probabilistic estimation of expression residuals (PEER) \citep{Stegle2012}.
Both the confounders and the fully processed, normalised and filtered gene expression data are available on
the website of the consortium\footnote{\url{https://gtexportal.org/home/datasets}. In addition, code to compute RSVP and subsampling versions, and also to reproduce all the results described in this section, is available at \url{https://github.com/benjaminfrot/RSVP}.}.

For each tissue $\mathcal{T}$, where $\mathcal{T}$ is for example whole blood, lung or thyroid,
there is available a data matrix $X_{\mathcal{T}}$ of gene expression levels with dimensions $n_\mathcal{T} \times p_\mathcal{T}$ along with an $n_\mathcal{T} \times q_\mathcal{T}$ matrix of confounders.
We removed tissues for which $n_{\mathcal{T}} \leq 100$; the 44 remaining tissues had a ratio $n_\mathcal{T}/p_\mathcal{T}$ ranging between $0.006$
and $0.03$ and values of $p_\mathcal{T}$ ranging between $14,337$
and $17,855$\footnote{The list of tissues as well the number of samples and variables for each of them can be found in the supplementary materials.}. In line with the analysis methods of the GTEX consortium,
we used all the PEER factors at our disposal\footnote{According to the \emph{Analysis Methods} section of the consortium's website ``the number of PEER factors was determined as function of sample size ($N$): 15 factors for $N<150$, 30 factors for $150\leq N<250$, 45 factors for $250\leq N<350$, and 60 factors for $N\geq350$ (...).''}, resulting in a total number of $q_{\mathcal{T}}$ confounders for each tissue equal to the number of 
PEER factors for that tissue plus five 
confounders derived from external sources (\emph{e.g.} donors' genotypes, gender, etc\dots).
Because these covariates and factors are deemed the most relevant by the GTEX consortium, we refer to a dataset $X_{\mathcal{T}}$ from which all $q_{\mathcal{T}}$ confounders have been 
removed as ``unconfounded''. However, it is possible that there is still unobserved confounding in the datasets.

For each tissue, we create a sequence of datasets by regressing out $0, 1, 2, \ldots, q_{\mathcal{T}}$ confounders.
On each of these datasets, we run RSVP, PC-removal with different values $\ell$ of components removed. We also run the neighbourhood selection with the square-root Lasso \cite{belloni2011} on both the sample covariance matrix of 
the raw dataset  (NS) and on the covariance matrix estimated by RSVP (RSVP + NS) .
Two commonly used proxies for pairs of genes being co-regulated are large off-diagonal entries in the covariance or non-zero entries in the inverse covariance matrix. We therefore form for each estimated covariance matrix, a sequence of estimated co-regulation networks containing edges corresponding to the largest $r$ entries, with $r$ ranging from $1$--$100$. In the case of NS and RSVP + NS, we vary the tuning parameter of the square-root Lasso until we obtain a graph with approximately 100 edges and then form a sequence of 100 networks corresponding to the largest $r$ entries in the estimated inverse covariance matrices, with $1 \leq r \leq 100$.

We first sought to quantify how sensitive the graphs returned by the various methods are to the addition of confounding.
To that end, for each (tissue, method, $r$) triple, we computed the Jaccard similarity between the edge set of a graph
estimated on the unconfounded data and the graph with $r$ edges estimated on the dataset with $k \in \{0, 1, 5, 10, 30\}$ confounders removed.
Figure \ref{fig:jacc_sim}
shows the resulting
Jaccard similarities averaged across
the 44 tissues.
Unsurprisingly the more confounders are removed, the more similar the estimated graphs are to that obtained on the unconfounded data ($k=q_{\mathcal{T}}$).
However, this change for RSVP is only very slight and the method yields large similarities across different numbers of edges and $k$.
This is an encouraging result, particularly given that a number of the confounders, such as gender and genotype data,  were derived entirely from \emph{external} data.
In contrast, the performances of PC-removal and NS are strongly influenced by the presence of
the confounders, with the Jaccard similarity between raw and unconfounded data close to zero.

\begin{figure}
\begin{center}
\includegraphics[width=0.95\textwidth]{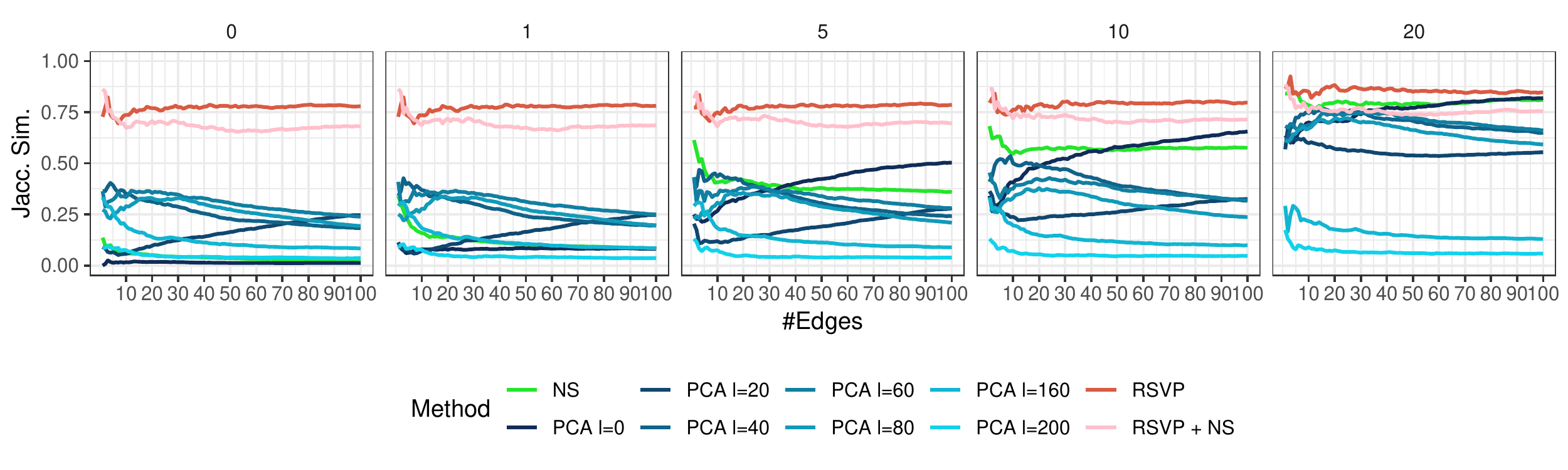}
\caption{
 Average Jaccard similarity between the edge sets of graphs estimated
 on the unconfounded data (with all confounders regressed out) and
 data from which $k$ confounders are removed, for $k = 0, 5, 10, 30$;
 similarities are averaged over 44 tissues. The RSVP estimate is here
 seen to be the most stable with respect to removal of confounders.
The RSVP
  estimator shows highest Jaccard similarity across all graph sizes
  when zero or just a few confounders are regressed out.
  \label{fig:jacc_sim}
  }
\end{center}
\end{figure}

Consistently returning the same set of edges irrespective
of confounding does not imply anything about the
quality of the estimates. To get a sense of their accuracy, we scored the graphs using a reference dataset: the gene ontology \citep{Ashburner2000}.
Briefly, the gene ontology (GO) is a popular database which allows the annotation of each
gene by a set of \emph{terms} classified in three categories:
cellular components, molecular function and biological process.
Genes that tend to perform similar functions or to interact
are expected to be annotated by similar terms.
By mapping each node of each graph to its GO terms,
one can compute a so-called enrichment statistic \citep{FrotEtAl18} reflecting whether
the graph contains edges between related genes more often than
would be expected in a random graph with a similar topology (such a
graph has an expected statistic of 1).
The top plot in Figure~\ref{fig:gtex_enrich_scores} shows
the enrichment scores obtained in
the raw dataset (no confounders regressed out), averaged across all tissues. The bottom plot
gives the average score
as a function of the number of confounders regressed out. In the supplementary materials, the scores for each of the 44 tissues is plotted.
Several comments are in order. RSVP performs well across the datasets, and is the best performer on average when applied to the unconfounded data. 
Interestingly, as shown in the supplementary materials, there is at least one selection of $\ell$ for each tissue where PC-removal performs comparably to RSVP, but the optimal value of $\ell$ changes from tissue to tissue. This would suggest a data-based selection for $\ell$; however the selection criteria of \citet{bai2002determining}
and~\citet{hallin2007determining} both yield $\ell=0$ on every tissue.
The performance of the neighbourhood selection (NS)
steadily increases as more and more
confounders are regressed out, until it outperforms RSVP.
This tends to confirm that the raw data does indeed contain latent confounders masking true biological signal. 
Moreover, the fact that methods forming networks based on the estimated inverse covariances (NS and RSVP + NS) perform best on the unconfounded datasets tends to confirm that
it is indeed the precision matrices which contain relevant signal when it comes to co-regulation networks. 

The computational cost of performing NS, is far greater than RSVP or the PC-removal approaches. We also note that the latter methods may be further sped up by using large inner product search algorithms. For example, the xyz algorithm of \citet{Thanei2018} is able to locate the large entries in the matrix product $VV^T$ that forms RSVP at a fraction of the cost of performing the full matrix multiplication. On these GTEX datasets, it delivers similar performance to regular RSVP but cuts the computational cost by a factor of around $2000$.

\begin{figure}
\begin{center}
\begin{subfigure}{\textwidth}
\includegraphics[width=0.95\textwidth]{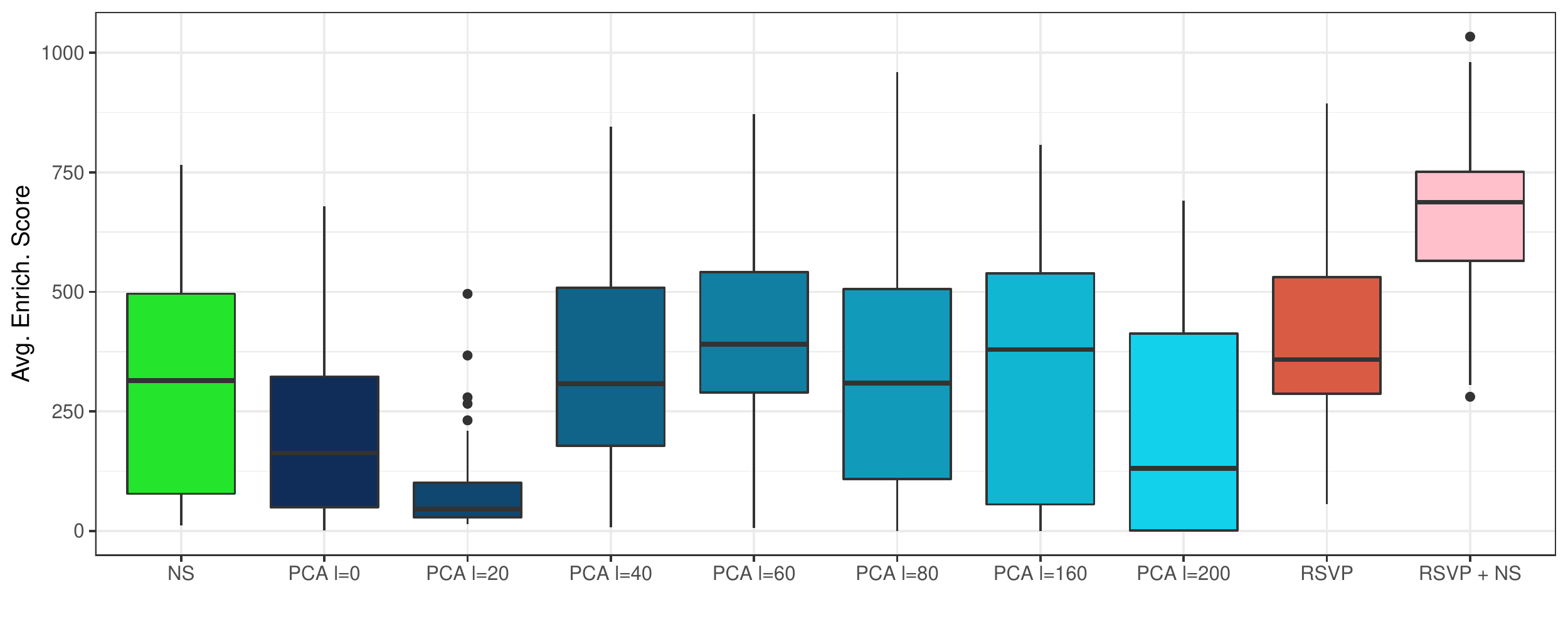}
\end{subfigure}
\\
\vspace{0.5cm}
\begin{subfigure}{\textwidth}
\includegraphics[width=0.95\textwidth]{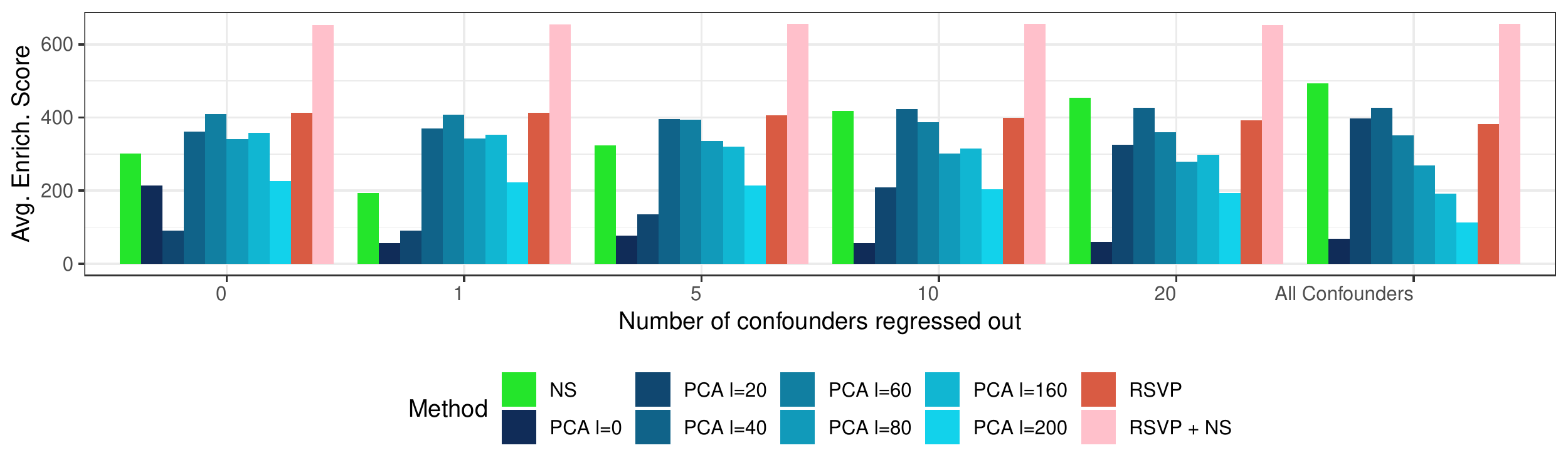}
\end{subfigure}
\end{center}
\caption{Top: Area under the curve (AUC) of the graph of enrichment score as a function of the number of edges, based on the raw data. We plot the distribution of the scores across the 44 tissues.
Bottom: The average of the AUCs across tissues, but for data with varying numbers of confounders regressed out. The individual results for each tissue are presented in Section~\ref{sec:GTEX_indiv} of the supplementary material.
\label{fig:gtex_enrich_scores}}
\end{figure}

\section{Discussion} \label{sec:discuss}
In this work, we have introduced RSVP as a simple and fast method for
estimating the idiosyncratic covariance $\Sigma$ given data where
latent factors are present. A notable aspect of the method is that all
information about $\Sigma$ contained in the spectrum of the empirical covariance matrix is thrown away. Estimation of $\Sigma$, which is permitted to have a diverging condition number, is performed using a scaled multiple of a projection matrix whose eigenvalues are necessarily in $\{0,1\}$. It may seem surprising at first sight that this should work at all, and the success of the method underlines the message that has emerged on the vast theory surrounding high-dimensional PCA and covariance estimation, saying that the eigenvalues of the empirical covariance matrix $\hat{\Theta}$ are extremely noisy. By removing the variance due to these noisy eigenvalues, RSVP is able to cope well even in settings that are particularly challenging for PC-removal approaches where the eigenvalues of the combined covariance $\Theta$ are not well-separated into two groups. A drawback of RSVP is that the scale of $\Sigma$ is lost, but this is of little consequence in a number of applications of interest, and has the advantage of allowing the method to be robust to certain heavy-tailed data, for example.


Our work leaves open a number of questions. For example, it would be interesting to explore whether there are other estimators of the form \eqref{eq:Sigma_H} that depend on the spectrum of $\hat{\Theta}$ such that the scale of $\Sigma$ is not lost, but in a sufficiently smooth way as to not have high variance even in the challenging scenarios mentioned above. Another interesting problem is that of controlling for latent confounding when the influence of the confounding is not linear, such as the max-linear settings \citep{gissibl2018} for example.

\end{cbunit}

\newpage
\setcounter{page}{1}

\begin{cbunit}
\appendix
\section*{Supplementary material}
This supplementary material contains the proofs of results presented in the main text. The proofs of Proposition~\ref{prop:pop}, Theorems~\ref{thm:conc_Pi}, \ref{thm:exp_Pi}, \ref{thm:Sigma_rec}, \ref{thm:Sigma_rec_bag} and \ref{thm:NS}, and derivations of \eqref{eq:rho_bd} and \eqref{eq:exact1} all rely on some basic results stated in Section~\ref{sec:basic}. In addition to the notation laid out in Section~\ref{sec:notation} of the main paper, here we will additionally use $\eqdist$ to denote equality in distribution, and for positive semidefinite matrices $A, B \in \R^{d \times d}$, $A \succeq B$ will mean that $A-B$ is positive semidefinite.
\section{Proof of Proposition~\ref{prop:pop}}
The proof of Proposition~\ref{prop:pop} relies heavily on the so-called Davis--Kahan $\sin(\theta)$ theorem \citep{davis1970rotation}. The following version of the result will be most useful for our purposes.
\begin{thm}[Davis--Kahan $\sin(\theta)$ theorem]
Let $M = R_0 M_0 R_0^T + R_1 M_1 R_1^T$ and $M + E= S_0 \Lambda_0 S_0^T + S_1^T \Lambda_1 S_1^T$ be real symmetric matrices with $(R_0, R_1)$ and $(S_0, S_1)$ orthogonal matrices where $R_0$ and $S_0$ have matching dimensions. If the eigenvalues of $M_0$ are
contained in an interval $(a, b)$, and the eigenvalues of $\Lambda_1$ are excluded from the interval $(a-\delta, b + \delta)$ for some $\delta  > 0$, then
\[
\|S_1^T R_0\| \leq \|S_1^T E R_0\|/\delta.
\]
\end{thm}
We apply this result with $M=R_0M_0R_0^T = \Gamma \Gamma^T$ and $E = \Sigma$.
Let $Q \in \R^{p \times q}$ be the matrix of left singular vectors of $\Gamma$, and let $\Gamma = QA$ where $A \in \R^{q \times q}$.
Also let $P_F \in \R^{p \times q}$ and $P_L \in \R^{p \times (p-q)}$ be the matrices of first $q$ and last $p-q$ eigenvectors of $\Theta$. Also let $D_F^2 \in \R^{q \times q}$ and $D_L^2 \in \R^{(p-q)\times (p-q)}$ be the top left and bottom right submatrices of $D^2$ respectively.
The Davis--Kahan theorem in conjunction with Proposition~\ref{prop:basic_eval} then tells us that
\begin{equation} \label{eq:DK}
\|P_L^T Q\| \leq \|P_L^T \Sigma Q\| / (\gamma_l - \sigma_u) \lesssim \rho_1 /\gamma_l.
\end{equation}
Now $Q^T P_L P_L^T Q = Q^T(I-P_F P_F^T)Q = I - Q^T P_F P_F^T Q$ so
\begin{equation} \label{eq:l_max_1}
\lambda_{\max}(Q^T P_L P_L^T Q) = 1 - \lambda_{\min}(Q^T P_F P_F^T Q).
\end{equation}
Also as $P_F^T(I -QQ^T)P_F = I- P_F^T QQ^TP_F$, we have
\begin{align}
\lambda_{\max}(P_F^T(I- QQ^T)P_F) &= 1 - \lambda_{\min}(P_F^T QQ^TP_F) \notag\\
&= 1 - \lambda_{\min}(Q^TP_F P_F^T Q) \notag\\
&= \lambda_{\max}(Q^T P_L P_L^T Q) \label{eq:l_max_2}
\end{align}
from \eqref{eq:l_max_1}. Thus
\begin{equation} \label{eq:l_max_3}
\|(I-QQ^T)P_F\| \lesssim \rho_1 /\gamma_l.
\end{equation}


With these facts in hand, we now turn to the problem of bounding $\|\Sigma - P_LP_L^T\Theta P_LP_L^T\|_\infty$. To this end, let us decompose
\[
P_LP_L^T\Theta P_LP_L^T  = P_L P_L^T\Sigma P_L P_L^T + P_L P_L^T \Gamma \Gamma^T P_L P_L^T.
\]
Consider the second term.
We see that
\begin{align*}
\lambda_{\max}(P_L P_L^T \Gamma \Gamma^T P_L P_L^T) &= \lambda_{\max}(P_L^T QAA^T Q^T P_L) \\
&\lesssim \gamma_u \rho_1^2/ \gamma_l^2.
\end{align*}
Also
\begin{align*}
\Sigma - P_LP_L^T\Sigma P_L P_L^T = P_F P_F^T\Sigma P_L P_L^T + P_L P_L^T\Sigma P_F P_F^T + P_F P_F^T \Sigma P_F P_F^T.
\end{align*}
Now
\begin{align*}
\|\Sigma P_F P_F^Te_j\|_2 &\leq \|\Sigma P_F P_F^T\| \|Q Q^T e_j\|_2 + \|\Sigma P_F P_F^T(I-QQ^T)\| \\
&\leq \rho_2 \|\Sigma QQ^T P_F P_F^T\| + \rho_2 \|\Sigma (I-QQ^T) P_F P_F^T\|  \\
&\qquad + \|\Sigma QQ^TP_F P_F^T(I-QQ^T)\| + \|\Sigma(I-QQ^T)P_F P_F^T(I-QQ^T)\| \\
&\leq \rho_2 \|\Sigma QQ^T\| \|P_F P_F^T\| + \rho_2 \sigma_u \|(I-QQ^T) P_FP_F^T\| \\
&\qquad + \|\Sigma QQ^T \| \|P_FP_F^T (I-QQ^T)\| + \sigma_u \|(I-QQ^T) P_F\|^2.
\end{align*}
Note that $\|P_FP_F^T (I-QQ^T)\| = \|(I-QQ^T) P_F\| = \|(I-QQ^T) P_FP_F^T\|$. Thus from \eqref{eq:l_max_3} we have that the RHS of the last display ma be bounded above by a constant times
\[
\rho_2\|\Sigma QQ^T\| + \sigma_u \rho_1\rho_2/\gamma_l + \rho_1^2/\gamma_l + \sigma_u\rho_1^2/\gamma_l^2.
\]
Noting that $\|P_F P_F^T e_j\|_2, \|P_L P_L^T e_j\|_2 \leq 1$, we have by the Cauchy--Schwarz inequality that
\[
\|\Sigma - P_L P_L^T\Sigma P_LP_L^T\|_\infty \lesssim \rho_1 \rho_2  + \rho_1^2/\gamma_l + \sigma_u\rho_1^2/\gamma_l^2
\]
Putting things together we have
\begin{equation}
\|\Sigma - P_LP_L^T\Theta P_LP_L^T\|_\infty \lesssim \rho_1 \rho_2 + \rho_1^2/\gamma_l + \gamma_u \rho_1^2 / \gamma_l^2.
\end{equation}

\section{Proof of Proposition~\ref{prop:spec_trans}}
Let us fix $H$ and write $\hat{\Sigma}_H := \tilde{\Sigma}(X)$, making the dependence on $X$ explicit. Write $X = Z \Theta^{1/2} = ZPDP^T$ where $Z \in \R^{n \times p}$ has i.i.d.\ $\mathcal{N}(0, 1)$ entries.
Now as $Z \eqdist ZP$, we have $\E\tilde{\Sigma}(X) = \E\tilde{\Sigma}(ZDP^T)$.

Let $W$ be the matrix of right singular vectors of $\Pi Z D$. Note that the matrix of right singular vectors of $\Pi ZDP^T$ is $PW$. Also the diagonal matrix of singular values $\Upsilon$ of $ZD$ is the same as that of $ZDP^T$. Thus $\tilde{\Sigma}(ZDP^T) = P \tilde{\Sigma}(ZD) P^T$. It therefore suffices to show that $\E \tilde{\Sigma}(ZD)$ is diagonal.

Consider for $j \neq k$
\begin{equation} \label{eq:W_invar}
\{\tilde{\Sigma}(ZD)\}_{jk} = \sum_{m=1}^n W_{j,m} W_{k, m} H(\Upsilon)_{mm}.
\end{equation}
Now let $\check{Z}$ be a copy of $Z$ but with $j$th column replaced by $-Z_j$. It is straightforward to check that the matrix $\check{W}$ of right singular vectors of $\Pi \check{Z}D$ satisfies $\check{W}_{-j,.} = W_{-j,.}$ and $\check{W}_{j,.}=-W_{j,.}$. Also, the singular vectors of $\Pi \check{Z}D$ are the same as those of $\Pi ZD$. As $\check{Z}D \eqdist ZD$, we have $\check{W} \eqdist W$. In particular from \eqref{eq:W_invar} we see that $\{\tilde{\Sigma}(ZD)\}_{jk} \eqdist -\{\tilde{\Sigma}(ZD)\}_{jk}$, whence $\E(\{\tilde{\Sigma}(ZD)\}_{jk})=0$ as required.

\section{Proof of Proposition~\ref{prop:RSVP_dist}}
Let $Z$ satisfy \eqref{eq:RSVP_dist}, so $X = M Z \Theta^{1/2}$. Denote by $\mathcal{O}(p)$ the set of $p \times p$ orthogonal matrices. Let the SVD of $Z$ be given by $Z=U \Lambda V^T$. Here $U \in \mathcal{O}(p)$, $\Lambda \in \R^{p \times (n+1)}$ and $V \in \R^{p \times (n+1)}$ has orthonormal columns. We claim that $\hat{\Sigma}_{\rsvp}$ depends only on $V$. This follows from the facts that $\hat{\Sigma}_{\rsvp}$ is a projection on to the row space of $X$, and $X$ and $ \Lambda^{-1} U^TM^{-1} X$ has the same row space as $X$.

Next observe that as $Z R \eqdist Z$ for any $R \in \mathcal{O}(p)$, $V$ is uniformly distributed on the Stiefel manifold $V_n(\R^p)$. In particular, the distribution of $V$, on which $\hat{\Sigma}_{\rsvp}$ depends, is uniquely determined  by the fact the $Z$ has a spherically symmetric distribution. Thus we may assume, without loss of generality, that $Z \in \R^{n \times p}$ has i.i.d.\ $\mathcal{N}(0, 1)$ entries, and $M$ is the identity matrix, which gives the required distribution for the rows of $X$.
\section{Some basic results} \label{sec:basic}
The following corollary of Proposition~\ref{prop:RSVP_dist} will be useful in many of our results. It allows us to treat the centred $\Pi X \in \R^{(n+1) \times p}$ as an uncentred $n \times p$ matrix, but with i.i.d.\ Gaussian rows.
\begin{cor} \label{cor:Y_dist}
Suppose the distribution of $X$ satisfies \eqref{eq:RSVP_dist}. Then if $Y \in \R^{n \times p}$ has independent rows distributed as $\mathcal{N}_p(0, \Theta)$, we have that
\[
\hat{\Sigma}_{\rsvp} = Y^T(Y Y^T)^{-1}Y
\]
almost surely.
\end{cor}
\begin{proof}
From Proposition~\ref{prop:RSVP_dist}, we know we may assume that $X$ has independent $\mathcal{N}_p(0, \Theta)$ rows. Let the eigendecomposition of the projection $\Pi$ be $QEQ^T$ where $E$ is diagonal with $E_{n+1,n+1}=0$ and $E_{jj}=1$ for $j \leq n$. Then the row space of $X$ coincides with that of $E Q^T X$. But $Q^T X \eqdist X$, and $E X \eqdist Y$. Projection on to the row space of $Y$ may be written as $Y^T(Y Y^T)^{-1}Y$ when $YY^T$ is invertible, which is the case almost surely.
\end{proof}
In view of this result, we can write
\begin{equation} \label{eq:Sigma_Z_dist}
\hat{\Sigma}_{\rsvp} = PDZ^T(ZD^2Z^T)^{-1}ZDP^T
\end{equation}
where $Z \in \R^{n \times p}$ has i.i.d.\ $\mathcal{N}(0, 1)$ entries and $PD^2P^T$ is the eigendecomposition of $\Theta$. We will adopt representation \eqref{eq:Sigma_Z_dist} in subsequent results without further comment.

The following straightforward consequence of Weyl's inequality will be used for several of the results.
\begin{prop} \label{prop:basic_eval}
The first $q$ eigenvalues of $\Theta$ lie within the interval $[\gamma_u + \sigma_u, \gamma_l + \sigma_l]$ and the remaining eigenvalues lie in $[\sigma_u, \sigma_l]$.
\end{prop}

\section{Proof of Theorem~\ref{thm:conc_Pi}}
We will prove the stronger result, Lemma~\ref{lem:conc_Pi} below. Theorem~\ref{thm:conc_Pi} follows easily using a union bound.
\begin{lem} \label{lem:conc_Pi}
Assume that $\sigma_u \lesssim p / (n \log p)$, $q \lesssim n / \log(p)$ and $p > cn$ for some $c > 1$. Then for any fixed $a,b \in \R^p$
and any fixed $r > 0$, we have that there exist $c_1, c_2 > 0$ with
\[
\pr(|a^T\hat{\Sigma}_{\rsvp} b - \E a^T\hat{\Sigma}_{\rsvp} b| > t) \lesssim \exp\left(-\frac{c_1 t^2p^2}{\|\Theta^{1/2} a \|_2^2 \|\Theta^{1/2} b\|_2^2 n}\right) + e^{-c_2 n} + \frac{1}{np^r}.
\]
for all $t >0$.
\end{lem}
Lemma~\ref{lem:conc_Pi} is proved in Section~\ref{sec:res_final}, though the result relies on several lemmas to be presented in the sequel. Below we outline our proof strategy.

By the decomposition $\hat{\Sigma}_{\rsvp} = PDZ^T(ZD^2Z^T)^{-1}ZDP^T$ with $Z \in \R^{n \times p}$ (see Section~\ref{sec:basic}), it suffices to study the concentration of $u^TDZ^T(ZD^2Z^T)^{-1}ZDw$ for $u,w \in \R^p$.
The map $Z \mapsto u^TDZ^T(ZD^2Z^T)^{-1}ZDw$ is not Lipschitz so we cannot directly apply the Gaussian concentration inequality. However, the function is differentiable almost everywhere and the gradient, which is computed in Lemma~\ref{lem:grad_comp}, is bounded on regions where $Z$ lies with high probability.
As our Theorem~\ref{thm:grad_conc} shows, this is enough to ensure concentration. Much of the work in the proof of Theorem~\ref{thm:conc_Pi} is therefore obtaining a high probability bound on the $\ell_2$-norm of the gradient, so that we may apply Theorem~\ref{thm:grad_conc}.

We begin by deriving the form of the gradient in Section~\ref{sec:grad_comp}, after which we present our variant of the Gaussian concentration inequality. In Section~\ref{sec:grad_bd} we compute a high probability bound on the gradient, and in Section~\ref{sec:res_final} we put things together to obtain the final result.

We will make frequent use of the following notation: $Z_j \in \R^n$ will be the $j$th column of $Z$, and $Z_{-j} \in \R^{n \times (p-1)}$ and $Z_{-jk} \in \R^{n \times (p-2)}$ for $j \neq k$ will be a copies of $Z$ excluding the $j$th, and $j$th and $k$th columns respectively. Also, given a square matrix $M \in \R^{p \times p}$, $M_{-j,-j} \in \R^{(p-1) \times (p-1)}$ and $M_{-jk,-jk} \in \R^{(p-2) \times (p-2)}$ will be copies of $M$ excluding the $j$th, and $j$th and $k$th rows and columns respectively.

\subsection{Gradient computation} \label{sec:grad_comp}

\begin{lem} \label{lem:grad_comp}
Consider the map
\begin{align*}
\psi: \mathcal{D} &\to \R \\
M &\mapsto u^T DM^T(MD^2 M^T)^{-1}M D w.
\end{align*}
where $\mathcal{D} = \{ M \in \R^{n \times p} : \; MD^2M^T \text{ is invertible}\}$ and $p > n$. Write
\begin{align*}
B &=(MD^2M^T)^{-1} M D \\
C &=D\{I-DM^T(MD^2 M^T)^{-1}M D\}.
\end{align*}
Then
\[
\|\nabla \psi(M) \|_2^2  \leq  2\left( \|B^Tu\|_2^2 \|C^Tw\|_2^2 + \|C^Tu\|_2^2 \|B^Tw\|_2^2\right).
\]
\end{lem}
\begin{proof}
A Taylor series expansion gives that when $M D^2 M^T$ is invertible, for $E \in \R^{n \times p}$ with $\|E\|$ sufficiently small we have
\[
\psi(M+E) -\psi(M) = \tilde{E}^TA \tilde{M} + \tilde{M}^TA\tilde{E} - \tilde{M}^TA(\tilde{E}\tilde{M}^T + \tilde{M} \tilde{E}^T)A \tilde{M} + O(\|E\|^2),
\]
where $\tilde{E} := ED$, $\tilde{M} = MD$ and $A = (\tilde{M}\tilde{M}^T)^{-1}$. A straightforward calculation then yields
\[
\nabla \psi(M) = (MD^2M^T)^{-1} M D(wu^T + uw^T)(I-DM^T(MD^2 M^T)^{-1}M D)D.
\]
Thus
\begin{align*}
\|\nabla \psi(M)\|_2^2  &=  \|Bu\|_2^2 \|Cw\|_2^2 + \|Cu\|_2^2 \|Bw\|_2^2 + 2u^TB^TBw u^TC^TCw \\
&\leq \|Bu\|_2^2 \|Cw\|_2^2 + \|Cu\|_2^2 \|Bw\|_2^2 + 2 \|Bu\|_2 \|Cw\|_2 \|Cu\|_2 \|Bw\|_2 \\
&\leq 2\left( \|Bu\|_2^2 \|Cw\|_2^2 + \|Cu\|_2^2 \|Bw\|_2^2\right),
\end{align*}
using the Cauchy--Schwarz inequality and $a^2+b^2 \geq 2ab$ in the penultimate and final lines respectively.
\end{proof}

\subsection{Gaussian concentration} \label{sec:Gauss_conc}
Our variant of the Gaussian concentration inequality is based on the following more classical result that appears in \citet{wainwright_2019}.

\begin{lem} \label{lem:Lipschitz}
Let $W \sim \mathcal{N}_d(0, I)$ and let $f:\R^d \to \R$ be differentiable. Then for any convex function $\phi:\R \to \R$ we have
\[
\E \phi\{f(W) - \E f(W)\}  \leq \E \left\{\phi\left(\nabla f(W)^TV /2\right) \right\}
\]
where $V \sim \mathcal{N}_d(0, I)$ and $V$ is independent of $W$.
\end{lem}

\begin{thm} \label{thm:grad_conc}
Let $W \sim \mathcal{N}_d(0, I)$ and let $f:\R^d \to \R$ be differentiable.
Let $\Psi(L) = \pr(\|\nabla f(W)\|_2^2 > L)$ for $L \geq 0$, and define
\[
\Delta(L, \alpha,t) = e^{\alpha^2 \pi^2 L / 4 -\alpha t} +  \left(\Psi(L) + \frac{\pi \alpha}{2}\sqrt{\Psi(L)\E (\|\nabla f(W)\|_2^2)}\right).
\]
Then for all $\alpha > 0$ we have
\begin{equation} \label{eq:L_mgf_bd}
\E \left\{\exp\left(\alpha \{f(W)-\E f(W)\}\right) \wedge e^{\alpha t} \right\} \leq e^{\alpha t} \inf_{L> 0} \Delta(L, \alpha, t),
\end{equation}
and
\begin{equation} \label{eq:L_prob_bd}
\pr(|f(W) - \E f(W)| > t) \leq 2 \inf_{\alpha, L >0} \Delta(L, \alpha, t).
\end{equation}
In particular, we have that for all $L >0$,
\[
\pr(|f(W) - \E f(W)| > t) \leq 2\exp\left(-\frac{t^2}{\pi^2 L}\right) + 2(1 + L^{-1/2})\big\{L^{-1}\E (\|\nabla f(W)\|_2^2) \,\Psi(L)\big\}^{1/2}.
\]
\end{thm}
\begin{proof}
For each $\alpha, t > 0$ define $\phi_{\alpha,t}: \R \to \R$ by
\[
\phi_{\alpha, t}(w) = e^{\alpha w} \ind_{\{w \leq t\}}   + e^{\alpha t}\{\alpha(w-t) +  1\}\ind_{\{w > t\}}.
\]
Note that $\phi_{\alpha,t}$ is convex for each $(\alpha, t)$ and
\begin{equation} \label{eq:phi_ineq}
e^{\alpha w} \wedge e^{\alpha t} \leq \phi_{\alpha, t}(w) \leq e^{\alpha w}.
\end{equation}
Also we see that $\phi_{\alpha,t}(w) \leq e^{\alpha t}\{\alpha(w-t)_+ +  1\}$.
Thus for any event $A$ and random variable $U$,
\[
\phi_{\alpha, t}(U) \leq e^{\alpha U}\ind_{A} + e^{\alpha t}\{1 + \alpha(U-t)_+\}\ind_{A^c}.
\]
Taking expectations and using the Cauchy--Schwarz inequality we have
\[
\E \phi_{\alpha, t}(U) \leq \E (e^{\alpha U}\ind_{A}) + e^{\alpha t}\{\pr(A^c) + \alpha \sqrt{\E(U^2) \pr(A^c)}\}.
\]
Now let $V \sim \mathcal{N}_d(0, I)$ independently of $W$.
Substituting $U = \pi \nabla f(W)^T V/2 $ and $A = \{\|\nabla f(W)\|_2^2 \leq L\}$, we have
\begin{align*}
\E \phi_{\alpha, t}(\pi \nabla f(W)^T V/2) &\leq \E\{\exp(\alpha \pi \nabla f(W)^T V/2)\ind_{\{\|\nabla f(W)\|_2^2 \leq L\}}\} \\
&\qquad + e^{\alpha t}\left(\Psi(L) + \frac{\pi \alpha}{2}\sqrt{\Psi(L)\E (\|\nabla f(W)\|_2^2)}\right)
\end{align*}
as $\E \{(\nabla f(W)^TV)^2 | W\} =  \|\nabla f(W)\|_2^2$. Considering the first term on the RHS in the last display, we have
\begin{align*}
\E\{\exp(\alpha \pi \nabla f(W)^T V/2)\ind_{\{\|\nabla f(W)\|_2^2 \leq L\}}|W\} &= \exp(\alpha^2 \pi^2 \|\nabla f(W)\|_2^2 / 4 )\ind_{\{\|\nabla f(W)\|_2^2 \leq L\}} \\
&\leq \exp(\alpha^2 \pi^2 L / 4 ).
\end{align*}
Putting things together we have
\begin{equation} \label{eq:E_phi_bd}
\E \phi_{\alpha, t}(\pi \nabla f(W)^T V/2) \leq \exp(\alpha^2 \pi^2 L / 4 ) +  e^{\alpha t}\left(\Psi(L) + \frac{\pi \alpha}{2}\sqrt{\Psi(L)\E (\|\nabla f(W)\|_2^2)}\right).
\end{equation}

Now observe that
\begin{align*}
e^{\alpha t}\ind_{\{f(W) - \E f(W) > t\}} \leq e^{\alpha\{f(W) - \E f(W)\}} \wedge e^{\alpha t}.
\end{align*}
Taking expectations and using the first inequality of \eqref{eq:phi_ineq}, we obtain
\[
e^{\alpha t} \pr(f(W) - \E f(W) > t) \leq \E(e^{\alpha\{f(W) - \E f(W)\}} \wedge e^{\alpha t}) \leq \E \phi_{\alpha, t}(f(W) - \E f(W))
\]
for all $\alpha > 0$.
Applying Lemma~\ref{lem:Lipschitz} with $\phi = \phi_{\alpha, t}$ and then using \eqref{eq:E_phi_bd}, we see that the RHS of the last display is bounded above by
\[
 \inf_{L >0}\left\{\exp(\alpha^2 \pi^2 L / 4) +  e^{\alpha t}\left(\Psi(L) + \frac{\pi \alpha}{2}\sqrt{\Psi(L)\E (\|\nabla f(W)\|_2^2)}\right) \right\}.
\]
This gives \eqref{eq:L_mgf_bd} and after repeating the argument replacing $f(W)$ with $-f(W)$ and using a union bound we also get \eqref{eq:L_prob_bd}. For the last inequlity, we argue as follows.
Dividing by $e^{\alpha t}$ and setting $\alpha = 2t/(\pi^2 L)$ we arrive at
\[
\pr(f(W) - \E f(W) > t) \leq \exp\left(-\frac{t^2}{\pi^2 L}\right) + \Psi(L) + \frac{t}{\pi L} \sqrt{\Psi(L)\E (\|\nabla f(W)\|_2^2)}.
\]
Now observe that if $t^2 / (\pi^2 L) > 1$ then the first term on the RHS above exceeds 1, so in fact the following holds:
\[
\pr(f(W) - \E f(W) > t) \leq \exp\left(-\frac{t^2}{\pi^2 L}\right) + \Psi(L) + \frac{1}{L}\sqrt{\Psi(L)\E (\|\nabla f(W)\|_2^2)}.
\]
Then noting that by Markov's inequality $\E (\|\nabla f(W)\|_2^2) /L \geq \Psi(L)$, we get
\[
\pr(f(W) - \E f(W) > t) \leq \exp\left(-\frac{t^2}{\pi^2 L}\right) + (L^{-1/2} + 1)\sqrt{\Psi(L)\E (\|\nabla f(W)\|_2^2/L)}.
\]
Repeating the argument replacing $f(W)$ with $-f(W)$ and using a union bound gives the final result.
\end{proof}

\subsection{Bounding the gradient} \label{sec:grad_bd}
We bound the two terms $\|(ZD^2Z^T)^{-1}Z D(u+w)\|_2^2$ and $ \|D\{I-DZ^T(ZD^2Z^T)^{-1}ZD\}(u+w)\|_2^2$ involved in the gradient (see Lemma~\ref{lem:grad_comp}) separately. The following standard result from random matrix theory \citep{Vershynin2010} will be used repeatedly.
\begin{lem} \label{lem:eval}
Let $W \in \R^{n \times d}$ have independent $\mathcal{N}(0, 1)$ entries.
For all $t>0$, with probability at least $1-2e^{-t^2/2}$,
\[
\sqrt{d} - \sqrt{n} - t \leq \lambda_{\min}^{1/2}(WW^T) \leq \lambda_{\max}^{1/2}(WW^T) \leq \sqrt{d} + \sqrt{n} + t.
\]
\end{lem}

\begin{lem} \label{lem:gradbd_1}
Consider the setup of Theorem~\ref{thm:conc_Pi}.
We have that with probability at least $1 - c_1 e^{-c_2 n}$, for each fixed $v \in \R^d$,
\[
\|(ZD^2Z^T)^{-1}ZDv\|_2^2 \lesssim \frac{n}{p^2} \|Dv\|_2^2.
\]
\end{lem}
\begin{proof}
We have
\begin{align*}
\|(ZD^2Z^T)^{-1}ZDv\|_2^2 &\leq \|ZDv\|_2^2 \{\lambda_{\max}((ZD^2Z^T)^{-1})\}^2 \\
&\leq \|ZDv\|_2^2 \{\lambda_{\min}(D^2)\}^{-2} \{\lambda_{\min}(ZZ^T)\}^{-2}.
\end{align*}
Taking $t=p(1-c^{-1/2})/4$ in Lemma~\ref{lem:eval} (recall that $p > cn$ with $c > 1$), we know that with probability at least $1-e^{-c_1 p}$, we have $\lambda_{\min}^{1/2}(ZZ^T) \gtrsim \sqrt{p}$. Also, each component of $ZDv$ is independent and distributed as $\mathcal{N}(0, \|Dv\|_2^2)$. Thus from Lemma~\ref{lem:gen_chi_sq} we have
\[
\pr\left( \|ZDv\|_2^2 \gtrsim n\|Dv\|_2^2\right) \leq e^{-c_2 n}.
\]
Putting things together, we see that
\[
\|(ZD^2Z^T)^{-1}ZDv\|_2^2 \lesssim \frac{n}{p^2} \|Dv\|_2^2
\]
with probability at least $1 - c_1 e^{-c_2 n}$.
\end{proof}

\begin{lem} \label{lem:gradbd_2}Consider the setup of Theorem~\ref{thm:conc_Pi}.
There exist positive constants $c_1, c_2$ such that for all $t>0$ and fixed $v \in \R^d$, with probability at least $1-pe^{-t^2/2} - c_1pe^{-c_2n}$,
\[
\|D\{I-DZ^T(ZD^2Z^T)^{-1}ZD\}v\|_2^2 \lesssim \|Dv\|_2^2 + \frac{nt^2}{p^2}\|Dv\|_2^2 \sum_j \min\left(D_{jj}^4, \frac{p^2}{n^2} \right).
\]
\end{lem}
\begin{proof}
Consider the $j$th component $b_j$ of $D\{I-DZ^T(ZD^2Z^T)^{-1}ZD\}v$.
Appealing to the Sherman--Morrison formula, we see that
\begin{align}
Z_j^T(ZD^2Z^T)^{-1} &= Z_j^T(Z_{-j}D_{-j,-j}^2Z_{-j}^T)^{-1} - \frac{D_{jj}^2 A_j^2}{1 + D_{jj}^2 A_j^2} Z_j^T(Z_{-j}D_{-j,-j}^2Z_{-j}^T)^{-1} \notag\\
&= \frac{1}{1 + D_{jj}^2 A_j^2} Z_j^T(Z_{-j}D_{-j,-j}^2Z_{-j}^T)^{-1}, \label{eq:Sherman-Morrison}
\end{align}
where $A_j^2 = Z_j^T(Z_{-j}D_{-j,-j}^2Z_{-j}^T)^{-1}Z_j$.
Thus we have
\[
b_j=D_{jj}v_j - D_{jj} \frac{D_{jj}Z_j^T(Z_{-j}D_{-j,-j}^2Z_{-j}^T)^{-1}ZDv}{1 + D_{jj}^2 A_j^2}.
\]
Next, writing $ZDV = Z_j D_{jj}v_j + Z_{-j}D_{-j,-j}v_{-j}$, we have
\[
b_j =  \frac{v_j D_{jj}}{1 + D_{jj}^2A_j^2} - \frac{D_{jj}^2Z_{j}^T(Z_{-j}D_{-j,-j}^2Z_{-j}^T)^{-1}Z_{-j}D_{-j,-j}v_{-j}}{1 + D_{jj}^2A_{j}^2} =: \text{I} - \text{II}.
\]
Considering the numerator of $\text{II}$, observe that
\[
Z_{j}^T(Z_{-j}D_{-j,-j}^2Z_{-j}^T)^{-1}Z_{-j}D_{-j,-j}v_{-j} | Z_{-j} \sim \mathcal{N}(0, \|(Z_{-j}D_{-j,-j}^2Z_{-j}^T)^{-1}Z_{-j}D_{-j,-j}v_{-j}\|_2^2).
\]
Thus with probability $1 - pe^{-t^2/2}$,
\begin{align*}
|Z_{j}^T(Z_{-j}D_{-j,-j}^2 Z_{-j}^T)^{-1}Z_{-j}D_{-j,-j}v_{-j}| &\leq t \|(Z_{-j}D_{-j,-j}^2Z_{-j}^T)^{-1}Z_{-j}D_{-j,-j}v_{-j}\|_2
\end{align*}
for all $j$.
From Lemma~\ref{lem:gen_chi_sq} and a union bound, we have
\[
 \pr(\|Z_{-j}D_{-j,-j}v_{-j}\|_2^2  \leq n(1+r) \|Dv\|_2^2) \geq 1 - pe^{nr^2/8}
\]
for $r \in (0, 1)$ and all $j$. Furthermore, Lemma~\ref{lem:eval} gives
\[
\lambda_{\max}\{(Z_{-j}D_{-j,-j}^2Z_{-j}^T)^{-1}\} \leq \lambda_{\min}^{-1}(D^2) \lambda_{\min}^{-1}(Z Z^T) \lesssim \frac{1}{p}\lambda_{\min}^{-1}(D^2)
 \]
with probability at least $1-e^{-cp}$.
Thus, we have that with probability at least $1-pe^{-t^2/2} - c_1pe^{-c_2n}$,
\[
|Z_{j}^T(Z_{-j}D_{-j,-j}^2Z_{-j}^T)^{-1}Z_{-j}D_{-j,-j}v_{-j}| \lesssim  t \frac{\sqrt{n}}{p}\|Dv\|_2\lambda_{\min}^{-1}(D^2)
\]
for all $j$.

We see from Lemma~\ref{lem:A_j_bd} that with probability at least $1-c_1pe^{-c_2 n}$, $A_j^2 \gtrsim  n/p$ for all $j$.
Thus with probability at least $1-c_1pe^{-c_2 n}$ we have
\[
\frac{D_{jj}^2}{1 + D_{jj}^2 A_j^2} \lesssim \min(D_{jj}^2, p/n).
\]

Putting things together we have that with probability at least $1-pe^{-t^2/2} - c_1pe^{-c_2n}$,
\[
| \text{II}| \lesssim \frac{t\sqrt{n}}{p} \|Dv\|_2 \min\left(D_{jj}^2, \frac{p}{n} \right).
\]
for all $j$.
Squaring and summing over $j$ we get
\[
\|b\|_2^2 \lesssim \|Dv\|_2^2 + \frac{nt^2}{p^2}\|Dv\|_2^2 \sum_j \min\left(D_{jj}^4, \frac{p^2}{n^2} \right).
\]
with probability at least $1-pe^{-t^2/2} - c_1pe^{-c_2 n}$.
\end{proof}

\begin{lem} \label{lem:A_j_bd}
Consider the setup of Lemma~\ref{lem:conc_Pi}.
Let $A_j^2 = Z_j^T (Z_{-j}D_{-j,-j}^2Z_{-j}^T)^{-1}Z_j$.
There exists $c_1>0$ such that for all $j$ we have $c_1^{-1}n/p>\E(A_j^2) \geq n / \tr(D^2) > c_1n/p$. Furthermore,
\[
\pr(\max_j |A_j^2 -\E A_j^2| > t \E A_j^2 ) \lesssim pe^{-c_3 n t^2} + e^{-c_4 p}.
\]
\end{lem}
\begin{proof}
We first bound $A_j^2$ from below. We have by Jensen's inequality
\begin{align*}
\E A_j^2 &= \E\tr \{(Z_{-j}D_{-j,-j}^2 Z_{-j}^T)^{-1}\} \\
&\geq \tr \{\E(Z_{-j}D_{-j,-j}^2Z_{-j}^T)\}^{-1} \\
&= \frac{n}{\tr(D_{-j,-j}^2)} \geq \frac{n}{\tr(D^2)}.
\end{align*}
Next, from Lemma~\ref{lem:mat_ineq1} we have $(Z_{-j}D_{-j,-j}^2 Z_{-j}^T)^{-1} \preceq \sigma_l^{-1}(Z_{-j} Z_{-j}^T)^{-1}$.
The first part of the result then follows from applying the formula for the mean of an inverse Wishart distribution.

Next we apply Lemma~\ref{lem:A_j_conc} taking $W=Z_{-j}$, $\Lambda = D_{-j,-j}^2$, which easily yields the final result.
\end{proof}

\subsection{Proof of Lemma~\ref{lem:conc_Pi}} \label{sec:res_final}
Let us write $Pb = u$ and $Pa=w$.
In order to apply Theorem~\ref{thm:grad_conc} we need an upper bound on the expectation of the $\ell_2$-norm squared of the gradient; such a bound is given by Lemma~\ref{lem:grad_comp} as
\begin{equation} \label{eq:exp_grad}
2 \E \left(\|B u\|_2^2 \|Cw\|_2^2 + \|C u\|_2^2 \|B w \|_2^2 \right) =: \text{I},
\end{equation}
where
\begin{align*}
B &=(ZD^2Z^T)^{-1} Z D \\
C &=D(I-DZ^T(ZD^2 Z^T)^{-1}Z D).
\end{align*}
Now we have
\begin{align*}
\|Cu\|_2^2 &\leq \lambda_{\max}(D^2) \|\{I-DZ^T(ZD^2Z^T)^{-1}ZD\}u\|_2^2 \\
&\leq \lambda_{\max}(\Theta)\|u\|_2^2,
\end{align*}
using the fact that $I-DZ^T(ZD^2Z^T)^{-1}ZD$ is a projection in the last line.
It thus remains to bound $\E \|B w\|_2^2 $. Note that by Lemma~\ref{lem:mat_ineq1},
\begin{align*}
\|B w\|_2^2 \leq \lambda_{\min}^{-2}(D^2)\|(ZZ^T)^{-1}Z Dw\|_2^2.
\end{align*}
Next observe that
\[
(ZZ^T)^{-1}Z v \eqdist (ZZ^T)^{-1}ZRv
\]
for any orthogonal matrix $R \in \R^{p \times p}$. By choosing a rotation on to the $j$th unit vector $e_j$, we see that
\[
\|(ZZ^T)^{-1}Z Dw\|_2^2 \eqdist e_j^TZ^T(ZZ^T)^{-2}Ze_j \|Dw\|_2^2.
\]
Since this holds for all $j$, we have
\[
\E\|(ZZ^T)^{-1}Z D w\|_2^2 = \frac{1}{p} \E \{\tr(ZZ^T)^{-1}\}\|Dw\|_2^2 = \frac{1}{p-n-1}\|Dw\|_2^2,
\]
using the formula for the mean of an inverse Wishart distribution.

Putting things together we have that
\[
\text{I} \lesssim \frac{\lambda_{\max}(\Theta)(\|u\|_2^2 \|D w\|_2^2  +\|w\|_2^2 \|D u\|_2^2)}{p-n-1} \lesssim \|\Theta^{1/2} a \|_2^2 \|\Theta^{1/2} b \|_2^2
\]
using that $p \gtrsim \tr(\Theta) \gtrsim \lambda_{\max}(\Theta)$, and $\lambda_{\min}(\Theta) \geq \sigma_l \gtrsim 1$.

We can now apply Theorem~\ref{thm:grad_conc}. Adopting the notation from that result, let us take $L$ to be the bound for
\begin{equation} \label{eq:grad_bound}
\|B u\|_2^2 \|Cw\|_2^2 + \|C u\|_2^2 \|B w \|_2^2
\end{equation}
given by 
sum of the products of the bounds from Lemmas~\ref{lem:gradbd_1} and~\ref{lem:gradbd_2}. The latter bound requires a choice of $t^2$ which we take as $c_1 \log(p)$ for a suitable constant $c_1>0$. With this choice, we have with high probability that
\begin{align*}
\|Cw\|_2^2 \lesssim \|D w \|_2^2  + \frac{n \log(p)\|D w\|_2^2}{p^2}\left(q \frac{p^2}{n^2} + \sum_{j : D_{jj}^2 \leq \sigma_u} D_{jj}^4 \right).
\end{align*}
Now the sum on the right is maximised when as many of the $D_{jj}^4$ as possible take the value $\sigma_u^2$ subject to $\sum_j D_{jj}^2 \lesssim p$; thus is it bounded above by a constant times $p \sigma_u$. Using the facts that $\sigma_u \lesssim p/(n \log(p))$, $q \lesssim n / \log(p)$, we see that with high probability $\|Cw\|_2^2 \lesssim \|D w \|_2^2$. Our bound for \eqref{eq:grad_bound} thus takes the form a constant times
\[
\frac{n}{p^2} \|D u \|_2^2 \|D w\|_2^2 = \frac{n}{p^2} \|\Theta^{1/2} a \|_2^2 \|\Theta^{1/2} b\|_2^2.
\]
We may therefore take $L$ equal to the above multiplied by a constant $c_2 > 0$ to obtain
\[
(1 + L^{-1/2}) \sqrt{\text{I} \cdot \Psi(L) / L} \lesssim \frac{1}{np^r} + e^{-c_3 n}
\]
for any fixed $r>0$ (by taking $c_1$ sufficiently large).
Applying Theorem~\ref{thm:grad_conc} gives the result.

\subsection{Auxiliary lemmas}
\begin{lem} \label{lem:A_j_conc}
Let $W \in \R^{n \times d}$ and $u \in \R^n$ have independent $\mathcal{N}(0, 1)$ entries, and suppose $d \geq c_1n$ with $c_1 > 1$. Let $\Lambda \in \R^{d \times d}$ be a symmetric positive definite matrix with $\lambda_{\min}(\Lambda)$ bounded away from 0, $\lambda_{\max}(\Lambda)$ bounded above by $c_2 d$ for some constant $c_2>0$. Then
\[
\pr\left(|u^T(W\Lambda W^T)^{-1}u -\E\{\tr(W\Lambda W^T)^{-1}\}| > t\right) \lesssim e^{-c_3t^2d^2/n} + e^{-c_4 d}.
\]
with all constants depending only on $\lambda_{\max}(\Lambda), \lambda_{\min}(\Lambda)$ and $c_1$.
\end{lem}
\begin{proof}
Let us write $(W \Lambda W^T)^{-1} = A$, and let $\lambda \in \R^n$ be the vector of eigenvalues of $A$. We may assume, without loss of generality, that $\Lambda$ is diagonal as $W \eqdist WR$ for all orthogonal matrices $R$. Note that $\E(u^T A u|A) = \tr(A)$.
The standard Chernoff method gives us that
\begin{align} \label{eq:Chernoff}
\pr\left(u^TAu -\E \tr(A) \geq t |A \right) \leq \E [\exp \{\alpha(u^TA u - \E \tr(A) -t)\}|A] \wedge 1
\end{align}
for all $\alpha > 0$.
Note that conditional on $A$, $u^T A u \eqdist \sum \lambda_j u_j^2$, a weighted sum of independent $\chi^2$ random variables.
Thus using Lemma~\ref{lem:gen_chi_sq} we have that the RHS of the last display is bounded above by
\[
\left( \exp\left(2\alpha^2 \|\lambda\|_2^2 \right) \exp\{ \alpha(\tr A - \E\tr A)\} e^{-\alpha t} \vee \ind_{\{|\alpha| > \|\lambda\|_\infty^{-1} /4\}}\right) \wedge 1.
\]
As $\ind_{\{|\alpha| > \|\lambda\|_\infty^{-1}/4\}}\ind_{\{\|\lambda\|_\infty \leq s\}} = 0$ when $|\alpha| \leq s^{-1}/4$ we see the above display is in turn is bounded above by
\[
\exp\left(2\alpha^2 n s^2 \right) \exp\{ \alpha(\tr A - \E\tr A)\} e^{-\alpha t} \wedge 1 + \ind_{\{\|\lambda\|_\infty > s\}}
\]
provided $|\alpha| \leq s^{-1}/4$.

With a view to applying Lemma~\ref{lem:Lipschitz} to bound the expectation of the first term, observe that if $f(W) = \tr(W\Lambda W^T)^{-1}$, then
\[
\nabla f(W) = -2(W \Lambda W^T)^{-2} W \Lambda.
\]
Thus
\begin{align}
\|\nabla f (W) \|_2^2 &=  4\tr\{(W\Lambda W^T)^{-4}(W\Lambda^2 W^T)\} \notag\\
&\leq 4 \lambda_{\max}(\Lambda) \tr\{(W\Lambda W^T)^{-3}\} \notag\\ \label{eq:grad_bd}
&\leq 4 \lambda_{\max}(\Lambda) \lambda_{\min}^{-3}(\Lambda) \tr\{(WW^T)^{-3}\}.
\end{align}
Expectations of moments of inverse Wishart distributions are computed in \citet{Dietrich_von_Rosen1988-nc}. From here we have that
\[
\E\tr\{(WW^T)^{-3}\} = \frac{8n}{(d-n+1)(d-n-1)(d-n-3)} \lesssim \frac{n}{d^3}
\]
whence
\[
\E \|\nabla f (W) \|_2^2 \lesssim \frac{n}{d^2}.
\]
From Lemma~\ref{lem:eval} we know that with probability at least $1-c_2e^{-c_2 d}$, we have $\lambda_{\min}^{1/2}(WW^T) \gtrsim \sqrt{d}$.
Using $\tr\{(WW^T)^{-3}\} \leq n \lambda_{\min}^{-3}(WW^T)$, we have from \eqref{eq:grad_bd} that
\begin{align*}
\pr(\|\nabla f (W) \|_2^2 \gtrsim n/d^2) \lesssim e^{-c_1 d}.
\end{align*}
Applying \eqref{eq:L_mgf_bd} of Theorem~\ref{thm:grad_conc} we have
\begin{align*}
\E\left[e^{2\alpha^2 n s^2 - \alpha t} \exp\{ \alpha(\tr A - \E\tr A)\} \wedge 1 \right] \lesssim e^{2\alpha^2 n s^2 - \alpha t} \exp( c_1 \alpha^2  n/d^2) + e^{-c_2 d}.
\end{align*}
Next
\[
\|\lambda\|_\infty = \lambda_{\max}\{(W\Lambda W^T)^{-1}\} \leq \lambda_{\min}^{-1}(\Lambda) \lambda_{\min}^{-1}(WW^T),
\]
so Lemma~\ref{lem:eval} gives
\[
\pr(\|\lambda\|_\infty > c_1/d) \lesssim e^{-c_2 d}
\]
for some $c_1, c_2 > 0$.
Taking $s=c_1/d$, and returning to \eqref{eq:Chernoff} we have
\begin{align*}
\pr(u^TAu - \E\tr(A)\geq t) &\leq \inf_{0<\alpha<c_1d}\exp(c_2 \alpha^2 n/d^2 - \alpha t) + c_3e^{-c_4d} \\
&\leq e^{-c_5t^2d^2/n} + c_2 e^{-c_3 d}.
\end{align*}
We now repeat the argument replacing $u^TAu -\E \tr(A)$ by $\E \tr(A)-u^TAu$ in \eqref{eq:Chernoff}.
\end{proof}

\begin{lem} \label{lem:gen_chi_sq}
Let $W \sim \mathcal{N}_d(0, I)$ and let $a \in (0,\infty)^d$ be a vector. Then
\begin{align*}
\E \bigg[ \exp\bigg\{\alpha \bigg(\sum_{j=1}^d a_jW_j^2 - \|a\|_1\bigg)\bigg\}\bigg] &\leq \exp(2\alpha^2 \|a\|_2^2), \\
\pr\bigg( \bigg|\sum_{j=1}^d a_jW_j^2 - \|a\|_1\bigg| \geq t\|a\|_2^2 \bigg) &\leq 2e^{-\|a\|_2^2 t^2/8}
\end{align*}
for $|\alpha| \leq \|a\|_\infty^{-1}/4$ and $0<t< \|a\|_\infty^{-1}$.
\end{lem}
\begin{proof}
Using the facts that $\E e^{\alpha W_1^2} = (1-2\alpha)^{-1/2}$ for $\alpha < 1/2$ and $e^{-\alpha}/\sqrt{1-2\alpha}\leq e^{2\alpha^2}$ for $|\alpha|<1/4$, we have
\begin{align*}
\E \bigg[ \exp\bigg\{\alpha \bigg(\sum_{j=1}^d a_jW_j^2 - \|a\|_1\bigg)\bigg\}\bigg] &= \prod_{j=1}^d \{(1-2\alpha a_j)^{-1/2}e^{-\alpha a_j}\} \leq e^{2\alpha^2 \|a\|_2^2}
\end{align*}
for $|\alpha|\|a\|_\infty  \leq 1/4$.
The final bound follows easily by the Chernoff method.
\end{proof}

\begin{lem} \label{lem:mat_ineq1}
Let $A \in \R^{n \times d}$ and $B \in \R^{d \times d}$ be a symmetric positive definite matrix. Suppose $AA^T$ is invertible. Then for $k > 0$ we have
\[
(ABA^T)^{-k} \preceq \lambda_{\min}(B)^{-k}(AA^T)^{-k}.
\]
\end{lem}
\begin{proof}
Let the SVD of $A$ be given by $A=UDV^T$ where $U \in \R^{n \times n}$, $D \in \R^{n \times n}$ and $V \in \R^{d \times n}$. Then $D^{-k}(V^T B V)^{-k} D^{-k}  \preceq \lambda_{\min}^{-k}(B) D^{-2k}$ as $\lambda_{\min}(B) = \lambda_{\min}(V^T B V)$. But then
\[
(ABA^T)^{-k} = U D^{-k}(V^T B V)^{-k} D^{-k}U^T \preceq \lambda_{\min}^{-k}(B)U  D^{-2k} U^T = \lambda_{\min}(B)^{-k}(AA^T)^{-k}. \qedhere
\]
\end{proof}

\section{Proof of Theorem~\ref{thm:exp_Pi} and derivation of \eqref{eq:exact1}} \label{sec:exp_Pi}
From \eqref{eq:Sigma_Z_dist} and Proposition~\ref{prop:spec_trans}, we know that
\[
\E \hat{\Sigma}_{\rsvp} = \E PDZ^T(ZD^2Z^T)^{-1}ZDP^T=: PC^2P^T.
\]
Here $Z \in \R^{n\times p}$ has i.i.d.\ $\mathcal{N}(0, 1)$ entries and $C$ is diagonal.
In what follows, we will make frequent use of the following notation: $Z_j \in \R^n$ will be the $j$th column of $Z$, and $Z_{-j} \in \R^{n \times (p-1)}$ and $Z_{-jk} \in \R^{n \times (p-2)}$ for $j \neq k$ will be a copies of $Z$ excluding the $j$th, and $j$th and $k$th columns respectively. Also, given a square matrix $M \in \R^{p \times p}$, $M_{-j,-j} \in \R^{(p-1) \times (p-1)}$ and $M_{-jk,-jk} \in \R^{(p-2) \times (p-2)}$ will be copies of $M$ excluding the $j$th, and $j$th and $k$th rows and columns

From \eqref{eq:Sherman-Morrison} we have that
\begin{equation} \label{eq:C_j}
C^2_{jj} = \E \left(\frac{D_{jj}^2A_j^2}{1 + D_{jj}^2A_j^2}\right),
\end{equation}
where $A_j^2 = Z_{j}^T(Z_{-j}D_{-j,-j,}^2Z_{-j}^T)^{-1}Z_j$.
Using the inequality $(1+x)^{-1} \geq 1-x$ we have
\begin{align} \label{eq:C_j_bd_all}
D_{jj}^2 \E A_j^2 - D_{jj}^4 \E A_j^4 \leq C_{jj}^2 \leq D_{jj}^2 \E A_j^2.
\end{align}
We note that Lemma~\ref{lem:A_j4bd} below shows $\E A_j^4 \lesssim n^2/p^2$ uniformly in $j$. For $j \geq q+1$ Proposition~\ref{prop:basic_eval} gives us that $D_{jj}^2 \leq \sigma_u$. Thus we have that
\begin{equation} \label{eq:C_j_bd}
\E A_j^2 - \sigma_u \frac{n^2}{p^2} \lesssim \frac{C_{jj}^2}{D_{jj}^2} \leq \E A_j^2
\end{equation}
for all $j \geq q$.

\subsection{Proof of Theorem~\ref{thm:exp_Pi}}
We now turn to the proof of Theorem~\ref{thm:exp_Pi}. For $j,k \geq q+1$ we have from \eqref{eq:C_j_bd} that
\[
\frac{C_{jj}^2}{D_{jj}^2} - \frac{C_{kk}^2}{D_{kk}^2}  \lesssim \E A_j^2 - \E A_k^2 + \sigma_u \frac{n^2}{p^2}.
\]
Nw
\begin{align*}
\E (A_j^2 - A_k^2) = \E\tr\{(Z_{-j}D_{-j,-j}^2Z_{-j}^T)^{-1} - (Z_{-k}D_{-k,-k}^2Z_{-k}^T)^{-1}\}.
\end{align*}
Let us write $A = Z_{-jk}D_{-jk,-jk}^2Z_{-jk}^T$.
The Sherman--Morrison formula gives us that
\begin{align*}
\E \tr\{(Z_{-j}D_{-j,-j}^2Z_{-j}^T)^{-1} - (Z_{-k}D_{-k,-k}^2Z_{-k}^T)^{-1}\} &\leq \max_j \E \left(\frac{D_{jj}^2 Z_j^T A^{-2} Z_j}{1 + D_{jj}^2 Z_j^T A^{-1} Z_j}\right) \\
&\leq \sigma_u \E \tr(A^{-2}).
\end{align*}
By Lemma~\ref{lem:mat_ineq1}, $(Z_{-jk}D_{-jk,-jk}^2Z_{-jk}^T)^{-2} \preceq \sigma_{l}^{-2} (Z_{-jk} Z_{-jk}^T)^{-2}$, so $\tr(A^{-2}) \lesssim n/p^2$ using the formula for the second moment of an inverse Wishart \citep{Dietrich_von_Rosen1988-nc}. Putting things together we have
\[
\max_{j,k \in \{q+1,\ldots, p\}} \abs{\frac{C_{jj}^2}{D_{jj}^2} - \frac{C_{kk}^2}{D_{kk}^2}} \lesssim \sigma_u \frac{n^2}{p^2}.
\]

\subsection{Derivation of \eqref{eq:exact1}}
We must bound the expectation of $\tr \{(Z_{-j} D_{-j,-j}^2 Z_{-j}^T)^{-1}\}$ from above and below.

By Proposition~\ref{prop:basic_eval}, $D^2$ has its first $q$ diagonal entries in $[\gamma_l - \sigma_u, \gamma_u + \sigma_u]$ with the remaining diagonal entries in $[\sigma_l, \sigma_u]$. Let us fix $j \geq q+1$. To simply notation, let us write $D^2_{-j,-j} = \Lambda$ and let $\Lambda_F$ be the $q \times q$ top left submatrix of $D^2_{-j,-j,}$ , and let $\Lambda_L$ be the bottom right submatrix containing the remaining entries of $D^2_{-j,-j}$. Let us also write $W=Z_{-j}$ and $W_F \in \R^{n \times q}$ for the submatrix of $W$ consisting of the first $q$ columns of $Z_{-j}$, and let $W_L$ be the remaining $p-q-1$ columns.
We may decompose $W \Lambda W^T$ as follows.
\begin{align*}
(W \Lambda W^T)^{-1} &= (W_L \Lambda_L W_L^T)^{-1} - (W_L \Lambda_L W_L^T)^{-1}W_F\{\Lambda_F^{-1} + W_F^T (W_L \Lambda_L W_L^T)^{-1} W_F\}^{-1} W_F^T (W_L \Lambda_L W_L^T)^{-1} \\
&=: \text{I} - \text{II}.
\end{align*}
Now by Lemma~\ref{lem:eval_conc}, we know there exist constants $c_1$ and $c_2$ depending only on $\sigma_l$ and $\sigma_u$ such that with probability $1-e^{-c_1n}$
\begin{align} \label{eq:eig_bd}
\tr(\Lambda_L) - c_2\sqrt{n\tr(\Lambda_L)} \leq \lambda_{\min}(W\Lambda W^T) \leq \lambda_{\max}(W\Lambda W^T) \leq \tr(\Lambda_L) + c_2\sqrt{n\tr(\Lambda_L)}.
\end{align}
For all $r>0$, we have
\begin{align} \label{eq:eig_arg}
\tr\{(W_L \Lambda_L W_L^T)^{-1}\} \leq \frac{n}{r} \ind_{\{\lambda_{\min}(W_L \Lambda_L W_L^T) \geq r\}} +  \E[\tr\{(W_L \Lambda_L W_L^T)^{-1}\}]\ind_{\{\lambda_{\min}(W_L \Lambda_L W_L^T) < r\}}.
\end{align}
Taking expectations, setting $r=\tr(\Lambda_L) - c_2\sqrt{n\tr(\Lambda_L)}$ and using the Cauchcy--Schwarz inequality for the second term, we have
\begin{align*}
\E\tr\{(W_L \Lambda_L W_L^T)^{-1}\} \leq \frac{n}{\tr(\Lambda_L) - c_1\sqrt{n\tr(\Lambda_L)}} + \sqrt{\E\left([\tr\{(W_L \Lambda_L W_L^T)^{-1}\}]^2\right)} c_2 e^{-c_3 n}.
\end{align*}
Now
\begin{equation} \label{eq:I_bd1}
\E\left([\tr\{(W_L \Lambda_L W_L^T)^{-1}\}]^2\right) \leq \sigma_l^{-2}\E [\{\tr(\Omega)\}^2]
\end{equation}
where $\Omega$ has an inverse Wishart with $p-q-1$ degrees of freedom. From \citet{Dietrich_von_Rosen1988-nc} we have
\begin{equation} \label{eq:I_bd2}
\E [\{\tr(\Omega)\}^2] \leq n^2 \E(\Omega_{11}^2) = \frac{4n^2}{(p-q-n)(p-q-n-2)}.
\end{equation}
By Jensen's inequality we also have the lower bound
\[
\E\tr\{(W_L \Lambda_L W_L^T)^{-1}\} \geq \tr[ \{\E(W_L \Lambda_L W_L^T)\}^{-1}] = \frac{n}{\tr(\Lambda_L)}.
\]
Putting things together we obtain
\begin{align*}
\frac{n}{\tr(\Lambda_L)} \leq \E\tr\{(W_L \Lambda_L W_L^T)^{-1}\} \leq \frac{n}{\tr(\Lambda_L)}\{1 + O(\sqrt{n/p})\},
\end{align*}
using the fact that $\tr(\Lambda_L) \geq (p-q-2)\sigma_l \gtrsim p$.
We now turn to $\text{II}$. We have
\begin{align*}
\tr(\text{II}) &\leq \lambda_{\max}^2\{(W_L \Lambda_L W_L^T)^{-1}\} \tr[W_F\{\Lambda_F^{-1}+W_F^T(W_L \Lambda_L W_L^T)^{-1}W_F\}^{-1} W_F^T] \\
&\leq \frac{\lambda_{\max}(W_L \Lambda_L W_L^T)}{\lambda_{\min}^2(W_L \Lambda_L W_L^T)}q
\end{align*}
using Lemma~\ref{lem:tr_bd2}. Note that
$\text{II} \leq \tr\{(W_L \Lambda_LW_L^T)^{-1}\}$, so \eqref{eq:I_bd1} and \eqref{eq:I_bd2} provide an upper bound on $\E(\text{II}^2)$.
Considering events on which the inequality \eqref{eq:eig_bd} occurs and arguing as in \eqref{eq:eig_arg}, we may then arrive at
\[
\E\tr(\text{II}) \leq \frac{q}{\tr(\Lambda_L)}\{1 + O(\sqrt{n/p})\}.
\]
For the lower bound we again appeal to Lemma~\ref{lem:tr_bd2} to obtain
\begin{align*}
\tr(\text{II}) &\geq \lambda_{\min}^2\{(W_L \Lambda_L W_L^T)^{-1}\} \tr[W_F\{\Lambda_F^{-1}+W_F^T(W_L \Lambda_L W_L^T)^{-1}W_F\}^{-1} W_F^T] \\
&\geq \frac{\lambda_{\min}(W_L \Lambda_L W_L^T)}{\lambda_{\max}^2(W_L \Lambda_L W_L^T)}[q-\tr\{(I + \lambda_{\min}^{-1}(W_L \Lambda_L W_L^T)\Lambda_F^{1/2}W_F^TW_F\Lambda_F^{1/2})^{-1}\}].
\end{align*}
Now
\begin{align*}
\tr\{(I + \lambda_{\min}^{-1}(W_L \Lambda_L W_L^T)\Lambda_F^{1/2}W_F^TW_F\Lambda_F^{1/2})^{-1}\} &\leq q \lambda_{\min}^{-1}(I + \lambda_{\min}^{-1}(W_L \Lambda_L W_L^T)\Lambda_F^{1/2}W_F^TW_F\Lambda_F^{1/2}) \\
&\leq q \{1 + \lambda_{\min}^{-1}(W_L \Lambda_L W_L^T) \lambda_{\min}(\Lambda_F) \lambda_{\min}(W_F^TW_F)\}^{-1}.
\end{align*}
We also have from Lemma~\ref{lem:eval} that $\lambda_{\min}(W_F^TW_F) \geq n - c_2\sqrt{nq}$ with probability at least $1-e^{c_3 n}$.
When \eqref{eq:eig_bd} occurs and also $\lambda_{\min}(W_F^TW_F) \geq n - c_2\sqrt{nq}$ we have
\begin{align*}
\tr(\text{II}) &\geq q \frac{\tr(\Lambda_L)-c\sqrt{n\tr(\Lambda_L)}}{(\tr(\Lambda_L)+c\sqrt{n\tr(\Lambda_L)})^2}\frac{(\gamma_l -\sigma_u)(n - c_2\sqrt{nq})}{(\gamma_l-\sigma_u) (n - c_2\sqrt{nq}) + c_3\{p + c_4\sqrt{np}\}} \\
&\geq \frac{q}{\tr(\Lambda_L)} \left\{1 + O(\sqrt{n/p}) + O(p/(\gamma_l n))\right\}.
\end{align*}
Noting that the event in question has a probability decreasing exponentially in $n$, we arrive at
\[
\E\tr\{(W\Lambda W^T)^{-1}\} = \frac{n-q}{\tr(\Lambda_L)} \left\{ 1 + O\left(\sqrt{\frac{n}{p}}\right) + O\left(\frac{p}{\gamma_l n}\right)\right\}.
\]
Thus substituting in \eqref{eq:C_j_bd} we have that for all $j \geq q+1$,
\[
\frac{C_{jj}^2}{D_{jj}^2} = \frac{n-q}{\tr(\Lambda_L)} \left\{ 1 + O\left(\sqrt{\frac{n}{p}}\right) + O\left(\frac{p}{\gamma_l n}\right)\right\}.
\]

\subsection{Auxilliary lemmas}
\begin{lem} \label{lem:A_j4bd}
Consider the setup of Theorem~\ref{thm:exp_Pi} and define
\[
A_j^2 = Z_j^T(Z_{-j} D_{-j,-j}^2Z_{-j}^T)^{-1} Z_j.
\]
We have
\[
\E A_j^4 \lesssim \frac{n^2}{p^2}.
\]
\end{lem}
\begin{proof}
From Lemma~\ref{lem:mat_ineq1} we have $(Z_{-j} D_{-j,-j}^2Z_{-j}^T)^{-1} \preceq \sigma_l^{-1} (Z_{-j}Z_{-j}^T)^{-1}$. Fix $j$ and let $\Omega = (Z_{-j}Z_{-j}^T)^{-1}$ and $u = Z_j$. We have
\begin{align*}
\E A_j^4 &\lesssim \E (u^T\Omega u)^2 \\
&= \sum_{j,k,l,m} \E (u_j \Omega_{jk}u_ku_l\Omega_{lm}um) \\
&\leq \sum_{j,k} \E u_j \Omega_{jj}u_ju_k \Omega_{kk} u_k + 2 \sum_{j,k} \E(u_j\Omega_{jk}u_k)^2 \\
&\lesssim \E\{\tr(\Omega)\}^2 + \E \tr(\Omega^2) \\
&\lesssim \frac{n^2}{p^2}
\end{align*}
using results on inverse Wishart distributions from \citet{Dietrich_von_Rosen1988-nc}, and specifically Corollary 3.1 (v) in that paper.
\end{proof}

\begin{lem} \label{lem:eval_conc}
Let $W \in \R^{n \times d}$ have independent $\mathcal{N}(0, 1)$ entries. Let $\Lambda \in \R^{d \times d}$ be a symmetric positive definite matrix.
Define $\tr(\Lambda) = N$ and $\tr(\Lambda^2) = M$. Suppose $\lambda_{\max}(\Lambda) \sqrt{8\log(9) n/M} < c$ for constant $c < 1$.
There exist constants $c_1, c_2, c_3>0$ such that with probability at least $1-c_1 e^{-c_2n}$ we have
\[
N - c_3\sqrt{nM} \leq \lambda_{\min}(W\Lambda W^T) \leq \lambda_{\max}(W\Lambda W^T) \leq N + c_3\sqrt{nM}.
\]
\end{lem}
\begin{proof}
Let $\mathcal{N}$ be a $1/4$-net of $S^{n-1}$. From Lemma 5.4 of \citet{Vershynin2010} we have that $|\mathcal{N}| \leq 9^n$ and
\begin{equation} \label{eq:eps_net}
\max_{v: \|v\|_2=1} |v^TW\Lambda W^Tv/N - 1| \leq 2 \max_{v \in \mathcal{N}} |v^TW\Lambda W^Tv/N - 1|.
\end{equation}
By Lemma~\ref{lem:gen_chi_sq}, we have that for each fixed $v$
\[
\pr\left(|v^TW\Lambda W^Tv - N| > tM \right) \leq 2e^{-Mt^2/8}
\]
for $t < \lambda_{\max}^{-1}(\Lambda)$.
Thus
\begin{align*}
\pr(\max_{v \in \mathcal{N}}|v^TW\Lambda W^Tv - N| > tM) \leq 2 \cdot 9^n \cdot e^{-Mt^2/8}.
\end{align*}
Choosing $t = \sqrt{8 \log(9) n/M}/c$ and appealing to \eqref{eq:eps_net} we see that
\begin{align*}
\pr\left( \max_{x:\|x\|_2=1} |x^TW\Lambda W^Tx / N - 1| \leq c \sqrt{nM}/N \right) \leq 2e^{-c_1 n}.
\end{align*}
From this the result follows easily.
\end{proof}


\begin{lem} \label{lem:tr_bd2}
Let $A \in \R^{n\times d}$ and let $B \in \R^{n \times n}$ be symmetric positive semi-definite. Then
\begin{align*}
\frac{1}{\lambda_{\max}(B)}[d-\tr\{(I+\lambda_{\max}(B)A^T A)^{-1}\}] &\leq \tr\{A(I+A^TBA)^{-1}A^T\} \\
&\leq \frac{1}{\lambda_{\min}(B)}[d-\tr\{(I+\lambda_{\min}(B)A^T A)^{-1}\}].
\end{align*}
\end{lem}
\begin{proof}
It is easy to see that $I + \lambda_{\min}(B) A^T A \preceq I+A^TBA \preceq I + \lambda_{\max}(B) A^T A$ whence $(I + \lambda_{\min}(B) A^T A)^{-1} \succeq (I+A^TBA)^{-1} \succeq (I + \lambda_{\max}(B) A^T A)^{-1}$. Therefore we also have
\[
A(I + \lambda_{\min}(B) A^T A)^{-1}A^T \succeq A(I+A^TBA)^{-1}A^T \succeq A(I + \lambda_{\max}(B) A^T A)^{-1}A^T.
\]
It remains only to show that
\[
\tr\{A(I + \lambda_{\min}(B) A^T A)^{-1}A^T\} =  \frac{1}{\lambda_{\min}(B)}[d-\tr\{(I+\lambda_{\min}(B)A^T A)^{-1}\}]
\]
and a similar equality involving $\lambda_{\max}(B)$. We have
\begin{align*}
\lambda_{\min}(B)\tr\{A(I + \lambda_{\min}(B) A^T A)^{-1}A^T\}  &= \tr[A^TA\{\lambda_{\min}^{-1}(B)I +  A^T A\}^{-1}] \\
&= \tr[\{A^TA + \lambda_{\min}^{-1}(B)I - \lambda_{\min}^{-1}(B)I\}\{\lambda_{\min}^{-1}(B)I +  A^T A\}^{-1}\} \\
&= d -\lambda_{\min}^{-1}(B) \tr[\{\lambda_{\min}^{-1}(B)I +A^TA\}^{-1}] \\
&= d - \tr[\{I + \lambda_{\min}(B)A^TA\}^{-1}].
\end{align*}
A similar argument involving $\lambda_{\max}(B)$ completes the proof of the result.
\end{proof}

\section{Proofs of Theorems~\ref{thm:Sigma_rec} and \ref{thm:Sigma_rec_bag}}
\subsection{Proof of Theorem~\ref{thm:Sigma_rec}}

Let $P_F \in \R^{p \times q}$ and $P_L \in \R^{p \times (p-q)}$ be the matrices of first $q$ and last $p-q$ eigenvectors of $\Theta$. Also let $D_F^2 \in \R^{q \times q}$ and $D_L^2 \in \R^{(p-q)\times (p-q)}$ be the top left and bottom right submatrices of $D^2$ respectively.

We have
\[
\Sigma - \kappa\E \hat{\Sigma}_{\rsvp} = \Sigma - P_LP_L^T\Theta P_LP_L^T + P_LP_L^T\Theta P_LP_L^T - \kappa \E \hat{\Sigma}_{\rsvp}
\]
for each $\kappa$.
Proposition~\ref{prop:pop} gives us that
\begin{equation}
\|\Sigma - P_LP_L^T\Theta P_LP_L^T\|_\infty \lesssim \rho_1 \rho_2 + \rho_1^2/\gamma_l + \gamma_u \rho_1^2 / \gamma_l^2.
\end{equation}
It remains to bound
\begin{align*}
\inf_\kappa \|P_LP_L^T\Theta P_LP_L^T - \kappa \E \hat{\Sigma}_{\rsvp}\|_\infty &= \inf_\kappa \|P_L D_L^2 P_L^T - \kappa \E \hat{\Sigma}_{\rsvp}\|_\infty.
\end{align*}
Now from Theorem~\ref{thm:exp_Pi} we have that $\E \hat{\Sigma}_{\rsvp} = P C^2 P^T$ where $C$ is diagonal, so the RHS of the display above is bounded above by
\begin{equation*}
\inf_\kappa \{\|P_L (D_L^2-\kappa C_L^2)P_L^T\| + \kappa\|P_FC_F^2P_F^T\|_\infty\} = \inf_\kappa \{\max_{j=q+1,\ldots,p} |D_{jj}^2 - \kappa C_{jj}^2| + \kappa \|P_FC_F^2P_F^T\|_\infty\}.
\end{equation*}
Here $C_L$ and $C_F$ are defined analogously to $D_L$ and $D_F$.
Let us set $k=q+1$ and take
\begin{equation} \label{eq:kappa_def}
\kappa = D_{kk}^2 / C_{kk}^2.
\end{equation}
Combining \eqref{eq:C_j_bd} and Lemma~\ref{lem:A_j_bd}, we see that
\[
\kappa^{-1} \gtrsim \frac{n}{p} - \sigma_u\frac{n^2}{p^2}.
\]
Thus using the the fact that $\sigma_u =o(p/n)$, we have $\kappa \lesssim p/n$.
Then
\begin{align*}
|D_{jj}^2 - \kappa C_{jj}^2| &= \kappa D_{jj}^2|C_{kk}^2/D_{kk}^2 - C_{jj}^2 / D_{jj}^2| \\
&\lesssim \sigma_u n/p,
\end{align*}
from Theorem~\ref{thm:exp_Pi}.

We now obtain an upper bound on $\lambda_{\max}(C_F)$. From \eqref{eq:C_j} we clearly have $C_{jj}^2 \leq 1$. Also given any random variable $V \in \R$, we have by Jensen's inequality that $1/(1+\E V) \leq \E\{1/(1+V)\}$. Thus Lemma~\ref{lem:A_j_conc} gives us that
\[
C_{jj}^2 \leq \frac{D_{jj}^2}{\tr(D^2)/n + D_{jj}^2}.
\]
Now let $Q \in \R^{p \times q}$ be the matrix of left singular vectors of $\Gamma$. We have
\begin{align*}
\|P_F C_F^2 P_F^T \|_\infty \leq \max_j \|C_F P_F^T e_j\|_2^2 \leq \lambda_{\max}(C_F)^2 \max_j \|P_F^T e_j\|_2^2 \lesssim \min(1, n\gamma_u / p)  \max_j \|P_F^T e_j\|_2^2,
\end{align*}
using \eqref{eq:DK} and \eqref{eq:l_max_2}.

Now
\begin{align*}
\|P_F^T e_j\|_2 &\leq \|P_F^T QQ^T e_j\|_2 + \|P_F^T (I-QQ^T)e_j\|_2 \\
&\leq \|QQ^T e_j\|_2 + \|P_F^T (I-QQ^T)\| \\
&\leq \rho_2 + \rho_1 / \gamma_l
\end{align*}
using \eqref{eq:l_max_3} in the final line.
Thus we finally arrive at
\begin{align} \label{eq:bias_fin}
\|\Sigma - \kappa \E \hat{\Sigma}_{\rsvp}\|_\infty \lesssim \frac{\gamma_u \rho_1^2}{\gamma_l^2} + \rho_1 \rho_2  + \min\left(\frac{p}{n}, \gamma_u\right)\rho_2^2 + \sigma_u \frac{n}{p}.
\end{align}
Then from Theorem~\ref{thm:conc_Pi}, we may conclude that with probability at least $1-c/p$, we have
\[
\|\Sigma - \kappa \hat{\Sigma}_{\rsvp}\|_\infty \lesssim \frac{\gamma_u \rho_1^2}{\gamma_l^2} + \rho_1 \rho_2 + \min\left(\frac{p}{n}, \gamma_u\right)\rho_2^2 + \sigma_u \frac{n}{p} + \sqrt{\frac{\log(p)}{n}}.
\]
\subsection{Proof of Theorem~\ref{thm:Sigma_rec_bag}}
The proof of this result makes use of the proof of Theorem~\ref{thm:Sigma_rec}, and we will refer to equations presented in the previous subsection.
Let $\hat{\Sigma}^{(b)}$ be the RSVP estimate constructed from the $b$th subsample. Note that $\hat{\Sigma}^{(b)}$ are i.i.d.\ with mean $\check{\Sigma} := \E \Sigma^{(1)}$. Let $\kappa$ be defined in relation to $\check{\Sigma}$ as in \eqref{eq:kappa_def}. Then we know that $\kappa \lesssim p/m$ as $\sigma_u = o(p/m)$ by assumption. We have
\begin{align} \label{eq:subsamp_decomp}
\norm{\Sigma - \frac{\kappa}{B} \sum_{b=1}^B \hat{\Sigma}^{(b)}}_\infty \leq \|\Sigma - \kappa \check{\Sigma}\|_\infty + \frac{\kappa}{B}\norm{ \sum_{b=1}^B(\hat{\Sigma}^{(b)} - \check{\Sigma})}_\infty.
\end{align}
From \eqref{eq:bias_fin} we have that
\begin{align} \label{eq:subsamp_bias}
\|\Sigma - \kappa \E \hat{\Sigma}_{\rsvp}\|_\infty \lesssim \frac{\gamma_u \rho_1^2}{\gamma_l^2} + \rho_1 \rho_2  + \min\left(\frac{p}{m}, \gamma_u\right)\rho_2^2 + \sigma_u^2 \frac{m}{p}.
\end{align}
The second term in \eqref{eq:subsamp_decomp} involves an average of i.i.d.\ mean-zero random matrices $\hat{\Sigma}^{(b)} - \check{\Sigma}$; let us fix $j,k \in \{1,\ldots,p\}$ and let $W_b$ be the $jk$th entry. Now as $\hat{\Sigma}^{(b)}$ is a projection matrix, we have $|W_b| \leq 1$. Also from Lemma~\ref{lem:conc_Pi} we have that for any fixed $r>0$, there exists $c_1>0$ sufficiently large and $c_2 > 0$ (both independent of $m$) such that
\[
\pr(|W_b| \geq c_1 \sqrt{m \log p}/p) \lesssim p^{-r} + e^{-c_2 m}.
\]
Here and below, $\lesssim$ signs contain hidden constants that do not depend on $m$. Applying Lemma~\ref{lem:Hoeffding_var} to $W_1,\ldots,W_B$ we have
\[
\pr\left(\abs{\frac{1}{B}\sum_{b=1}^B W_b} \geq t \right) \lesssim \exp\left(-\frac{Bp^2 t^2}{c_3 m \log(p)}\right)+ B p^{-r}
\]
where $c_3 > 0$ is a constant depending on $r$. Taking $r$ sufficiently large and applying a union bound, we have that with probability at least $1-c p^{-1}$, for some constant $c>0$,
\[
\frac{1}{B}\norm{ \sum_{b=1}^B(\hat{\Sigma}^{(b)} - \check{\Sigma})}_\infty \lesssim \frac{m}{p}\frac{\log(p)}{\sqrt{mB}}.
\]
Substituting this and \eqref{eq:subsamp_bias} into \eqref{eq:subsamp_decomp} gives the result.

\subsection{Auxiliary lemmas}
\begin{lem} \label{lem:Hoeffding_var}
Let $W_1,\ldots,W_n$ be i.i.d.\ mean-zero random variables. Suppose $\pr(|W_i| \geq \tau) \leq \alpha$ and $|W_i| \leq M$ almost surely. Then provided $t > 2M\alpha$, we have
\[
\pr\left(\abs{\frac{1}{n}\sum_{i=1}^n W_i} \geq t \right) \leq 2e^{-nt^2/(8\tau^2)} + n\alpha.
\]
\end{lem}
\begin{proof}
Let $V_i = W_i \ind_{\{|W_i| \leq \tau\}}$, and let $\Omega = \{\|W\|_\infty \leq \tau\}$. Also define $\bar{W} = n^{-1}\sum_{i=1}^n W_i$ and $\bar{V} = n^{-1} \sum_{i=1}^n V_i$. We have
\begin{align*}
\pr(|\bar{W}| \geq t ) &\leq \pr(|\bar{W}| \geq t \cap \Omega) + \pr(\Omega^c) \\
&\leq \pr(|\bar{V}| \geq t) + n\alpha.
\end{align*}
From Hoeffding's inequality we have
\[
\pr(|\bar{V} - \E V_1| \geq r ) \leq 2e^{-nr^2/(2\tau^2)}.
\]
Also
\begin{align*}
\E V_1 = \E W_1 \ind_{\{|W_1|\leq \tau\}} = \E W_1 \ind_{\{|W_1|> \tau\}}
\end{align*}
as $\E W_1 = 0$. By H\"older's inequality we have $ \E W_1 \ind_{\{|W_1|> \tau\}} \leq M \alpha$. Putting things together, we arrive at
\[
\pr(|\bar{W}| \geq t ) \leq 2e^{-n(t-M\alpha)^2/(2\tau^2)} + n\alpha.
\]
\end{proof}
\section{Derivation of \eqref{eq:rho_bd}} \label{sec:rho_bd}
We will make use of the following result which appears as Lemma 2.2 in \citet{dasgupta2003elementary}.
\begin{lem} \label{lem:conc_unif}
Let $W \in \R^{p \times q}$ be uniformly distributed on the Stiefel manifold $V_{q}(\R^p)$ and let $v \in \R^p$ be a unit vector. Then for $t>0$ we have
\begin{align*}
\pr\left( \|W^T v\|_2^2 > (1+t)\frac{q}{p}\right) \leq \exp\left(\frac{q}{2} \{-t + \log(1+t)\} \right)
\end{align*}
\end{lem}
Note that $\Pi_{\Gamma} \eqdist WW^T$ with $W$ defined as above. Now $t - \log(1+t) \geq t\min(1, t)/4$ and if $t\min(1, t)/4 =a$ then $t=\max(2\sqrt{a}, 4a)$. Thus setting $t=\max\{4\sqrt{\log(p)/q}, 16\log(p)/q\}$, we have $q\{-t + \log(1+t)\}/2 < -2\log(p)$.

Applying a union bound we therefore obtain
\begin{align*}
\pr\left\{ \max_j \|\Pi_{\Gamma} e_j\|_2^2 > (1+t)\frac{q}{p} \right\} \leq p \, \pr\left\{ \|\Pi_{\Gamma} e_j\|_2^2 > (1+t)\frac{q}{p} \right\} \leq \frac{1}{p}.
\end{align*}

We now turn to the final part of the bound. Let the eigendecomposition of $\Sigma$ be given by $\Sigma = R B R^T$. Note that $\Pi_{\Gamma} R \eqdist \Pi_{\Gamma}$ so $\|\Pi_{\Gamma} \Sigma\| \eqdist \|\Pi_{\Gamma}B R^T \| = \|\Pi_{\Gamma} B\|$. By an argument analogous to that of Corollary~\ref{cor:Y_dist}, we know that $\Pi_{\Gamma} \eqdist W^T (WW^T)^{-1} W$ where $W \in \R^{q \times p}$ has i.i.d.\ $\mathcal{N}(0, 1)$ entries. Now from Lemma~\ref{lem:mat_ineq1} we have
\[
B \Pi_{\Gamma} B  = B W^T (WW^T)^{-1} W B \preceq \lambda_{\min}^{-1}(WW^T) B W^T W B.
\]
Thus
\[
\|\Pi_{\Gamma} B\|^2 \leq \lambda_{\min}^{-1}(WW^T) \lambda_{\max}(B W^T W B) = \lambda_{\min}^{-1}(WW^T) \lambda_{\max}(W B^2 W^T).
\]
From Lemma~\ref{lem:eval_conc} we have that there exists constants $c_1, c_2, c_3 >0$ such that with probability at least $1-c_1 e^{-c_2 q}$ we have
\[
\lambda_{\min}^{-1}(WW^T) \lambda_{\max}(W B^2 W^T) \leq \frac{\tr(\Sigma^2) + c_3\sqrt{q \tr(\Sigma^4)} }{p-c_3 \sqrt{qp}}.
\]

\section{Derivation of \eqref{eq:eta_bd}} \label{sec:eta_bd}
Let the eigendecomposition of $\Gamma \Gamma^T$ be given by $Q^T A^2 Q$ where $Q \in \R^{p \times q}$ has orthonormal columns and $A \in \R^{q \times q}$ is diagonal. Now if $Q$ is uniformly distributed on the Stiefel manifold $V_q(\R^p)$, then
\[
(\beta^{(j)})^T QA^2 Q \beta^{(j)} \eqdist (\beta^{(j)})^T R I^{p,q} A^2 I^{q,p} R^T \beta^{(j)} \eqdist \|\beta^{(j)}\|_2^2 v^T A^2 v / \|v\|_2^2
\]
where $R$ is uniformly distributed on the set of orthogonal matrices, $I^{p,q}$ denotes the $p \times q$ matrix with 1's along its diagonal and 0's elsewhere, and $v$ has i.i.d.\ $\mathcal{N}(0,1)$ entries. Note that $\tr(A^2) = \tr(\Gamma \Gamma^T) \leq \tr(\Theta) \lesssim p$. Now $\tr(A^4)$ is maximised for a given value of $\tr(A^2)$ if a single entry of $A^2$ equals $\tr(A^2)$ with all others equal to zero. Thus $\tr(A^4) \lesssim p^2$.
From Lemma~\ref{lem:gen_chi_sq} we have that with probability at least $1-2e^{-c_1 p}$,
\[
v^T A^2 v \lesssim p \;\; \text{ and } \;\; \|v \|_2^2 \gtrsim p.
\]
Thus, for $p$ sufficiently large, with probability at least $1-e^{-c_2 p}$ we have $\eta_j \lesssim 1$ for all $j$.

\section{Proof of Theorem~\ref{thm:NS}} \label{sec:NS_proof}

We adopt the notation of the proof of Theorem~\ref{thm:Sigma_rec}; in particular, we take $\kappa = D_{q+1,q+1}^2 / C_{q+1,q+1}^2$. Let us write $\hat{\Sigma}$ for $\hat{\Sigma}_{\rsvp}$ for notational simplicity. We fix $j$ and define
\[
\hat{\beta} = \argmin{b \in \R^{p-1}} \frac{1}{2} b^T \hat{\Sigma}_{-j,-j} b - b^T \hat{\Sigma}_{-j,j} + \lambda_j \|b\|_1.
\]
Note $\hat{\beta}_l = \hat{\beta}^{(j)}_l$ for $l < j$ and $\hat{\beta}_l = \hat{\beta}^{(j)}_{l+1}$ for $l \geq j$.
The KKT conditions for $\hat{\beta}$ take the following form:
\begin{equation} \label{eq:KKT_NS}
 \hat{\Sigma}_{-j,j} - \hat{\Sigma}_{-j,-j} \hat{\beta} = \lambda_j \hat{\nu}
\end{equation}
where $\|\hat{\nu}\|_\infty \leq 1$ and writing $\hat{S} = \{k : \hat{\beta}_k \neq 0\}$ we have $\hat{\nu}_{\hat{S}} = \sgn(\hat{\beta}_{\hat{S}})$. Note that $\hat{\beta}^T\hat{\nu} = \|\hat{\beta}\|_1$.

Let $\beta = \Sigma_{-j,-j}^{-1}\Sigma_{-j,j}$,  and set $S = \{j: \beta_j \neq 0\}$ and $N = S^c$. Also define $\delta = \hat{\Sigma}_{-j,j} - \hat{\Sigma}_{-j,-j}\beta$. 
Dotting both sides of \eqref{eq:KKT_NS} with $\beta - \hat{\beta}$ and using H\"older's inequality, we obtain
\[
(\beta - \hat{\beta})^T \hat{\Sigma}_{-j,-j} (\beta - \hat{\beta}) \leq \lambda_j(\|\beta\|_1 - \|\hat{\beta}\|_1) + \|\delta\|_\infty \|\beta -  \hat{\beta}\|_1.
\]
We will show below that with high probability $\|\delta\|_\infty \leq \lambda_j/2$.
Working for now on the event $\Lambda^{(1)}_j$ where this holds, we have
\begin{equation} \label{eq:lasso_pf1}
(\beta - \hat{\beta})^T \hat{\Sigma}_{-j,-j} (\beta - \hat{\beta}) \leq \lambda_j(\|\beta\|_1 - \|\hat{\beta}\|_1) + \lambda_j\|\beta -  \hat{\beta}\|_1/2.
\end{equation}
We may now follow the standard proof for bounding the estimation error of the Lasso (see Chapter 6 of \citet{}, for example).
Note that as the LHS of \eqref{eq:lasso_pf1} is non-negative,
\[
2\|\beta_S\|_1 - 2\|\hat{\beta}_S \|_1 - \|\hat{\beta}_N\|_1 + \|\beta_S - \hat{\beta}_S\|_1 \geq 0.
\]
Therefore, using the triangle inequality we have
\begin{equation} \label{eq:cone}
3\|\hat{\beta}_S - \beta_S \|_1 \geq \|\hat{\beta}_N - \beta_N\|_1.
\end{equation}
For a symmetric positive-definite matrix $\Omega \in \R^{p \times p}$, let us define the $j$th restricted eigenvalue REF as
\[
\phi^2_j(\Omega) = \min_{b : \|b_{S_j^c}\|_1 \leq 3\|b_{S_j}\|_1 \neq 0} \frac{b^T\Omega b}{\|b\|_2^2}. 
\]
Also let $\phi^2_j := \phi^2_j(\hat{\Sigma})$.
We seek to bound $\phi^2_j$ from below, for which we use Corollary 10.1 of \citet{BuhlmannGeer2009}. This states in particular that if $\|\Omega - \hat{\Sigma}\|_\infty \leq \phi^2_j(\Omega) / (32 s_j)$, then $\phi^2_j \geq \phi^2_j(\Omega)/2$. We apply this with $\Omega = \kappa^{-1}\Sigma$, noting that $\phi^2_j(\kappa^{-1}\Sigma) \geq \kappa^{-1}\sigma_l > 0$. We then have that on the event $\Lambda^{(2)}$ that $\|\hat{\Sigma} - \kappa^{-1}\Sigma\|_\infty \leq \kappa^{-1}\sigma_l / (32s)$, $\phi^2_j \geq \kappa^{-1}\sigma_l/2$ for all $j$. From Theorem~\ref{thm:conc_Pi}, and using $s=o(\sqrt{\log(p)/n})$, we have that the probability of the event $\Lambda^{(2)}$ is at least $1-c/p$ for $n$ sufficiently large. In what follows we work on this event.

From the above, we have
\begin{align*}
\frac{2\kappa}{\sigma_l}(\beta - \hat{\beta})^T\hat{\Sigma}_{-j,-j}(\beta - \hat{\beta}) &\geq \|\beta - \hat{\beta}\|_2^2 \\
&\geq \|\beta_S - \hat{\beta}_S\|_2^2 \\
&\geq \frac{\{\sgn(\beta_S - \hat{\beta}_S)^T(\beta_S - \hat{\beta}_S)\}^2}{\|\sgn(\beta_S - \hat{\beta}_S)\|_2^2} \\
&\geq \|\beta_S - \hat{\beta}_S\|_1^2 / s,
\end{align*}
using the Cauchy--Schwarz inequality in the penultimate line.
Returning to \eqref{eq:lasso_pf1}, we have that
\[
\|\beta_S - \hat{\beta}_S\|_1^2 / s \leq \|\beta - \hat{\beta}\|_2^2 \leq \frac{3\kappa \lambda_j}{\sigma_l}\|\hat{\beta}_S - \beta_S\|_1,
\]
whence
\[
\|\beta_S - \hat{\beta}_S\|_1 \leq \frac{3s\kappa\lambda_j}{\sigma_l} \;\;\; \text{ and } \;\;\; \|\beta - \hat{\beta}\|_2 \leq \frac{3\sqrt{s}\kappa\lambda_j}{\sigma_l}.
\]
Also from \eqref{eq:cone} we have $\|\hat{\beta} - \beta\|_1 \leq 4 \|\hat{\beta}_S - \beta_S\|_1$, giving $\|\beta - \hat{\beta}\|_1 \leq 12s \kappa\lambda_j / \sigma_l$ as required.

It remains to show $\cap_j \Lambda^{(1)}_j$ occurs with high probability. Now using H\"older's inequality, we have
\begin{align*}
\kappa \|\delta\|_\infty &=\|\hat{\Sigma}_{-j,j} -  \hat{\Sigma}_{-j,-j}\beta\|_\infty \\
&\leq  \|\kappa \hat{\Sigma}_{-j,j} -  \Sigma_{-j,j}\|_\infty + \| \Sigma_{-j,-j} - \kappa \E \hat{\Sigma}_{-j,-j}\|_\infty \|\beta\|_1 + \kappa\|(\E \hat{\Sigma}_{-j,-j} - \hat{\Sigma}_{-j,-j})\beta\|_\infty \\
&=: \text{I}_j + \text{II}_j + \text{III}_j.
\end{align*}
We seek to control each of the terms $\text{I}_j, \text{II}_j, \text{III}_j$ above.
By Theorem~\ref{thm:conc_Pi}, with probability at least $1-c/p$, we have $\max_j \text{I}_j \lesssim \sqrt{\log(p)/n}$. Next, by Lemma~\ref{lem:conc_Pi}, we have that with probability at least $1-c/p^3$,
\begin{align*}
|e_l^T (\E \hat{\Sigma} - \hat{\Sigma})\beta^{(j)}| \lesssim \sqrt{\eta_j \log(p)/n}.
\end{align*}
Thus by a union bound we have with probability at least $1-c/p$,
\[
\text{III}_j = \max_l |e_l^T (\E \hat{\Sigma} - \hat{\Sigma})\beta^{(j)}| \lesssim  \sqrt{\eta_j \log(p)/n}.
\]
Turning to $\text{II}_j$, note by the Cauchy--Schwarz inequality
\[
\|\beta\|_1 = \beta_S^T \sgn(\beta_S) \leq \|\beta\|_2 \sqrt{s}.
\]
Also, \eqref{eq:bias_fin} gives us
\begin{align*}
\|\Sigma - \kappa \E \hat{\Sigma}_{\rsvp}\|_\infty \lesssim \Delta \sqrt{\log(p)/(sn)}.
\end{align*}
Now observe that $\|\beta\|_2 \lesssim 1$. Indeed, we have $1 \gtrsim \Sigma_{jj} \gtrsim \Sigma_{j,-j}\Sigma_{-j,-j}^{-1}\Sigma_{-j,j}$. Then
\[
\|\beta\|_2^2 =\Sigma_{j,-j}\Sigma_{-j,-j}^{-1/2} \Sigma_{-j,-j}^{-1}\Sigma_{-j,-j}^{-1/2} \Sigma_{-j,j} \leq \sigma_l^{-1} \|\Sigma_{j,-j}\Sigma_{-j,-j}^{-1/2}\|_2^2 \lesssim 1.
\]
Thus $\text{II}_j \lesssim \Delta \sqrt{\log(p)/n}$. Putting things together, we have 
\[
\kappa\|\delta\|_\infty \lesssim \sqrt{\max(\Delta, \eta_j, 1)\log(p)/n}.
\]
Now from \eqref{eq:C_j_bd} and Lemma~\ref{lem:A_j_bd} we have $\kappa \gtrsim p/n$. Hence $\|\delta\|_\infty \leq \lambda_j/2$ as required.

\section{Additional Plots for the GTEX data analysis} \label{sec:GTEX_indiv}
Below we provide additional plots concerning the  data analyses presented in Section~\ref{sec:experiments} of the main text. Figure~\ref{fig:gtex_enrich_scores} displayed average performance across all of the 44 tissues studied. Here, we present the area under the curve of the graph of enrichment scores as a function of the number of edges, based on the raw data for each
of these tissues.

\begin{figure}[h!]
\begin{center}
\includegraphics[width=0.95\textwidth]{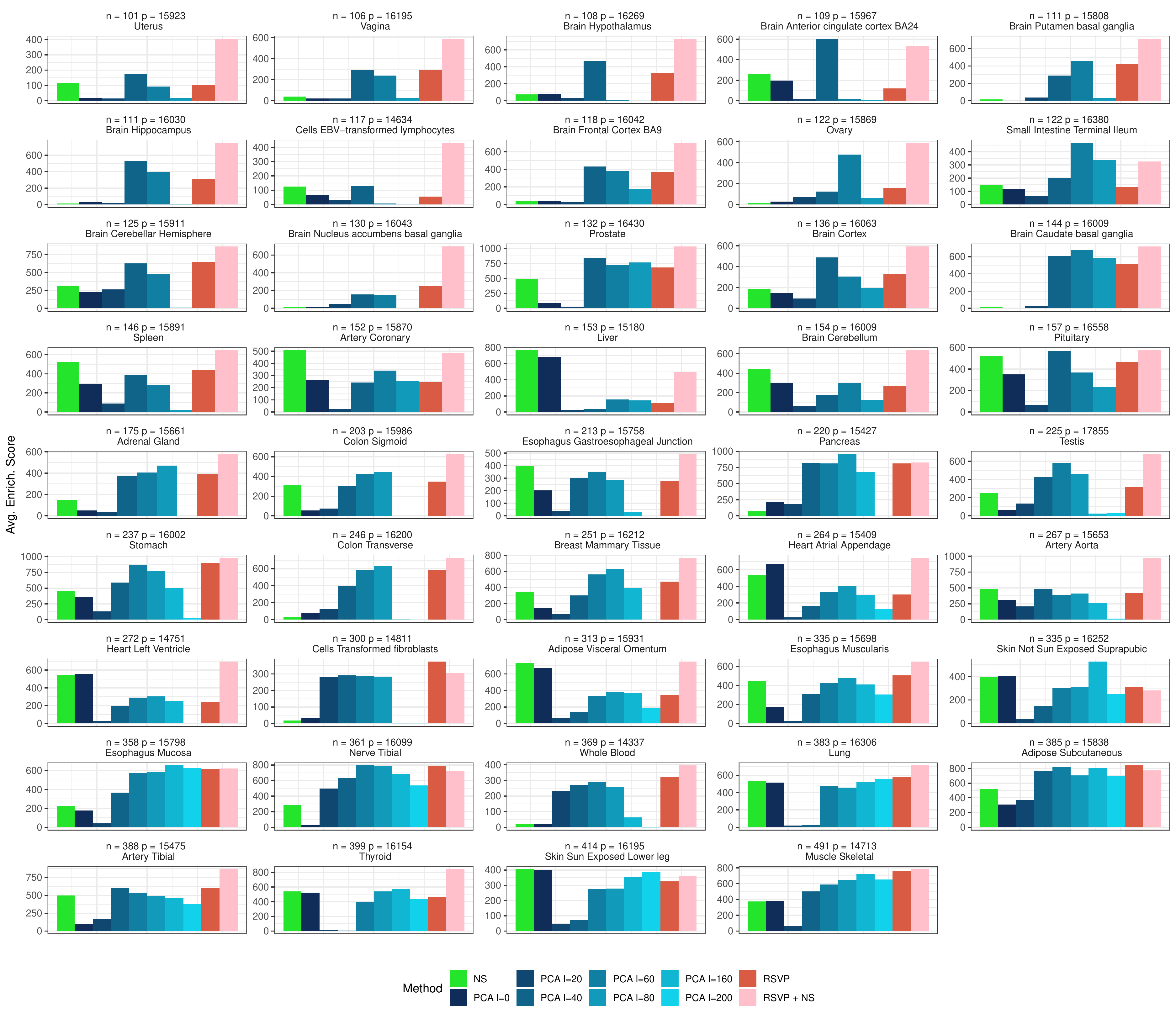}
\caption{The area under the curve (AUC) of the graph of enrichment score as a function of the number of edges, based on the raw data for each of the 44 tissues. \label{fig:suppl_figure_gtex}
  }
\end{center}
\end{figure}

\end{cbunit}


\begin{thebibliography}{43}
\providecommand{\natexlab}[1]{#1}
\providecommand{\url}[1]{\texttt{#1}}
\expandafter\ifx\csname urlstyle\endcsname\relax
  \providecommand{\doi}[1]{doi: #1}\else
  \providecommand{\doi}{doi: \begingroup \urlstyle{rm}\Url}\fi

\bibitem[Aguet et~al.(2017)]{Aguet2017}
F.~Aguet et~al.
\newblock Genetic effects on gene expression across human tissues.
\newblock \emph{Nature}, 550\penalty0 (7675):\penalty0 204--213, 2017.

\bibitem[Ashburner et~al.(2000)]{Ashburner2000}
M.~Ashburner et~al.
\newblock Gene ontology: Tool for the unification of biology.
\newblock \emph{Nature Genetics}, 25, 2000.

\bibitem[Bai and Ng(2002)]{bai2002determining}
J.~Bai and S.~Ng.
\newblock Determining the number of factors in approximate factor models.
\newblock \emph{Econometrica}, 70\penalty0 (1):\penalty0 191--221, 2002.

\bibitem[Barigozzi and Cho(2018)]{barigozzi2018consistent}
M.~Barigozzi and H.~Cho.
\newblock Consistent estimation of high-dimensional factor models when the
  factor number is over-estimated.
\newblock \emph{arXiv preprint arXiv:1811.00306}, 2018.

\bibitem[Belloni et~al.(2011)]{belloni2011}
A.~Belloni and V.~Chernozhukov~L.~Wang.
\newblock {Square-root lasso: pivotal recovery of sparse signals via conic programming}.
\newblock \emph{Biometrika}, 98\penalty0(4): 791--806, 2011.

\bibitem[Bickel and Levina(2008)]{bickel2008covariance}
P.~Bickel and E.~Levina.
\newblock {Covariance regularization by thresholding}.
\newblock \emph{The Annals of Statistics}, 36:\penalty0 (6):\penalty0  2577--2604, 2008.

\bibitem[Breiman(1996)]{breiman1996bagging}
L.~Breiman.
\newblock {Bagging predictors}.
\newblock \emph{Machine Learning}, 24:\penalty0 (2):\penalty0 123--140, 1996.

\bibitem[Cai et~al.(2011)Cai, Liu, and Luo]{Cai2011}
T.~Cai, W.~Liu, and X.~Luo.
\newblock A constrained $\ell_1$ minimization approach to sparse precision
  matrix estimation.
\newblock \emph{Journal of the American Statistical Association}, 106\penalty0
  (494):\penalty0 594--607, 2011.


\bibitem[Cai et~al.(2016)Cai, Ren, and Zhou]{cai2016}
T.~T. Cai, Z.~Ren, and H.~H. Zhou.
\newblock Estimating structured high-dimensional covariance and precision
  matrices: Optimal rates and adaptive estimation.
\newblock \emph{Electron. J. Statist.}, 10\penalty0 (1):\penalty0 1--59,
  2016.

\bibitem[Cand\`es et~al.(2011)Cand\`es, Li, Ma, and Wright]{CandesEtAl11}
E.~Cand\`es, X.~Li, Y.~Ma, and J.~Wright.
\newblock Robust principal component analysis?
\newblock \emph{J. ACM}, 58:\penalty0 (3):\penalty0 Art. 11, 2011.

\bibitem[Cevid et~al.(2018)Cevid, B\"uhlmann, and Meinshausen]{Cevid18spectral}
D.~{\'C}evid, P.~B\"uhlmann, and N.~Meinshausen.
\newblock Spectral Deconfounding and Perturbed Sparse Linear Models
\newblock \emph{arXiv preprint arXiv:1811.05352}, 2018.

\bibitem[Chandrasekaran et~al.(2011)Chandrasekaran, Sanghavi, Parrilo, and
  Willsky]{ChandrasekaranEtAl11}
V.~Chandrasekaran, S.~Sanghavi, P.~A. Parrilo, and A.~S. Willsky.
\newblock Rank-sparsity incoherence for matrix decomposition.
\newblock \emph{SIAM Journal on Optimization}, 21:\penalty0 (2):\penalty0 572--596, 2011.

\bibitem[Chandrasekaran et~al.(2012)Chandrasekaran, Parrilo, Willsky,
  et~al.]{chandrasekaran2012latent}
V.~Chandrasekaran, P.~A. Parrilo, A.~S. Willsky, et~al.
\newblock Latent variable graphical model selection via convex optimization.
\newblock \emph{The Annals of Statistics}, 40\penalty0 (4):\penalty0
  1935--1967, 2012.

\bibitem[Chernozhukov et~al.(2017)Chernozhukov, Hansen, Liao,
  et~al.]{chernozhukov2017lava}
V.~Chernozhukov, C.~Hansen, Y.~Liao, et~al.
\newblock A lava attack on the recovery of sums of dense and sparse signals.
\newblock \emph{The Annals of Statistics}, 45\penalty0 (1):\penalty0 39--76,
  2017.

\bibitem[Davis and Kahan(1970)]{davis1970rotation}
C.~Davis and W.~M. Kahan.
\newblock The rotation of eigenvectors by a perturbation. iii.
\newblock \emph{SIAM Journal on Numerical Analysis}, 7\penalty0 (1):\penalty0
  1--46, 1970.

\bibitem[Donoho et~al.(2018)Donoho, Gavish, and Johnstone]{donoho2018}
D.~Donoho, M.~Gavish, and I.~Johnstone.
\newblock Optimal shrinkage of eigenvalues in the spiked covariance model.
\newblock \emph{The Annals of Statistics}, 46\penalty0 (4):\penalty0 1742--1778, 2018.


\bibitem[Fan et~al.(2013)Fan, Liao, and Mincheva]{fan2013large}
J.~Fan, Y.~Liao, and M.~Mincheva.
\newblock Large covariance estimation by thresholding principal orthogonal
  complements.
\newblock \emph{Journal of the Royal Statistical Society, Series B}, 75\penalty0 (4):\penalty0 603--680, 2013.

\bibitem[Fan et~al.(2018)Fan, Liu, and Wang]{Fan2018-jg}
J.~Fan, H.~Liu, and W.~Wang.
\newblock Large covariance estimation through elliptical factor models.
\newblock \emph{The Annals of Statistics}, 46\penalty0 (4):\penalty0 1383--1414, 2018.

\bibitem[Friedman et~al.(2008)Friedman, Hastie, and
  Tibshirani]{Friedman2008}
J.~Friedman, T.~Hastie, and R.~Tibshirani.
\newblock Sparse inverse covariance estimation with the graphical lasso.
\newblock \emph{Biostatistics}, 9\penalty0 (3):\penalty0 432--441,
  2008.

\bibitem[Friedman et~al.(2018)Friedman, Hastie, and
  Tibshirani]{Friedman2018}
  J.~Friedman, T.~Hastie, and R.~Tibshirani.
\newblock glasso: Graphical Lasso: Estimation of Gaussian Graphical Models.
\newblock \emph{R package version 1.10.} \url{  https://CRAN.R-project.org/package=glasso
}

\bibitem[Frot et~al.(2018)Frot, Jostins, and McVean]{FrotEtAl18}
B.~Frot, L.~Jostins, and G.~McVean.
\newblock Graphical model delection for gaussian conditional random fields in
  the presence of latent variables.
\newblock \emph{Journal of the American Statistical Association},  114\penalty0 (526):\penalty0 723--734, 2019

\bibitem[Frot et~al.(2019)Frot, Nanda, and Maathuis]{FrotEtAl19}
B.~Frot, P.~Nandy, and M.~Maathuis.
\newblock Robust causal structure learning with some hidden variables.
\newblock \emph{Journal of the Royal Statistical Society, Series B}, 81\penalty0 (3):\penalty0 459--487, 2019.

\bibitem[Gagnon-Bartsch et~al.(2013)Gagnon-Bartsch, Jacob, and
  Speed]{Speed2013}
J.~A. Gagnon-Bartsch, L.~Jacob, and T.~P. Speed.
\newblock Removing unwanted variation from high dimensional data with negative
  controls.
\newblock Technical Report 820, Department of Statistics, University of
  California at Berkeley, December 2013.
  
\bibitem[Gissibl and Kl\"uppelberg(2018)]{gissibl2018}
N.~Gissibl and C.~Kl\"uppelberg.
\newblock Max-linear models on directed acyclic graphs.
\newblock \emph{Bernoulli}, 24\penalty0 (4A), 2693--2720, 2018.


\bibitem[Haavelmo(1944)]{haavelmo1944statistical}
T.~Haavelmo
\newblock The probability approach in econometrics.
\newblock \emph{Econometrica}, 12, iii--115, 1944

\bibitem[Hallin and Li{\v{s}}ka(2007)]{hallin2007determining}
M.~Hallin and R.~Li{\v{s}}ka.
\newblock Determining the number of factors in the general dynamic factor
  model.
\newblock \emph{Journal of the American Statistical Association}, 102\penalty0
  (478):\penalty0 603--617, 2007.

\bibitem[Harris and Drton(2013)]{harris2013pc}
N.~Harris and M.~Drton.
\newblock PC algorithm for nonparanormal graphical models.
\newblock \emph{The Journal of Machine Learning Research}, 14\penalty0
  (1):\penalty0 3365--3383, 2013.



\bibitem[Heinze-Deml et~al.(2018)Heinze-Deml, Maathuis, and
  Meinshausen]{heinze2018causal}
C.~Heinze-Deml, M.~Maathuis, and N.~Meinshausen.
\newblock Causal structure learning.
\newblock \emph{Annual Review of Statistics and Its Application}, 5:\penalty0 371--391, 2018.


\bibitem[Jia et~al.(2015)Jia, Rohe, et~al.]{jia2015preconditioning}
J.~Jia, K.~Rohe, et~al.
\newblock Preconditioning the lasso for sign consistency.
\newblock \emph{Electronic Journal of Statistics}, 9\penalty0 (1):\penalty0
  1150--1172, 2015.

\bibitem[Kalisch and B{\"u}hlmann(2007)]{kalisch2007estimating}
M.~Kalisch and P.~B{\"u}hlmann.
\newblock Estimating high-dimensional directed acyclic graphs with the
  pc-algorithm.
\newblock \emph{Journal of Machine Learning Research}, 8:\penalty0 613--636, 2007.

\bibitem[{Klochkov} and {Zhivotovskiy}(2018)]{2018arXiv181203548K}
Y.~{Klochkov} and N.~{Zhivotovskiy}.
\newblock {Uniform Hanson-Wright type concentration inequalities for unbounded
  entries via the entropy method}.
\newblock \emph{arXiv e-prints}, arXiv:1812.03548, Dec 2018.

\bibitem[Lauritzen(1996)]{lauritzen96graphical}
S.~Lauritzen.
\newblock \emph{{Graphical Models}}.
\newblock Oxford University Press, 1996.

\bibitem[Ledoit and Wolf(2004)]{ledoit2004well}
O.~Ledoit and M.~Wolf.
\newblock A well-conditioned estimator for large-dimensional covariance
  matrices.
\newblock \emph{Journal of multivariate analysis}, 88\penalty0 (2):\penalty0
  365--411, 2004.

\bibitem[Leek and Storey(2007)]{Leek2007}
J.~T. Leek and J.~D. Storey.
\newblock Capturing heterogeneity in gene expression studies by surrogate
  variable analysis.
\newblock \emph{{PLoS} Genetics}, 3:\penalty0 e161, 2007.

\bibitem[Meek(1995)]{Meek1995}
C.~Meek.
\newblock Strong completeness and faithfulness in Bayesian networks.
\newblock In \emph{Uncertainty in Artificial Intelligence}, pages 411--418, 1995.

\bibitem[Meinshausen and B\"uhlmann(2006)]{meinshausen04consistent}
N.~Meinshausen and P.~B\"uhlmann.
\newblock {High dimensional graphs and variable selection with the Lasso}.
\newblock \emph{The Annals of Statistics}, 34:\penalty0 (3):\penalty0 1436--1462, 2006.

\bibitem[Menchero et~al.(2010)Menchero, Morozov, and
  Shepard]{menchero2010global}
J.~Menchero, A.~Morozov, and P.~Shepard.
\newblock Global equity risk modeling.
\newblock In \emph{Handbook of Portfolio Construction}, pages 439--480.
  Springer, 2010.

\bibitem[Pearl(2009)]{pearl2009causality}
J.~Pearl.
\newblock \emph{Causality}.
\newblock Cambridge University Press, 2009.


\bibitem[Ren et~al.(2015)Ren, Sun, Zhang, Zhou, et~al.]{ren2015asymptotic}
Z.~Ren, T.~Sun, C.-H. Zhang, H.~H. Zhou, et~al.
\newblock Asymptotic normality and optimalities in estimation of large gaussian
  graphical models.
\newblock \emph{The Annals of Statistics}, 43\penalty0 (3):\penalty0 991--1026,
  2015.



\bibitem[Robins(1986)]{robins1986new}
J.~Robins
\newblock A new approach to causal inference in mortality studies with a sustained exposure period—application to control of the healthy worker survivor effect.
\newblock \emph{Mathematical modelling}, 7\penalty0 (9):\penalty0 1393--1512,
  2015.



\bibitem[Rohe(2014)]{rohe2014note}
K.~Rohe.
\newblock A note relating ridge regression and OLS p-values to preconditioned sparse penalized regression.
\newblock \emph{arXiv preprint arXiv:1411.7405}, 2014.

\bibitem[Spirtes et~al.(2000)Spirtes, Glymour, and
  Scheines]{spirtes2000causation}
P.~Spirtes, C.~Glymour, and R.~Scheines.
\newblock Causation, prediction, and search. adaptive computation and machine
  learning, 2000.

\bibitem[Stegle et~al.(2012)Stegle, Parts, Piipari, Winn, and
  Durbin]{Stegle2012}
O.~Stegle, L.~Parts, M.~Piipari, J.~Winn, and R.~Durbin.
\newblock Using probabilistic estimation of expression residuals ({PEER}) to
  obtain increased power and interpretability of gene expression analyses.
\newblock \emph{Nature Protocols}, 7:\penalty0 500--507, 2012.

\bibitem[Thanei et~al.(2018)Thanei, Meinshausen, and Shah]{Thanei2018}
G.-A.~Thanei, N.~Meinshausen, and R.~D.~Shah.
\newblock The xyz algorithm for fast interaction search in high-dimensional data.
\newblock \emph{The Journal of Machine Learning Research}, 19\penalty0 (1):\penalty0 1343--1384, 2018.

\bibitem[Wang and Leng(2015)]{Wang15}
X.~Wang and C.~Leng.
\newblock High dimensional ordinary least squares projection for screening variables.
\newblock \emph{Journal of the Royal Statistical Society, Series B}, 78:\penalty0 589--611, 2015.

\bibitem[Yuan(2010)]{Yuan2010}
M.~Yuan.
\newblock High dimensional inverse covariance matrix estimation via linear
  programming.
\newblock \emph{The Journal of Machine Learning Research}, 11:\penalty0 2261--2286, 2010.

\bibitem[Yuan and Lin(2007)]{yuan05model}
M.~Yuan and Y.~Lin.
\newblock {Model selection and estimation in the Gaussian graphical model}.
\newblock \emph{Biometrika}, 94:\penalty0 19--35, 2007.

\end{thebibliography}

\begin{thebibliography}{9}
\providecommand{\natexlab}[1]{#1}
\providecommand{\url}[1]{\texttt{#1}}
\expandafter\ifx\csname urlstyle\endcsname\relax
  \providecommand{\doi}[1]{doi: #1}\else
  \providecommand{\doi}{doi: \begingroup \urlstyle{rm}\Url}\fi

\bibitem[{B\"uhlmann} and {van de Geer}, 2011]{hds}
P.~{B\"uhlmann} and S.~{van de Geer}.
\newblock {\em {Statistics for High-Dimensional Data}}.
\newblock Springer, 2011.

\bibitem[Dasgupta and Gupta(2003)]{dasgupta2003elementary}
S.~Dasgupta and A.~Gupta.
\newblock An elementary proof of a theorem of johnson and lindenstrauss.
\newblock \emph{Random Structures \& Algorithms}, 22\penalty0 (1):\penalty0
  60--65, 2003.

\bibitem[Davis and Kahan(1970)]{davis1970rotation}
C.~Davis and W.~M. Kahan.
\newblock The rotation of eigenvectors by a perturbation. iii.
\newblock \emph{SIAM Journal on Numerical Analysis}, 7\penalty0 (1):\penalty0
  1--46, 1970.

\bibitem[{van de Geer} and B\"uhlmann(2011)]{BuhlmannGeer2009}
S.~{van de Geer} and P.~B\"uhlmann.
\newblock {On the conditions used to prove oracle results for the lasso}.
\newblock \emph{Electronic Journal of Statistics}, 3:\penalty0 1360--1392,
  2011.

\bibitem[{Vershynin}(2010)]{Vershynin2010}
R.~{Vershynin}.
\newblock {Introduction to the non-asymptotic analysis of random matrices}.
\newblock \emph{arXiv preprint arXiv:1011.3027}, 2010.


\bibitem[von Rosen(1988)]{Dietrich_von_Rosen1988-nc}
D.~von Rosen.
\newblock Moments for the inverted wishart distribution.
\newblock \emph{Scand. Stat. Theory Appl.}, 15\penalty0 (2):\penalty0 97--109,
  1988.

\bibitem[Wainwright(2019)]{wainwright_2019}
M.~J. Wainwright.
\newblock \emph{High-Dimensional Statistics: A Non-Asymptotic Viewpoint}.
\newblock Cambridge Series in Statistical and Probabilistic Mathematics.
  Cambridge University Press, 2019.
\end{thebibliography}

\end{document}